\documentclass{article}
\usepackage{amsmath,inputenc}

\title{Credit-Based Congestion Pricing: Equilibrium Properties and Optimal Scheme Design}
\author{Devansh Jalota$^\dagger$, Jessica Lazarus$^\ddagger$, Alexandre Bayen$^\ddagger$, and Marco Pavone$^\dagger$%
\thanks{$^\dagger$ Stanford University; {\tt \{djalota, pavone\}@stanford.edu}.}%
\thanks{$^\ddagger$ University of California Berkeley; {\tt jlaz@berkeley.edu, bayen@berkeley.edu}.}%
}

\newif\ifarxiv   
\arxivtrue 

\ifarxiv
\else
\fi

\date{October 2022}

\usepackage{graphicx}
\usepackage{csvsimple}
\usepackage{rotating}
\usepackage[letterpaper, portrait, margin=1in]{geometry}
\usepackage{hyperref}
\hypersetup{
    colorlinks=true,
    linkcolor=blue,
    citecolor=red,
    filecolor=magenta,      
    urlcolor=cyan
    }

\usepackage[nocomma]{optidef}
\usepackage{listings}
\usepackage{comment}
\usepackage{appendix}
\usepackage[ruled,vlined]{algorithm2e}
\usepackage{amsfonts} 
\usepackage{amssymb}
\usepackage{bm}
\usepackage[square,comma,numbers]{natbib}

\usepackage{booktabs}

\usepackage{tikz}
\usepackage{pgfplots}
\DeclareUnicodeCharacter{2212}{−}
\usepgfplotslibrary{groupplots,dateplot}
\usetikzlibrary{patterns,shapes.arrows}
\pgfplotsset{compat=newest}

\usepackage{caption}
\usepackage{subcaption}

\usepackage{color} 
\definecolor{mygreen}{RGB}{28,172,0} 
\definecolor{mylilas}{RGB}{170,55,241}
\DeclareFixedFont{\ttb}{T1}{txtt}{bx}{n}{12} 
\DeclareFixedFont{\ttm}{T1}{txtt}{m}{n}{12}  
\usepackage{bbm}
\usepackage{amsmath,amsthm}

\usepackage{enumitem}

\newtheorem{theorem}{Theorem}
\newtheorem{corollary}{Corollary}
\newtheorem{proposition}{Proposition}

\newtheorem{lemma}{Lemma}
\newtheorem{observation}{Observation}

\theoremstyle{definition}
\newtheorem{definition}{Definition}
\newtheorem{remark}{Remark}
\newtheorem{example}{Example}

\newenvironment{hproof}{%
  \renewcommand{\proofname}{Proof (Sketch)}\proof}{\endproof}

\usepackage{color}
\definecolor{deepblue}{rgb}{0,0,0.5}
\definecolor{deepred}{rgb}{0.6,0,0}
\definecolor{deepgreen}{rgb}{0,0.5,0}

\usepackage{listings}

\renewcommand\footnotemark{}

\DeclareMathOperator*{\argmin}{arg\,min}
\newcommand{\x}{\mathbf{x}}
\newcommand{\y}{\mathbf{y}}
\newcommand{\z}{\mathbf{z}}
\newcommand{\0}{\mathbf{0}}

\newcommand{\F}{\mathcal{F}}
\newcommand{\E}{\mathcal{E}}
\newcommand{\A}{\mathcal{A}}
\newcommand{\B}{\mathcal{B}}
\newcommand{\C}{\mathcal{C}}
\newcommand{\Y}{\mathcal{Y}}
\newcommand{\G}{\mathcal{G}}
\newcommand{\I}{\mathcal{I}}
\newcommand{\Ll}{\mathcal{L}}
\newcommand{\llambda}{\boldsymbol{\lambda}}

\newcommand{\ttau}{\boldsymbol{\tau}}

\begin{document}

\maketitle

\begin{abstract}
    Credit-based congestion pricing (CBCP) has emerged as a mechanism to alleviate the social inequity concerns of road congestion pricing - a promising strategy for traffic congestion mitigation - by providing low-income users with travel credits to offset some of their toll payments. While CBCP offers immense potential for addressing inequity issues that hamper the practical viability of congestion pricing, the deployment of CBCP in practice is nascent, and the potential efficacy and optimal design of CBCP schemes have yet to be formalized. In this work, we study the design of CBCP schemes to achieve particular societal objectives and investigate their influence on traffic patterns when routing heterogeneous users with different values of time (VoTs) in a multi-lane highway with an express lane. We introduce a new non-atomic congestion game model of a \emph{mixed-economy}, wherein eligible users receive travel credits while the remaining ineligible users pay \emph{out-of-pocket} to use the express lane. In this setting, we investigate the effect of CBCP schemes on traffic patterns by characterizing the properties (i.e., existence, comparative statics) of the corresponding Nash equilibria and, in the setting when eligible users have time-invariant VoTs, develop a convex program to compute these equilibria. We further present a bi-level optimization framework to design optimal CBCP schemes to achieve a central planner’s societal objectives. Finally, we conduct numerical experiments based on a case study of the San Mateo 101 Express Lanes Project, one of the first North American CBCP pilots. Our results demonstrate the potential of CBCP to enable low-income travelers to avail of the travel time savings provided by congestion pricing on express lanes while having comparatively low impacts on the travel costs of other road users.
\end{abstract}

\section{Introduction}

With the ever-worsening traffic congestion in urban metropolises, congestion pricing has emerged as one of the most promising traffic management policies across the world~\cite{CP-SF,BOGOTA}. The premise behind congestion pricing is to charge users for the negative externality they impose on others to reduce system inefficiencies caused by selfish travel behavior~\cite{self-routing-PoA,how-bad-is-selfish}. \ifarxiv In effect, congestion pricing serves to steer the user equilibrium traffic pattern towards the system optimum traffic pattern~\cite{Sheffi1985,heterogeneous-pricing-roughgarden,multicommodity-extension}. \fi While network-wide deployments of congestion pricing, wherein tolls are placed on all or some cordoned portion of roads in the \ifarxiv road \fi network, are less common at present, there has been a growing interest in introducing congestion fees on certain lanes on highways, known as \emph{express lanes} \ifarxiv or managed lanes\fi, to provide users with a faster and more reliable travel option during peak traffic periods. \ifarxiv In the Bay Area alone, there are more than 155 miles of express lanes at the time of writing this manuscript~\cite{bay-area-express-lane}, with varying tolls depending on the region and the length of the highway section. In many cases, existing \emph{high-occupancy vehicle} (HOV) lanes, which only grant access to vehicles with more than two or three passengers, have been converted to \emph{high-occupancy toll (HOT)} lanes that enable single-occupant vehicles (SOVs) to pay for access. \fi 

Despite the proliferation of express lanes to better manage highway traffic, congestion fees on express lanes, as with network-wide congestion pricing \ifarxiv schemes\fi, have come under scrutiny due to social inequity concerns. In particular, express lanes have been termed as elitist ``Lexus lanes'', as they offer only those with the highest willingness to pay (i.e., the most wealthy) a higher quality of service through reduced travel times~\cite{Regulation,patterson2008lexus,HALL2018113} while lower-income users bear the brunt of longer travel times on \ifarxiv consequently \fi more congested general purpose (GP) lanes. More generally, there is a widespread belief that congestion pricing amounts to a ``tax on the working class~\cite{nyt-cp},'' given its regressive nature~\cite{ELIASSON2006602,Levinson-equityeffects,gemici_et_al:LIPIcs:2019:10270}. \ifarxiv Thus, there has been a growing interest in designing equitable congestion pricing mechanisms~\cite{WU20121273}, with a focus on designing \emph{credit-based congestion pricing} (CBCP) schemes~\cite{KOCKELMAN2005671}. In CBCP schemes, road users, particularly those with lower incomes, are given travel credits to use priced roads such as express lanes to counter the adverse equity impacts of congestion pricing. \fi Although numerous variations of CBCP have been explored \ifarxiv in the literature\fi, demonstrating their potential to provide positive equity benefits, their deployment in practice is nascent, with one of the first North American pilots launching in San Mateo County, California, in 2022. \ifarxiv  In particular, the \emph{San Mateo 101 Express Lanes Project}, which is converting 22 miles of HOV lanes to HOT express lanes, recently launched the ``Community Transportation Benefits Program'' to improve equitable access to the new facility~\cite{CBCP-SanMateo}. The equity program provides low-income residents with a \$100 credit for the express lane, enabling them to use the express lane for free, which they may not have been able to avail of otherwise, and thus reduce their travel times. \fi

While CBCP programs such as the one deployed in San Mateo County offer great potential to improve equity outcomes, a principled design of CBCP schemes is necessary to realize the benefits of its implementation. To this end, we study CBCP schemes to route heterogeneous users with different values of time (VoTs) on a multi-lane highway \ifarxiv segment \fi with express lanes. In particular, we introduce a \emph{mixed-economy} model, wherein eligible users receive travel credits to use the express lane while ineligible users have quasi-linear costs as they pay \emph{out-of-pocket} to access the express lane. We note that our mixed-economy model is unlike traditional \emph{single-economy} settings in traffic routing, wherein either all users have quasi-linear costs~\cite{YANG20041} or budget constraints~\cite{ezzat-karma} as in markets with artificial currencies. In this mixed-economy setting, we study the influence of CBCP schemes on \ifarxiv the resulting \fi traffic patterns by characterizing the properties of the induced equilibria and present a bi-level optimization framework to design optimal CBCP schemes. \ifarxiv Furthermore, we demonstrate the real-world applicability of our proposed model with numerical experiments based on a case study of the San Mateo 101 Express Lanes Project. \fi

\ifarxiv
\subsection{Our Contributions}
\fi
We study CBCP schemes in non-atomic congestion games to route heterogeneous users with varying VoTs on a multi-lane highway segment with a tolled express lane. In alignment with practically deployed CBCP schemes, such as in the San Mateo 101 Express Lanes Project, we introduce a new model of a mixed-economy wherein eligible users receive travel credits while ineligible users pay out-of-pocket to use the express lane. \ifarxiv Given the different optimization objectives of eligible and ineligible users, we introduce a new Nash equilibrium notion, which we term CBCP equilibria, to investigate the influence of a CBCP scheme on the resulting traffic patterns. \fi

\ifarxiv
In this mixed economy setting, we begin by investigating the properties of CBCP equilibria. In particular, \fi we establish the existence of CBCP equilibria, prove the uniqueness of the resulting edge flows, and, in the setting when eligible users have time-invariant VoTs, develop a convex program to compute CBCP equilibria. We further perform a comparative statics analysis to investigate the changes in the equilibria induced by CBCP schemes when the express lane tolls or eligible user budgets are increased or decreased. \ifarxiv In particular, we establish several natural monotonicity relations between the aggregate edge flows and the express lane tolls and budgets of eligible users that align with economic intuition. However, we do mention that we also obtain some counter-intuitive results, such as the violation of a natural substitutes condition (see Section~\ref{subsec:hardness}), due to the introduction of travel credits for eligible users. \fi

\ifarxiv

We then study the design of optimal CBCP schemes to achieve specific societal objectives of a central planner in the setting when eligible users have time-invariant VoTs. In this context, we develop a bi-level optimization framework for designing optimal CBCP schemes and present a \emph{dense sampling} approach for computing an approximation to the optimal CBCP scheme. Our dense sampling approach involves discretizing the set of feasible CBCP schemes and choosing the scheme that induces an equilibrium with the minimum societal cost. \ifarxiv We further motivate the efficacy of our dense sampling approach by establishing continuity relations between the resulting CBCP equilibria (i.e., the aggregate edge flows in the network) and the corresponding toll and budget parameters that characterize a CBCP scheme. In particular, our obtained continuity relations imply that, for a small enough discretization, the CBCP scheme output by dense sampling is a good approximation to the optimal CBCP scheme, as small changes in the toll and budget parameters only result in small changes in the societal cost function. \fi

\fi

\ifarxiv
Finally, we present numerical experiments to investigate the influence of CBCP schemes on traffic patterns and study their optimal design through a case study of the San Mateo 101 Express Lanes Project. \fi Our numerical results validate our comparative statics analysis results and indicate that the magnitude of the changes in the travel time and proportion of users on the express lane with the tolls and budgets observed in the experiments are reflective of real-world multi-lane highways with express lanes. Furthermore, since the optimal CBCP scheme can vary widely based on the central planner’s objective, our results show that a principled approach using bi-level optimization proposed in this work is key to realizing the benefits of CBCP schemes. We also discuss the policy implications of this work, noting avenues for future work and considerations for further mitigating the inequity concerns of congestion pricing.

\ifarxiv

\paragraph{Organization:} This paper is organized as follows. Section~\ref{sec:relatedWork} reviews related literature. We then present a model of traffic flow and introduce a new Nash equilibrium notion, the CBCP equilibrium, studied in this work in Section~\ref{sec:model}. Next, we investigate the existence of Nash equilibria induced by CBCP schemes in Section~\ref{sec:existence} and perform a comparative statics analysis of the CBCP equilibria in Section~\ref{sec:CompStatics}. Then, we introduce a bi-level optimization framework to design optimal CBCP schemes in Section~\ref{sec:CBCPDesign} and present numerical experiments based on a real-world case study of the San Mateo 101 Express Lanes Project in Section~\ref{sec:experiments}. Finally, we conclude the paper and provide directions for future work in Section~\ref{sec:conclusion}.

\fi

\section{Related Literature} \label{sec:relatedWork}

\ifarxiv

While achieving system efficiency in resource allocation is often the holy grail for central planners, there are many settings when a range of measures of user welfare beyond system efficiency, e.g., fairness or equity considerations, need to be taken into account to design practically deployable allocation mechanisms. For instance, in a seminal work,~\cite{weitzman-seminal} showed that there are many settings, e.g., cancer treatment, when it is undesirable to efficiently allocate resources to users with the highest willingness to pay. In a similar spirit,~\cite{CONDORELLI2013582} studied the allocation of identical objects to users and showed that the choice between market and non-market allocation mechanisms crucially relies on the similarity between a central planner's objective and users' willingness to pay. Furthermore,~\cite{RAM-Akbarpour} developed a resource allocation model to allocate heterogeneous objects to a continuum of agents to trade off allocative efficiency with redistributive considerations. More generally, this trade-off between system efficiency and other welfare measures, such as fairness or equity, has been investigated in a range of resource allocation applications, including vaccine allocation~\cite{Vaccine-abdk}, aircraft scheduling~\cite{bertsimas-PoF}, and traffic routing~\cite{so-routing-seminal,jalota-balancing}.

The consideration of welfare measures beyond system efficiency is, in particular, important in the context of congestion pricing, as, despite several successful real-world deployments~\cite{gomez1994road,litman2003london,langmyhr1999understanding}, its practical adoption has been quite limited~\cite{ELIASSON2006602,Levinson-equityeffects,CP-Low-acceptance}, in part due to social inequity concerns. Thus, there has been a growing interest in designing more equitable congestion pricing schemes~\cite{WU20121273}, with a particular focus on revenue redistribution approaches~\cite{small1992using}, wherein the collected toll revenues are refunded as lump-sum transfers to users. In particular, several works~\cite{arnott1994,DAGANZO1995139,ADLER2001447,GUO2010972} consider the problem of refunding the collected toll revenues to users in problem settings that range from Vickrey's bottleneck congestion model~\cite{vickrey1969congestion} - a benchmark representation of peak-period traffic congestion on a single lane - to general road networks with a multiple origin-desitination travel demand and heterogeneous users. Building on these works,~\cite{jalota-acm-eaamo} show that, through appropriate revenue refunding, a central planner can achieve both the efficiency and equity objectives of sustainable transportation. As in these works, we also consider refunding a proportion of the collected revenues to certain groups of users. However, compared to these works wherein the refund serves as additional money ``in the pocket'' of users, in our setting, we consider credit-based congestion pricing schemes, wherein the credits provided to users function solely for use of the express lane (e.g., an express lane travel allowance).

\fi


\ifarxiv

CBCP schemes have been explored extensively in the literature with a focus on \emph{tradable} credit schemes wherein additional credits may be purchased~\cite{LiCBCPAustin,YangCBCP2011,WU20121273} or earned by exhibiting desirable travel behavior (e.g., choosing not to travel, travelling at an off-peak period, or travelling by a more sustainable mode of transportation)~\cite{XiaoTradableCBCP,NieCBCP2015}. Moreover, these credits typically have monetary value beyond paying for access to priced facilities (e.g., for using public transit) and have the potential to improve the equitable distribution of benefits from congestion pricing~\cite{WU20121273,KOCKELMAN2005671}. However, the equity improvements from tradable CBCP schemes typically result from lower-income users that reduce travel to sell their allocated credits or use them on driving alternatives~\cite{NieCBCP2015,KOCKELMAN2005671}. In contrast to these works, we consider non-tradable CBCP schemes in which eligible users receive a fixed budget of credits to access the express lane as an alternative to driving on slower GP lanes. In particular, in the setting studied in this work, credits cannot be traded nor provide value for anything other than for paying express lane tolls, which mirrors the operation of the San Mateo 101 Express Lanes Project. Furthermore, under non-tradable CBCP schemes, eligible users will avail of a fast and reliable mode of travel, i.e., the express lane, which is in contrast to tradable CBCP schemes, wherein low-income users generally use less convenient modes of travel to obtain monetary gains from selling excess credits.

\fi


\ifarxiv


Since travel credits provide no value to users beyond paying for express lane tolls, our work is also closely related to the design of artificial currency mechanisms~\cite{HZ,kash2007optimizing,gorokh2020nonmonetary}, which have found many applications~\cite{prendergast-2017,Budish,ezzat-karma}. However, in contrast to these applications that consider a single-economy, wherein artificial currencies serve as the only medium of transfer to avail resources, we consider a mixed-economy wherein only a fraction of the users receive travel credits (artificial currencies) to use the express lane.

\fi

From a methodological viewpoint, as in prior studies that characterize user equilibria \ifarxiv with heterogeneous users \fi in congestion games~\cite{Dafermos-multiclass,YANG20041} \ifarxiv and study equilibrium existence under alternative pricing schemes, such as tradable credits~\cite{yang-tradable-credits,WANG2012426}, \fi we also investigate the properties of equilibria induced by CBCP schemes. However, \ifarxiv we note that \fi CBCP equilibria differ markedly from prior equilibrium notions in the congestion pricing literature (see Section~\ref{subsec:nashEq}). Beyond \ifarxiv an equilibrium characterization of CBCP schemes, \fi we also develop a bi-level \ifarxiv optimization \fi framework to optimize over CBCP schemes \ifarxiv to achieve particular societal objectives\fi. While bi-level optimization~\cite{colson2007overview} has been \ifarxiv extensively studied in \fi traffic routing contexts, \ifarxiv e.g., second-best tolling~\cite{VERHOEF2002281,bilevel-labbe,Larsson1998,bilevel-patricksson} and revenue maximization~\cite{bilevel-tolls},\fi our bi-level framework involves optimizing over both tolls and budgets rather than only road tolls.

\section{Model} \label{sec:model}

\ifarxiv
In this section, we introduce the basic definitions of traffic flow for a multi-lane freeway (Section~\ref{subsec:preliminaries}), the operation of CBCP schemes and corresponding user costs (Section~\ref{subsec:cbcpDef}), and the Nash equilibrium notion we study in this work (Section~\ref{subsec:nashEq}).
\fi

\subsection{Preliminaries} \label{subsec:preliminaries}

We study the problem of designing CBCP schemes to route heterogeneous users with different VoTs in a multi-lane highway section. The highway consists of one express lane that can be tolled while the remaining general purpose (GP) lanes must remain untolled. Without loss of generality, we model the freeway section as a two-edge Pigou network consisting of a source vertex $s$, a destination vertex $d$, and two directed edges $e \in \{1, 2 \}$ between the source and destination vertices, where the first edge ($e = 1$) denotes the express lane while the second edge ($e = 2$) corresponds to the GP lanes \ifarxiv without tolls\fi. We note that modeling all GP lanes as a single edge in the traffic network is without loss of generality, as we focus on equilibrium formation in this work; hence, these lanes are indistinguishable from the lens of a user as none of these lanes have any tolls. Furthermore, to model the travel times on each edge $e$ \ifarxiv in this traffic network\fi, we consider a flow-dependent travel-time (latency) function $l_e : \mathbb{R}_{\geq 0} \rightarrow \mathbb{R}_{\geq 0}$, which maps $x_e$, the traffic flow rate on edge $e$, to the travel time $l_e(x_e)$. As in prior literature on traffic routing, we assume that the function $l_e$, for both edges $e$, is differentiable, convex, and monotonically increasing.

Users make trips in the \ifarxiv transportation \fi network over $T$ periods (e.g., days) over which the CBCP scheme is run (see Section~\ref{subsec:cbcpDef}) and belong to a finite set of discrete groups characterized by their (i) level of income and (ii) value of time. Let $\G$ denote the set of all user groups, wherein users, based on their income, are subdivided into two categories, \emph{eligible} and \emph{ineligible}, depending on their eligibility to receive travel credit (budget) to use the express lane as pre-determined by a central planner. We let $\G_{E}$ and $\G_{I}$ denote the sets of eligible and ineligible user groups, respectively. Furthermore, a user in a group $g \in \G$ at each period $t \in [T]$ has a VoT $v_{t,g}$, which captures users' willingness to pay for travel time savings, i.e., their trade-off between travel time and money. The total travel demand of a user group $g$ is given by $d_g$, which represents the flow of users in group $g$ to be routed \ifarxiv in the network \fi at each period. We assume that the travel demand for each user group stays fixed across time, as is consistent with weekday rush hour traffic wherein users commute to and from work. For the simplicity of exposition, we normalize $d_g$ to one for all groups $g$ and note that our results naturally extend to the general travel demands setting.

A flow pattern $\y = \{ y_{e,t}^g: e \in \{1, 2 \}, t \in [T], g \in \G \}$ specifies for each user group $g$ the amount of flow $y_{e,t}^g \geq 0$ routed on edge $e$ at period $t$. The resulting flows must satisfy the user demand at each period $t$, i.e.,
\ifarxiv
\begin{align*}
    \sum_{e = 1}^2 y_{e,t}^g = 1, \quad \text{for all } g \in \G, t \in [T].
\end{align*}
\fi
Furthermore, we represent the edge flows corresponding to the flow pattern $\y$ by the vector $\x = \{ x_{e,t}: e \in \{1, 2 \}, t \in [T] \}$, where 
\ifarxiv
\begin{align*}
    x_{e,t} = \sum_{g \in \G} y_{e,t}^g, \quad \text{for all } e \in \{1, 2 \}, t \in [T].
\end{align*}
\fi

\ifarxiv
\subsection{Credit-based Congestion Pricing Schemes and User Optimization} \label{subsec:cbcpDef}
\fi

A \ifarxiv credit-based congestion pricing (CBCP) \fi scheme is characterized by a tuple $(\ttau, B)$, where $\ttau \in \mathbb{R}_{\geq 0}^T$ represents the vector of tolls on the express lane over $T$ periods, and $B$ is the travel credit (budget) given to eligible users to use the express lane over the $T$ periods. Both eligible and ineligible users pay the corresponding toll when using the express lane; however, ineligible users pay for the tolls out-of-pocket while eligible users pay \ifarxiv for the tolls \fi using their available budget. For simplicity, we assume that eligible users never spend out-of-pocket to use the express lane and thus use it only if they have sufficient credits. Such an assumption is consistent with \ifarxiv congestion pricing schemes in \fi real-world traffic networks, wherein lower-income users, i.e., those \ifarxiv that constitute \fi the eligible group, are less likely to spend money out-of-pocket to use tolled roads. However, \ifarxiv we note that \fi this modeling assumption can readily be relaxed to the setting where eligible users can also spend out-of-pocket, and we defer a thorough treatment of this setting to future research. We also note that the finiteness of the time horizon $T$ is crucial to successfully deploying a CBCP scheme to improve express lane access for eligible users as the ratio $\frac{B}{T}$ represents the per-period budget for these users to use the express lane. Finally, \ifarxiv we mention that, \fi as in the traffic routing literature~\cite{GUO2010972,KOCKELMAN2005671}, we focus on single-occupancy vehicles in this work and defer the consideration of high-occupancy alternatives (e.g., carpools) to future work.

\ifarxiv

In response to a given CBCP scheme $(\ttau, B)$, users are assumed to be selfish and thus choose whether to use the express lane over the $T$ periods to minimize their cumulative travel cost. We now present the individual user optimization problems for both the ineligible and eligible users and note that these differ from each other as we are in a mixed economy setting, wherein ineligible users pay for the express lane tolls out-of-pocket while eligible users use travel credit to use the express lane.

\fi

\ifarxiv
\paragraph{Ineligible Users:}
\fi
Since ineligible users spend out-of-pocket to use the express lane, the cumulative travel cost for ineligible users is assumed to be a linear function of their travel time and tolls, \ifarxiv which is \fi a commonly used modeling assumption~\cite{heterogeneous-pricing-roughgarden,multicommodity-extension}. In particular, given a CBCP scheme $(\ttau, B)$ and a vector of edge flows $\x$, the individual optimization of an ineligible user in a group $g \in \G_I$ is \ifarxiv given by \fi
\begin{mini!}|s|[2]                   
    {\z^g \in \mathbb{R}_{\geq 0}^{2 \times T}}                               
    {\! \sum_{t \in [T]} \!  \sum_{e \in [2]} \! \left( v_{t,g} l_e(x_{e,t}) \! + \! \mathbbm{1}_{e = 1} \tau_t \right) \! z_{e,t}^g,  \label{eq:objIOPIn}}   
    {\label{eq:Eg002In}}             
    {\mu^{g*}(\x, \ttau, B) \! = \!}                                
    \addConstraint{z_{1,t}^g + z_{2,t}^g}{ = 1, \quad \forall t \in [T] \label{eq:con2IOpIn}}    
\end{mini!}
where~\eqref{eq:objIOPIn} is the travel cost objective of the ineligible users 
and~\eqref{eq:con2IOpIn} are user allocation constraints at each period $t \in [T]$. Observe that Problem~\eqref{eq:objIOPIn}-\eqref{eq:con2IOpIn} is a linear program as the decision variables are given by \ifarxiv the vector \fi $\z^g = \{ z_{e,t}^g: e \in \{1, 2 \}, t \in [T] \}$, which correspond to the actions of an infinitesimal user and thus does not influence the edge flow \ifarxiv vector \fi $\x$. Here, we denote the decision variables for any user in group $g \in \G_I$ by the vector $\z^g$ to distinguish it from the cumulative flow of all users in $g \in \G_I$ given by $\y^g = \{ y_{e,t}^g: e \in \{1, 2 \}, t \in [T] \}$ and note that the decision variables $z_{e,t}^g \in [0, 1]$ can be interpreted as the probability that a user in group $g$ uses edge $e$ at period $t$ (or the fraction of users in group $g$ on edge $e$ at period $t$). Further, for succinctness, \ifarxiv for any feasible $\z^g$, \fi we denote the travel cost for users in a group $g \in \G_I$ as $\mu_{\z^g}(\x, \ttau, B) = \sum_{t = 1}^T  \sum_{e = 1}^2 \left( v_{t,g} l_e(x_{e,t})  + \mathbbm{1}_{e = 1} \tau_t \right) z_{e,t}^g$, the travel cost on edge $e$ at period $t$ for \ifarxiv ineligible users in a group $g \in \G_I$ \fi as $\mu_{e,t}^g(\x, \ttau, B) = v_{t,g} l_e(x_{e,t})  + \mathbbm{1}_{e = 1} \tau_t$, and let $\mu^{g*}(\x, \ttau, B)$ denote the minimum travel cost. 

\emph{Optimal Solution of Problem~\eqref{eq:objIOPIn}-\eqref{eq:con2IOpIn}:} By the separability of the travel cost function and the constraints across periods, the optimal solution of Problem~\eqref{eq:objIOPIn}-\eqref{eq:con2IOpIn} corresponds to travel decisions that are independent across periods. In particular, ineligible users choose the edge with the minimum travel cost at each period, i.e., at each period $t$ ineligible users choose the edge $e \in \argmin_{e \in \{1, 2 \}} \{ v_{t,g} l_e(x_{e,t}) + \mathbbm{1}_{e = 1} \tau_t\}$. 
Observe that the travel behavior of ineligible users at any given period is akin to the well-studied model of heterogeneous users in non-atomic congestion games~\cite{heterogeneous-pricing-roughgarden,multicommodity-extension,YANG20041}, wherein users choose routes with the minimum travel cost.

\ifarxiv
\paragraph{Eligible Users:}

\fi
On the other hand, since eligible users only utilize travel credit to use the express lane \ifarxiv (and do not spend any money out-of-pocket)\fi, the cumulative travel cost for a user belonging to a group $g \in \G_{E}$ only consists of the travel time component of the cost of the ineligible users. As a result, given a CBCP scheme $(\ttau, B)$ and a vector of edge flows $\x$, the individual optimization of an eligible user in a group $g \in \G_E$ is \ifarxiv given by \fi
\begin{mini!}|s|[2]                   
    {\z^g \in \mathbb{R}_{\geq 0}^{2 \times T}}                               
    {\sum_{t \in [T]} \sum_{e \in [2]} v_{t,g} l_e(x_{e,t}) z_{e,t}^g ,  \label{eq:objIOP}}   
    {\label{eq:Eg002}}             
    {\mu^{g*}(\x, \ttau, B) = }                                
    \addConstraint{z_{1,t}^g + z_{2,t}^g}{ = 1, \forall t \in [T] \label{eq:con2IOp}}    
    \addConstraint{\sum_{t \in [T]} z_{e,t}^g \tau_t}{ \leq B, \label{eq:con3IOp}}
\end{mini!}
where
~\eqref{eq:con2IOp} are user allocation constraints at each period $t \in [T]$, and~\eqref{eq:con3IOp} is the budget constraint, which ensures that no user spends more travel credits than the provided allowance. \ifarxiv As with the ineligible users, \fi Problem~\eqref{eq:objIOP}-\eqref{eq:con3IOp} is a linear program and we denote the travel cost for users in a group $g \in \G_E$ as $\mu_{\z^g}(\x, \ttau, B) = \sum_{t = 1}^T \sum_{e = 1}^2 v_{t,g} l_e(x_{e,t}) z_{e,t}^g$ \ifarxiv for any feasible $\z^g = \{ z_{e,t}^g: e \in \{1, 2 \}, t \in [T] \}$ and let $\mu^{g*}(\x, \ttau, B)$ denote the minimum travel cost. \fi  

\emph{Optimal Solution of Problem~\eqref{eq:objIOP}-\eqref{eq:con3IOp}:} To characterize the optimal solution of Problem~\eqref{eq:objIOP}-\eqref{eq:con3IOp}, we first note that due to the budget Constraint~\eqref{eq:con3IOp}, the travel decisions of eligible users are coupled across the periods unlike the ineligible users. In particular, the structure of the optimal solution of Problem~\eqref{eq:objIOP}-\eqref{eq:con3IOp} depends on the notion of a travel \emph{bang-per-buck} ratio, which we define below.

\begin{definition} [Travel Bang-Per-Buck Ratio] \label{def:travel-bpb}
Consider a CBCP scheme $(\ttau, B)$ with a corresponding edge flow vector $\x$. Then, the travel bang-per-buck ratio at any period $t$ for an eligible user in a group $g \in \G_E$ is given by $\frac{v_{t,g}(l_2(x_{2,t}) - l_1(x_{1,t}))}{\tau_t}$.
\end{definition}

In other words, the travel bang-per-buck ratio represents the ratio between the travel time savings for an eligible user when using the express lane at period $t$ to the corresponding toll at that period. \ifarxiv We now present a characterization of the structure of the optimal solution of Problem~\eqref{eq:objIOP}-\eqref{eq:con3IOp}. \fi

\ifarxiv
\begin{lemma} [Optimal Solution of Individual Optimization Problem of Eligible Users] \label{lem:optEligible}
Consider a CBCP scheme $(\ttau, B)$ and an edge flow $\x$. Then, for any eligible user in a group $g \in \G_E$ the optimal solution to Problem~\eqref{eq:objIOP}-\eqref{eq:con3IOp} is such that they spend their budget $B$ in using the express lane at different periods in the descending order of the travel bang-per-buck ratios.
\end{lemma}
\fi

\ifarxiv
Lemma~\ref{lem:optEligible} establishes that eligible users use the express lane at different periods in the descending order of the travel bang-per-buck ratios until their budget is exhausted. This result on the optimal solution to Problem~\eqref{eq:objIOP}-\eqref{eq:con3IOp} is akin to prior characterizations of the optimal solution to users' individual optimization problems in matching markets~\cite{eq-matching-markets} as well as Fisher markets with additional knapsack constraints~\cite{jalota2021fisher}. We refer to Appendix~\ref{apdx:OptEligiblePf} for a proof of Lemma~\ref{lem:optEligible} and note that this result follows directly through an investigation of the first order necessary and sufficient KKT conditions of Problem~\eqref{eq:objIOP}-\eqref{eq:con3IOp}.
\fi

\subsection{CBCP Equilibria} \label{subsec:nashEq}

We evaluate the efficacy of a CBCP scheme in achieving particular societal scale goals of a central planner based on the induced Nash equilibria (see Section~\ref{sec:CBCPDesign}). To this end, \ifarxiv in this section, \fi we present the Nash equilibrium notion, which we term a CBCP equilibrium, studied in this work. In particular, given a CBCP scheme $(\ttau, B)$, a flow pattern $\y$ is a CBCP equilibrium if no user can reduce their travel cost through a unilateral deviation, as is formalized by the following definition.

\begin{definition} [CBCP Equilibrium] \label{def:MainNashEq}
For a given CBCP scheme $(\ttau, B)$, the flow $\y$, with corresponding edge flows $\x$, is a CBCP $(\ttau, B)$-equilibrium if for each ineligible group $g \in \G_I$ with $y_{e,t}^g>0$,
\ifarxiv
\begin{align*}
    \mu_{e,t}^g(\x, \ttau, B) \leq \mu_{e',t}^g(\x, \ttau, B), \quad \text{for all } e' \in \{1, 2 \}, t \in [T],
\end{align*}
\fi
and for each eligible group $g \in \G_E$ with $y_{e,t}^g > 0$, it holds that $z_{e,t}^{g*}>0$ 
for some optimal solution $\z^*$ to Problem~\eqref{eq:objIOP}-\eqref{eq:con3IOp}, i.e., 
\ifarxiv
\begin{align*}
    \mu_{\z^*}^{g}(\x, \ttau, B) \leq \mu_{\z}^{g}(\x, \ttau, B), \quad \text{for all } \z \geq \0 \text{ satisfying Constraints~\eqref{eq:con2IOp}-\eqref{eq:con3IOp}}.
\end{align*}
\fi
\end{definition}
\ifarxiv
A few comments about this equilibrium notion are in order. First, recall that the equilibrium condition of the ineligible users in Definition~\ref{def:MainNashEq} corresponds to the optimal solution of Problem~\eqref{eq:objIOPIn}-\eqref{eq:con2IOpIn}. Thus, if all users were ineligible, CBCP equilibria reduce to standard non-atomic Nash equilibria with heterogeneous users~\cite{heterogeneous-pricing-roughgarden,multicommodity-extension,YANG20041}, as their travel decisions are independent across periods. Next, the equilibrium notion in Definition~\ref{def:MainNashEq} additionally accounts for the preferences of eligible users whose travel decisions are coupled across the periods through their budget constraints. Thus, CBCP equilibria, which simultaneously capture the preferences of both ineligible and eligible users in a mixed-economy setting, differs from prior works that focus on the single-economy setting wherein all users have quasi-linear costs~\cite{heterogeneous-pricing-roughgarden,multicommodity-extension,YANG20041}. For a further discussion on the notion of CBCP equilibria, see Appendix~\ref{apdx:cbcp-eq-def}.

\fi

\section{Properties of CBCP Equilibria} \label{sec:existence}

\ifarxiv
The study of the properties of Nash equilibria, such as equilibrium existence and the uniqueness of edge flows, in traffic routing has been a focal point of transportation and game theory research, as the efficacy of any introduced policy, e.g., travel demand management strategies such as congestion pricing, is generally evaluated based on the corresponding Nash equilibria. In the context of congestion pricing, the Nash equilibrium characterizations have typically focused on the single-economy settings wherein all users have similar objective functions, i.e., all users have quasi-linear costs. For instance,~\cite{YANG20041} developed a convex program to establish the existence of Nash equilibria under any set of road tolls and show that the corresponding edge flows are unique for strictly convex travel time functions. In contrast to such classical single-economy settings, in this work, we study a mixed economy setting, wherein different groups of users have different objective functions, depending on whether they are eligible or ineligible. In particular, ineligible users have quasi-linear costs while eligible users solve a budget-constrained cost minimization Problem~\eqref{eq:objIOP}-\eqref{eq:con3IOp}.


As a result, in this section, we initiate our study of CBCP schemes in this mixed-economy setting by studying the properties of CBCP equilibria. In particular, we establish the existence of CBCP equilibria and show that, while the equilibria need not be unique, in general, the corresponding edge flow vector $\x$ is unique for any CBCP scheme (Section~\ref{subsec:eqExistence}). Then, in Section~\ref{subsec:ConvexProgram}, we establish that eligible users' values of time are time-invariant, i.e., their values of time are constant across the $T$ periods, then CBCP equilibria can be computed through the solution of a convex program. \ifarxiv \else We further investigate the properties of CBCP equilibria by performing a comparative statics analysis to characterize the changes in the equilibria induced by CBCP schemes when the tolls set on the express lane or budgets distributed to eligible users are increased or decreased. For conciseness, we present a summary of our comparative statics analysis results in Section~\ref{sec:CompStatics} and defer the in-depth analysis to Appendix~\ref{apdx:CompStatics}.
\fi

\fi

\subsection{Existence of CBCP Equilibria and Uniqueness of Edge Flows} \label{subsec:eqExistence}

\ifarxiv

In this section, we establish the existence of CBCP equilibria and the uniqueness of the corresponding edge flows\ifarxiv under any CBCP scheme $(\ttau, B)$\fi.


To this end, we first show that CBCP equilibria are guaranteed to exist in the general setting when heterogeneous eligible and ineligible users have time-varying VoTs, i.e., their VoTs can vary across the $T$ periods, as is elucidated through the following theorem.

\fi

\begin{theorem} [Existence of CBCP Equilibria] \label{thm:eqExistence}
For any CBCP scheme $(\ttau, B)$, where the tolls $\ttau \geq \0$ and budget $B \geq 0$, there exists a CBCP $(\ttau, B)$-equilibrium in the setting with heterogeneous eligible and ineligible users whose VoTs can vary over the $T$ periods.
\end{theorem}

\ifarxiv

Theorem~\ref{thm:eqExistence} establishes that introducing budgets for the eligible users does not preclude the existence of an equilibrium and augments the literature on investigating equilibrium existence under alternatives to standard congestion pricing~\cite{yang-tradable-credits,WANG2012426}. Furthermore, Theorem~\ref{thm:eqExistence} extends equilibrium existence results in the traffic routing literature focusing on the single-economy setting~\cite{YANG20041}, wherein all users have quasi-linear costs, to the mixed-economy setting \ifarxiv studied \fi in this work\ifarxiv with both eligible and ineligible users\fi. To prove Theorem~\ref{thm:eqExistence}, we leverage the following variational inequality characterization of CBCP equilibria.

\fi

\begin{lemma} [Variational Inequality Characterization of CBCP Equilibria] \label{lem:variationalinequality}
Consider a CBCP scheme $(\ttau, B)$. Then, a flow $\y^* = (y_{e,t}^{g*})$, with corresponding edge flows $\x^*$, is a CBCP $(\ttau, B)$-equilibrium if and only if it solves the following variational inequality problem:
\ifarxiv
\begin{align} \label{eq:variationalIneq}
    \sum_{t = 1}^T \sum_{e = 1}^2 \left[ \sum_{g \in \G_E} v_{t,g} l_e(x_{e,t}^*) (y_{e,t}^{g} - y_{e,t}^{g*}) + \sum_{g \in \G_I} (v_{t,g} l_e(x_{e,t}^{*}) + \mathbbm{1}_{e = 1} \tau_t ) (y_{e,t}^{g} - y_{e,t}^{g*}) \right] \geq 0, \quad \forall \text{ feasible } \y \in \Omega,
\end{align}
\fi
where the set $\Omega$ is described by $\y \geq \0$, $y_{1,t}^g + y_{2,t}^g = 1$ for all $t \in [T]$ and $g \in \G$, and all eligible users satisfy their budget Constraint~\eqref{eq:con3IOp}.
\end{lemma}

\vspace{-3pt}

The proof of Lemma~\ref{lem:variationalinequality} is presented in Appendix~\ref{apdx:mainVariationalIneqPf} and follows from an analogous argument to that in~\cite{NAGURNEY2000393}, as the variational Inequality~\eqref{eq:variationalIneq} is reminiscent of the finite dimensional variational inequalities used to study heterogeneous user equilibria in classical traffic routing settings~\cite{NAGURNEY2000393,YANG20041,WANG2012426}. However, in contrast to these approaches that consider a single-economy setting, Equation~\eqref{eq:variationalIneq} has two separate terms to capture the differing travel costs for the eligible and ineligible users in a mixed-economy setting.

\ifarxiv
We now use the variational inequality characterization in Lemma~\ref{lem:variationalinequality} to complete the proof of Theorem~\ref{thm:eqExistence}. To do so, we show that the variational Inequality~\eqref{eq:variationalIneq} admits a feasible solution, as is elucidated through the following lemma.
\fi


\vspace{-3pt}

\begin{lemma} [Feasibility of Variational Inequality] \label{lem:variationalIneq2}
There exists a solution $\y^*$ to the variational Inequality~\eqref{eq:variationalIneq}.
\end{lemma}

\vspace{-9pt}

\begin{proof}
First, note that the variational Inequality~\eqref{eq:variationalIneq} can be expressed in standard form $F(\y^*)^T (\y - \y^*)$ for all feasible $\y$, where $F = (F_{e,t}^g)_{e \in \{1, 2 \}, t \in [T], g \in \G}$, $F_{e,t}^g(\y^*) = v_{t,g} t_e(x_{e,t}^{*})$ for $g \in \G_E$ and $F_{e,t}^g(\y^*) = v_{t,g} t_e(x_{e,t}^{*}) + \mathbbm{1}_{e = 1} \tau_t$ for $g \in \G_I$. Then, following standard variational inequality theory~\cite{facchinei2003finite,variational-book-2000}, a feasible solution $\y^*$ to Equation~\eqref{eq:variationalIneq} exists as the feasible set $\Omega$, defined in Lemma~\ref{lem:variationalinequality}, is compact and the travel time functions are continuous. 
\end{proof}

\vspace{-2pt}

\ifarxiv
Lemmas~\ref{lem:variationalinequality} and~\ref{lem:variationalIneq2} jointly imply Theorem~\ref{thm:eqExistence}. While Theorem~\ref{thm:eqExistence} established the existence of CBCP equilibria, \ifarxiv we note that \fi CBCP equilibria need not, in general, be unique. For instance, there are settings when two different CBCP equilibria can be obtained by swapping the actions of two groups of users. Despite this non-uniqueness of CBCP equilibria, we now show that the corresponding aggregate equilibrium edge flows are unique for any CBCP scheme $(\ttau, B)$\ifarxiv, as is elucidated through the following lemma\fi.

\fi


\vspace{-2pt}

\begin{lemma} [Uniqueness of Edge Flows \ifarxiv Corresponding to CBCP Equilibria\fi] \label{lem:uniqueness-edge}
For any CBCP scheme $(\ttau, B)$, where the tolls $\ttau \geq \0$ and budget $B \geq 0$, the aggregate edge flow vector $\x$ corresponding to any CBCP $(\ttau, B)$-equilibrium is unique.
\end{lemma}

\vspace{-3pt}

\ifarxiv
For a proof of Lemma~\ref{lem:uniqueness-edge}, see Appendix~\ref{apdx:edge-uniquenessPf}. We note that the proof of Lemma~\ref{lem:uniqueness-edge} relies on the strict convexity of the edge travel time functions, assumed in Section~\ref{subsec:preliminaries}, which is standard in the traffic routing literature~\cite{YANG20041} and holds for commonly used travel time functions, such as the BPR function~\cite{utraffic}. Furthermore, this assumption on the strict convexity of the edge travel time functions is mild compared to other traffic routing settings that require stronger assumptions on the edge travel time functions, e.g., weighted average monotonicity~\cite{NAGURNEY2000393,WANG2012426}, to establish the uniqueness of edge flows.

\fi

\vspace{-3pt}

\ifarxiv
\subsection{Convex Program to Compute CBCP Equilibria} \label{subsec:ConvexProgram}
\fi

While Theorem~\ref{thm:eqExistence} established the existence of CBCP equilibria, determining a feasible solution to the variational Inequality~\eqref{eq:variationalIneq} may, in general, be challenging. \ifarxiv Several numerical methods to compute solutions to variational inequalities exist with provable convergence guarantees~\cite{NAGURNEY2000393}; however, these methods often take many steps to converge to a feasible solution. \fi Given the difficulty of computing solutions to variational inequalities~\cite{fabrikant2004complexity}, in this section, we present a convex program to compute CBCP equilibria in the setting when eligible users have time-invariant VoTs, i.e., for all \ifarxiv eligible user groups \fi $g \in \G_E$, $v_{t,g} = v_{t',g}>0$ for all $t,t' \in [T]$, while ineligible users can, in general, have time-varying VoTs. We defer the question of computing CBCP equilibria \ifarxiv in the setting \fi when eligible users' VoTs vary with time to future research and note that the time-invariance \ifarxiv assumption on \fi eligible users' VoTs has important practical significance. In particular, since the express lane is likely to be tolled during morning and evening rush hour periods on weekdays, the VoTs of users commuting to and from work are unlikely to differ much between one period and the next, e.g., between subsequent days. Furthermore, the individual optimization Problem~\eqref{eq:objIOP}-\eqref{eq:con3IOp} for the eligible users can involve quite sophisticated decision-making as eligible users' travel decisions are coupled across periods. Thus, given the lack of variability in users' VoTs between periods and since users may not have complete information on their values of time over the $T$ periods, eligible users may prefer to minimize their total travel time rather than the more complex Objective~\eqref{eq:objIOP}. 

\ifarxiv
We first present the convex program to compute CBCP equilibria in the setting when eligible users have VoTs that remain constant over the $T$ periods. In particular, given a CBCP scheme $(\ttau, B)$, we consider the following convex optimization problem
\fi
\begin{mini!}|s|[2]                   
    {\y \in \mathbb{R}_{\geq 0}^{2 \times T \times |\G|}}                              
    {\sum_{t \in [T]} \left[ \sum_{e \in [2]} \int_{0}^{x_{e,t}} l_e(\omega) d \omega + \sum_{g \in \G_I} \frac{y_{1,t}^g \tau_t}{v_{t,g}} \right], {\label{eq:obj}} }  
    {\label{eq:Eg001}}             
    {}                                
    \addConstraint{y_{1,t}^g + y_{2,t}^g}{= 1, \quad \forall t \in [T], g \in \G } {\label{eq:allocation}}
    \addConstraint{\sum_{t \in [T]} \tau_t y_{1,t}^g}{  \leq B, \quad \forall g \in \G_E , } {\label{eq:budget}}
    \addConstraint{\sum_{g \in \G} y_{e,t}^g}{ = x_{e,t}, \quad \forall e \in E , t \in [T],} {\label{eq:edgeConstraint}}
\end{mini!}
where~\eqref{eq:allocation} are allocation constraints,~\eqref{eq:budget} are eligible user budget constraints, and~\eqref{eq:edgeConstraint} are edge flow constraints. \ifarxiv For conciseness, we denote the set of feasible flows $\y$ satisfying the constraints of the above convex program as $\Omega$. \fi
\ifarxiv We reiterate here that even though we consider the setting when eligible users have time-invariant VoTs, the VoTs of the ineligible users can vary with time, as in Objective~\eqref{eq:obj}. \fi Problem~\eqref{eq:obj}-\eqref{eq:edgeConstraint} is akin to the convex program to compute heterogeneous user equilibria given a vector of road tolls~\cite{YANG20041}. However, as opposed to the convex program in~\cite{YANG20041} that considers a single-economy setting, Problem~\eqref{eq:obj}-\eqref{eq:edgeConstraint}, which applies to a mixed-economy setting, only has a toll component in the Objective~\eqref{eq:obj} for ineligible users and instead has a budget Constraint~\eqref{eq:budget} for eligible users. 

\ifarxiv
We now present the main result of this section, which establishes that any solution of Problem~\eqref{eq:obj}-\eqref{eq:edgeConstraint} is a CBCP equilibrium.
\fi

\ifarxiv
\begin{theorem} [Convex Program for CBCP Equilibrium Computation] \label{thm:existence-uniqueness-homogeneous}
Consider a CBCP scheme $(\ttau, B)$ and the setting when the VoTs of all eligible users do not vary with time, i.e., for all $g \in \G_E$, $v_{t,g} = v_{t',g}>0$ for all $t,t' \in [T]$. Then, the optimal solution $\y^*$ of the convex Program~\eqref{eq:obj}-\eqref{eq:edgeConstraint} is a CBCP $(\ttau, B)$-equilibrium.
\end{theorem}
\fi

\begin{hproof}
To prove this claim, we derive the \ifarxiv first-order necessary and sufficient \fi KKT conditions of Problem~\eqref{eq:obj}-\eqref{eq:edgeConstraint} and show that these conditions correspond to equilibrium conditions for all users. Given the close connection between Problem~\eqref{eq:obj}-\eqref{eq:edgeConstraint} and existing convex programs for computing heterogeneous user equilibria~\cite{YANG20041}, the correspondence between the KKT conditions of Problem~\eqref{eq:obj}-\eqref{eq:edgeConstraint} and the equilibrium conditions for the ineligible users is immediate. To establish that the KKT conditions of Problem~\eqref{eq:obj}-\eqref{eq:edgeConstraint} correspond to the equilibrium conditions for the eligible users, we fix a group $g$ and consider two cases depending on the value of the dual variable $\mu_g$ of Constraint~\eqref{eq:budget}: (i) $\mu_g = 0$, and (ii) $\mu^g > 0$. In both cases, we use the complimentary slackness conditions and the properties of the optimal solution of the individual optimization Problem~\eqref{eq:objIOP}-\eqref{eq:con3IOp} of eligible users, derived in Lemma~\ref{lem:optEligible}, to establish an equivalence between the KKT conditions of Problem~\eqref{eq:obj}-\eqref{eq:edgeConstraint} and the equilibrium conditions for the eligible users, which proves our claim.
\end{hproof}

\ifarxiv
For a complete proof of Theorem~\ref{thm:existence-uniqueness-homogeneous}, see Appendix~\ref{apdx:convexProgramPf}. The convex Program~\eqref{eq:obj}-\eqref{eq:edgeConstraint} provides an efficient method to compute CBCP equilibria \ifarxiv in the setting \fi when all eligible users have time-invariant VoTs. In particular, this convex program enables a central planner to evaluate the efficacy of different CBCP schemes on certain societal scale objectives by studying the corresponding equilibria (see Section~\ref{sec:CBCPDesign}). \ifarxiv Finally, we note an immediate consequence of Theorem~\ref{thm:existence-uniqueness-homogeneous}, independently of Lemma~\ref{lem:uniqueness-edge}, that the aggregate edge flow vector $\x$ is unique given any CBCP scheme $(\ttau, B)$ when the travel time functions are strictly convex in the setting when all eligible users have values of time that do not vary over the $T$ periods. \fi

\fi

\ifarxiv

\section{Comparative Statics Analysis of CBCP Equilibria} \label{sec:CompStatics}


In this section, we investigate the properties of CBCP equilibria by performing a comparative statics analysis to characterize the changes in the equilibria induced by CBCP schemes given changes in the tolls set on the express lane or budgets distributed to eligible users. An investigation of the comparative statics of CBCP equilibria provides insights regarding the traffic patterns that are likely to be realized under changes to a CBCP scheme through modifications of the road tolls or the distributed budgets. In particular, such an analysis can help guide a central planner regarding the direction in which the tolls or budgets should be adjusted to achieve desired traffic patterns in the system. Furthermore, as we consider a mixed-economy setting, as opposed to classically studied single-economy settings, a comparative statics analysis helps glean insights into how the introduction of budgets to a certain fraction of users influences traffic patterns.

To this end, we initiate our comparative statics analysis of CBCP equilibria by studying the changes in the aggregate equilibrium express lane flows when the express lane toll is increased or decreased (Sections~\ref{subsec:comparativeStatics} and~\ref{subsec:hardness}). We then investigate how an increase or decrease in the budgets distributed to eligible users influences the aggregate equilibrium express lane flows and the corresponding eligible user travel costs (Section~\ref{subsec:BudgetMonotonicity}). We note while several of our obtained comparative statics results align with standard economic intuition, we also obtain some counter-intuitive results, e.g., the violation of a natural substitutes condition (see Section~\ref{subsec:hardness} for a definition), due to the introduction of travel credits for eligible users.

\subsection{Aggregate Edge Flows in Response to Toll Changes} \label{subsec:comparativeStatics}

In this section, we study the change in the aggregate equilibrium express lane flow at a given period when the toll on the express lane is increased or decreased. In particular, in alignment with economic intuition, we show that an increase in the express lane toll at period $t$ results in a (weak) reduction in the aggregate equilibrium express lane flow at that period.

\begin{lemma} [Monotonicty of Edge Flows with Tolls] \label{lem:TollMonotonicity}
Consider two CBCP schemes $(\ttau, B)$ and $(\Tilde{\ttau}, B)$, where $\Tilde{\tau}_t > \tau_t$ and $\Tilde{\tau}_{t'} = \tau_{t'}$ for all $t' \neq t$. Then, at equilibrium, the aggregate flows on the express lane at period $t$ satisfies $x_{1,t}(\Tilde{\ttau}) \leq x_{1,t}(\ttau)$, where $x_{e,t}(\ttau)$ denotes the equilibrium aggregate flow on edge $e$ at period $t$ under the CBCP scheme with toll $\ttau$.
\end{lemma}

For a proof of Lemma~\ref{lem:TollMonotonicity}, see Appendix~\ref{apdx:PfmonotonicityTolls}.\footnote{We note that Lemma \ref{lem:TollMonotonicity} can also be extended to the setting when the tolls on the express lane are the same at each period, i.e., $\tau_t = \tau_{t'}$ for all $t, t' \in [T]$. In particular, using arguments similar to those used in the proof of Lemma \ref{lem:TollMonotonicity}, it can be shown that if the toll at each period is increased from $\tau$ to $\Tilde{\tau}$, then the aggregate flow on the express lane at each period is (weakly) reduced.} Lemma \ref{lem:TollMonotonicity} establishes a natural (weakly) monotonically decreasing relationship between the aggregate equilibrium flow on the express lane and the corresponding toll at period $t$. While this monotonicity relation between edge flows and tolls naturally holds in a single-economy setting, wherein all users pay money out-of-pocket, Lemma~\ref{lem:TollMonotonicity} extends this monotonicity relation to the mixed-economy setting wherein a certain proportion of the users are given travel credits. We reiterate that the result of Lemma~\ref{lem:TollMonotonicity} mirrors classical economic theory wherein the demand for a resource is (weakly) reduced with an increase in its price.

\subsection{Violation of Substitutes Condition} \label{subsec:hardness}

While the monotonicity relation established in Lemma~\ref{lem:TollMonotonicity} aligns with standard intuition from economic theory, we now show that, due to the introduction of travel credits for eligible users, a natural substitutes condition may be violated even in the setting when all eligible users have time-invariant values of time. We note that the substitutes condition is fundamental to the study of classical economic theory, as it serves as a critical condition for the existence of market-clearing prices at which the demand for the given resources equals the capacity for those resources~\cite{kelso1982job,hatfield2005matching}.

To establish that the substitutes condition does not hold in the traffic routing setting considered in this work, we first recall the substitutes condition from economic theory, which states that an increase in the price of a particular resource cannot result in reduced demand for other resources. Then, in the context of CBCP schemes, the substitutes condition can be stated as follows.

\begin{definition} [Substitutes Condition for CBCP Schemes] \label{def:substitutes}
Consider two CBCP schemes $(\ttau, B)$ and $(\Tilde{\ttau}, B)$, where $\Tilde{\tau}_t > \tau_t$ and $\Tilde{\tau}_{t'} = \tau_{t'}$ for all $t' \neq t$. Then, the aggregate equilibrium edge flows $\x(\ttau)$ satisfies the substitutes condition if $x_{1,t'}(\Tilde{\tau}) \geq x_{1,t'}(\ttau)$ holds for all periods $t' \neq t$, where $x_{e,t}(\ttau)$ denotes the equilibrium aggregate flow on edge $e$ at period $t$ under the CBCP scheme with toll $\ttau$.
\end{definition}

In particular, Definition~\ref{def:substitutes} states that if the toll on the express lane is increased at period $t$ with the express lane tolls at all other periods kept fixed, then the express lane aggregate flow at any period $t' \neq t$ weakly increases. We now show using a counterexample that the equilibria induced by CBCP schemes do not satisfy this substitutes condition even in the setting when all eligible users have time-invariant values of time, as is elucidated through the following proposition.

\begin{proposition} [Violation of Substitutes Property] \label{prop:SubstitutesViolation}
Consider two CBCP schemes $(\ttau, B)$ and $(\Tilde{\ttau}, B)$, where $\Tilde{\tau}_t > \tau_t$ and $\Tilde{\tau}_{t'} = \tau_{t'}$ for all $t' \neq t$. Then, if the eligible users have time-invariant values of time, there exists an instance such that for some period $t' \neq t$, the equilibrium aggregate flows on the express lane satisfies $x_{1,t'}(\Tilde{\ttau}) < x_{1,t'}(\ttau)$. Here, $x_{e,t}(\ttau)$ denotes the equilibrium aggregate flow on edge $e$ at period $t$ under the CBCP scheme with toll $\ttau$.
\end{proposition}

Proposition~\ref{prop:SubstitutesViolation} establishes that even in the setting when eligible users have time-invariant values of time, i.e., the condition under which CBCP equilibria can be computed using the convex Program~\eqref{eq:obj}-\eqref{eq:edgeConstraint}, the substitutes condition may not hold. While we defer a proof of Proposition~\ref{prop:SubstitutesViolation} to Appendix~\ref{apdx:pfSubstitutesViolation}, a few comments about the violation of the substitutes condition are in order, which also provide insights into the counterexample used to prove this result. To this end, first note that if all users are ineligible, then the substitutes condition trivially holds as ineligible users' travel decisions are independent across the different periods (see Section~\ref{subsec:cbcpDef}). In other words, an increase in the toll at a given period does not influence the aggregate equilibrium express lane flow at another period if all users are ineligible. As a result, the violation of the substitutes condition occurs due to eligible users whose travel decisions are coupled across periods because of their budget Constraint~\eqref{eq:con3IOp}. In particular, in the counterexample used to prove Proposition~\ref{prop:SubstitutesViolation}, an increase in the toll at a given period may result in eligible users spending more of their budget to continue using the express lane at that period, resulting in a decreased aggregate edge flow at another period.

We further note that the violation of the substitutes condition has important implications for a central planner seeking to enforce a particular traffic pattern or desired equilibrium flow in the system. For instance, the central planner may seek to maintain certain travel times on the express lane at all periods. As the substitutes condition is critical for the existence of market-clearing prices, the result of Proposition~\ref{prop:SubstitutesViolation} implies that enforcing a desired equilibrium traffic flow in the system (or maintaining certain travel times on the express lane) may not be possible for a central planner.

\subsection{Aggregate Edge Flows and Eligible User Travel Costs in Response to Budget Changes} \label{subsec:BudgetMonotonicity}

Having studied the influence of the changes in tolls on the equilibria induced by CBCP schemes, we now investigate how an increase or decrease in the budgets distributed to eligible users influences the aggregate equilibrium express lane flows and the corresponding eligible user travel costs. In particular, we first show that the aggregate equilibrium express lane flow at each period (weakly) increases with an increase in the budget of eligible users. While this result follows as eligible users can use the express lane for more periods with a higher budget, we then show, contrary to intuition, that there are instances when higher eligible user budgets result in increased travel costs for those users even when their values of time are time-invariant.

We begin by establishing the monotonicity relation between the budgets of eligible users and the corresponding aggregate equilibrium express lane flows, as is elucidated through the following lemma.

\begin{lemma} [Monotonicty of Edge Flows with Budgets] \label{lem:BudgetMonotonicity}
Consider two CBCP schemes $(\ttau, B)$ and $(\ttau, \Tilde{B})$, where $\Tilde{B} > B$. Then, at equilibrium, the aggregate flows on the express lane at all periods $t$ satisfies $x_{1,t}(B) \leq x_{1,t}(\Tilde{B})$, where $x_{e,t}(B)$ denotes the equilibrium aggregate flow on edge $e$ at period $t$ under the CBCP scheme where eligible users receive a budget $B$.
\end{lemma}

Lemma~\ref{lem:BudgetMonotonicity} states that the aggregate equilibrium express lane flow is monotonically non-decreasing with the budgets of the eligible users. This result follows as eligible users can use the express lane for more periods with an increased budget. For a proof of Lemma~\ref{lem:BudgetMonotonicity}, we refer to Appendix~\ref{apdx:pfBudgetMonotonicity} and note that it follows a similar line of reasoning to that in the proof of Lemma~\ref{lem:TollMonotonicity}.

We further note that Lemma~\ref{lem:BudgetMonotonicity} aligns with standard economic intuition that an increase in the budgets of eligible users should result in higher express lane flows as eligible users can use the express lane for more periods with a higher budget. Despite this result, we note that an increase in the budgets of eligible users does not necessarily result in reduced travel costs for those users. In particular, we show that there are instances when increasing the budget of eligible users increases their travel costs even when the values of time of eligible users are time-invariant.

\begin{proposition} [Non-Monotonicity of Eligible User Travel Costs with Budget Changes] \label{prop:nonMonotonicityBudget}
Consider two CBCP schemes $(\ttau, B)$ and $(\ttau, \Tilde{B})$, where $\Tilde{B} > B$. Then, even in the setting when eligible users have time-invariant values of time, there exists an instance such that the eligible users will incur a higher travel cost at the equilibrium induced by the CBCP scheme $(\ttau, \Tilde{B})$ with the higher budget as compared to that induced by the CBCP scheme $(\ttau, B)$ with the lower budget.
\end{proposition}

For the counterexample used to prove Proposition~\ref{prop:nonMonotonicityBudget}, see Appendix~\ref{apdx:pfBudgetMonViolation}. We note that the primary reason for the non-monotonic relationship between the change in the budgets and the travel costs of the eligible users, as established in Proposition~\ref{prop:nonMonotonicityBudget}, is that all eligible users (and not just an individual non-atomic user) receive a higher budget. As a result, an increase in the budget for all eligible users creates ``competition'' between them to use the express lane, driving up the express lane travel times and the travel costs of eligible users.
\fi


\vspace{-5pt}

\ifarxiv
\section{Credit-Based Congestion Pricing Scheme Design} \label{sec:CBCPDesign}
\fi

\vspace{-1pt}

\ifarxiv

Thus far, we have investigated the properties of equilibria induced by CBCP schemes and developed a convex program to compute CBCP equilibria in the setting when eligible users' values of time are time-invariant. While an analysis of equilibrium properties aids in developing a foundational understanding of the influence of a policy introduced by a central planner on the resulting traffic pattern, a central planner is typically interested in deploying an optimal policy to achieve particular societal goals. \ifarxiv For instance, traffic demand management methods, such as congestion pricing, are often deployed to minimize the total travel time of all users in the system. In the context of CBCP schemes on a multi-lane freeway section, central planners are typically concerned with societal goals beyond system efficiency as they provide budgets to lower-income users to balance the travel costs incurred by ineligible and eligible users. Given the range of possible societal objectives of a central planner, \else To this end, \fi in this section, we present a general framework to design optimal CBCP schemes in the setting when eligible users' values of time are time-invariant. \ifarxiv In particular, in Section~\ref{subsec:bi-level}, we present a bi-level optimization framework for the design of optimal CBCP schemes, wherein the upper-level problem involves a central planner choosing a CBCP scheme $(\ttau, B)$ to induce an equilibrium flow given by the solution of the convex Program~\eqref{eq:obj}-\eqref{eq:edgeConstraint} (the lower-level problem) that optimizes the societal objective of the central planner. Since solving bi-level programs is, in general, computationally challenging, we then present an algorithmic approach based on dense sampling to compute an approximation to the optimal CBCP scheme (Section~\ref{subsec:continuity}). We further motivate applying a dense sampling approach to solve the bi-level program in Section~\ref{subsec:continuity} by establishing some continuity relations between the resulting equilibria (and aggregate edge flows) and the corresponding toll and budget parameters that characterize a CBCP scheme. \else In particular, we present a bi-level optimization framework for designing optimal CBCP schemes (Section~\ref{subsec:bi-level}) and develop an algorithmic approach based on dense sampling to compute an approximation to the optimal CBCP scheme (Section~\ref{subsec:continuity}). \fi

\fi

\vspace{-6pt}

\subsection{Bi-Level Optimization Framework} \label{subsec:bi-level}

\vspace{-2pt}

In this section, we present a bi-level optimization framework for designing optimal CBCP schemes to achieve particular societal objectives of a central planner. We focus on the setting when eligible users have time-invariant VoTs, in which case CBCP equilibria can be computed through the solution of the convex Program~\eqref{eq:obj}-\eqref{eq:edgeConstraint}.

\ifarxiv
To present the bi-level optimization problem of the central planner, we first introduce some notation. In particular, we model the societal objective of a central planner through a cost function $f: \mathbb{R}^{2 \times T \times |\G|} \rightarrow \mathbb{R}$, where $f(\y)$ denotes the societal cost associated with the flow $\y \geq \0$ that lies in a feasible set $\Omega$ defined by Constraints~\eqref{eq:allocation}-\eqref{eq:edgeConstraint}. Furthermore, we denote $\F_U \subseteq \mathbb{R}_{\geq 0}^{T+1}$ as the set of feasible CBCP schemes $(\ttau, B)$. Then, the goal of the central planner is to find a feasible CBCP scheme $(\ttau^*, B^*) \in \F_U$ such that the resulting equilibria $\y(\ttau^*, B^*)$ has the lowest societal cost among all feasible CBCP schemes, i.e., $f(\y(\ttau^*, B^*)) \leq f(\y(\ttau, B))$ for all feasible CBCP schemes $(\ttau, B) \in \F_U$, where $\y(\ttau, B)$ is an equilibrium flow given by the solution of Problem~\eqref{eq:obj}-\eqref{eq:edgeConstraint} under the scheme $(\ttau, B)$. In particular, the objective of the central planner can be captured through the following bi-level optimization problem\ifarxiv\else\footnote{The bi-level Program~\eqref{eq:Bi-level-Obj}-\eqref{eq:LowerLevelProb} can be readily extended to the setting when a central planner imposes feasibility restrictions on the equilibrium flows $\y$, i.e., $\y \in \F_L \subseteq \Omega$, which can arise when the central planner seeks to ensure that the total number of users on the express lane does not exceed a specified threshold. For the simplicity of exposition, we focus the setting when a central planner's objectives are specified through a societal cost function $f$ rather than through constraints on the flows $\y$.}\fi
\fi
\begin{mini!}|s|[2]                   
    {\substack{\y \in \Omega, \\ (\ttau, B) \in \F_U}}                              
    { f(\y(\ttau, B)), {\label{eq:Bi-level-Obj}} }  
    {\label{eq:EgBi-level}}             
    {}                                
    \addConstraint{\y(\ttau, B)}{=\text{ Solution of Problem~\eqref{eq:obj}-\eqref{eq:edgeConstraint}},} {\label{eq:LowerLevelProb}}   
\end{mini!}
where Constraint~\eqref{eq:LowerLevelProb} represents the lower-level problem of computing the equilibrium flow given a CBCP scheme $(\ttau, B)$. \ifarxiv We now present some examples of the feasibility set $\F_U$ and the societal cost function $f$, which are relevant in CBCP implementations in practical traffic routing contexts, e.g., the San Mateo 101 express Lanes Project. \fi


\ifarxiv

\begin{example} [Feasibility set $\F_U$ as Interval Constraints] \label{eg:intervalConstraints}
As with second-best tolling~\cite{VERHOEF2002281,bilevel-labbe,Larsson1998,bilevel-patricksson}, wherein the sets of allowable tolls on each edge of the traffic network are often represented by interval constraints, the feasibility set $\F_U$ can also be specified by interval constraints on the tolls and budgets. In particular, the set $\F_U$ is such that the toll on the express lane at each period $t$ satisfies $\tau_t \in [\underline{\tau}, \Bar{\tau}]$ for some specified bounds $\underline{\tau}, \Bar{\tau} \geq 0$ and the budget $B \in [\underline{B}, \Bar{B}]$ for some specified bounds $\underline{B}, \Bar{B} \geq 0$.
\end{example}

\begin{example} [Feasibility set $\F_U$ with Time-Invariant Tolls and Interval Constraints] \label{eg:timeInvariantTolls}
In practical traffic routing settings, road tolling schemes are often static, with tolls that do not vary over time. In such a setting, the feasibility set $\F_U$ is a subset of the corresponding feasibility set in Example~\ref{eg:intervalConstraints}, with the additional restriction that the tolls on the express lane additionally satisfy $\tau_t = \tau_{t'}$ for all periods $t,t' \in [T]$.
\end{example}

\begin{example} [Societal Cost Function $f$ as a Weighted Combination of Travel Costs of Eligible and Ineligible Users] \label{eg:ParetoWeights}
When deploying a CBCP scheme, a central planner typically accounts for its social welfare effects on all groups of users. A widely studied social welfare function in redistributive market design involves aggregating the utilities (or costs) of all users and weighting users' costs by a social welfare weight, which corresponds to the relative social importance that the central planner places on the welfare of this user as compared to other users~\cite{RAM-Akbarpour,Stantcheva-2016}. In particular, each user group $g$ can be associated with a social welfare weight $\lambda_g$. Then, the cost function $f$ is given by $f(\y(\ttau, B)) = \sum_{g \in \G_E} \lambda_g v_{g} \sum_{t = 1}^T \sum_{e = 1}^2 l_e(x_{e,t}) y_{e,t}^g + \sum_{g \in \G_I} \lambda_g \sum_{t = 1}^T \sum_{e = 1}^2 (v_{t,g} l_e(x_{e,t}) + \mathbbm{1}_{e = 1} \tau_t) y_{e,t}^g$, where the VoT of users in eligible groups is denoted as $v_g$ as their values of time are assumed to be time-invariant. Note that a higher welfare weight for a given user group implies that the central planner is prioritizing those user groups relative to others.
\end{example}

\begin{example} [Societal Cost Function $f$ as Revenue Maximization] \label{eg:Revenue}
A common objective for a central planner is revenue maximization~\cite{bilevel-tolls}, in which case the function $f$ can be represented as the negative of the total tolls collected from the ineligible users when they use the express lane, i.e., $f(\y(\ttau, B)) = - \sum_{t = 1}^T \sum_{g \in \G_I} \tau_t y_{1,t}^g$. Note that no revenue is received from eligible users, as these users only utilize the provided travel credits to use the express lane. Furthermore, there is no loss in revenue in providing travel credits to eligible users, as these are eventually recuperated by the central planner when eligible users use the express lane and expend credit (or once the credits expire after the $T$ periods over which the CBCP scheme is run).
\end{example}

Finally, we note that the bi-level Program~\eqref{eq:Bi-level-Obj}-\eqref{eq:LowerLevelProb} can also be readily extended to the setting when a central planner imposes feasibility restrictions on the equilibrium flows $\y$, i.e., the flow $\y$ must belong to a set $\F_L \subseteq \Omega$. Such restrictions on the equilibrium flows can arise when the central planner seeks to maintain a certain quality of service on the express lane by ensuring that the total number of users on that lane does not exceed a specified threshold. For the simplicity of exposition, we do not focus our attention on such feasibility restrictions on the equilibrium flows, and for the remainder of this work, consider the setting when a central planner's objectives are purely specified through a societal cost function $f$ rather than through constraints on the set of feasible flows $\y$.

\fi

\vspace{-5pt}

\ifarxiv
\subsection{Algorithmic Approach to Solve Bi-level Optimization Problem} \label{subsec:continuity}
In the previous section, we showed that a central planner's objective of determining the optimal CBCP scheme to minimize a societal cost function $f$ can be expressed as a bi-level optimization problem. However, solving bi-level programs is, in general, computationally challenging, and, in particular, even bi-level linear programs are NP-hard to approximate up to any constant factor in general~\cite{bi-level-hardness}. Furthermore, the societal cost function $f$ may be non-linear as it may depend on the edge travel time functions \ifarxiv (see Example~\ref{eg:ParetoWeights})\else (see Example~\ref{eg:ParetoWeights} in Appendix~\ref{apdx:examples})\fi, which are generally non-linear. Given these difficulties, analytically characterizing the optimal solution of the bi-level program is likely not possible. \ifarxiv As a result, in this section, we present a \emph{dense sampling} approach to computing an approximate solution for the bi-level Problem~\eqref{eq:Bi-level-Obj}-\eqref{eq:LowerLevelProb} and discuss its computational tractability and practical applicability. We further motivate applying the dense sampling approach to solve the bi-level program by establishing continuity relations between the resulting equilibria (and aggregate edge flows) and the corresponding toll and budget parameters that characterize a CBCP scheme. For the simplicity of exposition, we suppose that the feasible set $\F_U$ is given by interval constraints, as in Example~\ref{eg:intervalConstraints}, and note that our approach is also applicable for a broader range of feasibility sets. \else As a result, in this section, we present a \emph{dense sampling} approach to computing an approximate solution for the bi-level Problem~\eqref{eq:Bi-level-Obj}-\eqref{eq:LowerLevelProb}. For the simplicity of exposition, we suppose that the feasible set $\F_U$ is given by interval constraints, wherein the express lane toll $\tau_t \in [\underline{\tau}, \Bar{\tau}]$ at each period $t$ for some $\underline{\tau}, \Bar{\tau} \geq 0$ and the budget $B \in [\underline{B}, \Bar{B}]$ for some $\underline{B}, \Bar{B} \geq 0$, as in Example~\ref{eg:intervalConstraints} in Appendix~\ref{apdx:examples}. We do note, however, that our approach is also applicable for a broader range of feasibility sets.

\fi

\fi

\ifarxiv
\paragraph{Dense Sampling Approach:}

\fi
To solve the bi-level Problem~\eqref{eq:Bi-level-Obj}-\eqref{eq:LowerLevelProb}, we discretize the feasible set $\F_U$ given by interval constraints as a grid with a step size of $s$ in each component (in general, the step size can vary across each component). That is, the express lane toll at any period $t$ lies in the set $\A_s = \{ \underline{\tau}, \underline{\tau}+s, \ldots, \Bar{\tau} \}$ and the eligible user budget lies in the set $\B_s = \{ \underline{B}, \underline{B}+s, \ldots, \Bar{B} \}$. Further, we let $\C_s$ be the set of all toll and budget combinations $(\ttau, B)$ in this discretized grid\ifarxiv with a step size of $s$.\fi Then, to compute a good solution to \ifarxiv the bi-level \fi Problem~\eqref{eq:Bi-level-Obj}-\eqref{eq:LowerLevelProb} with a low societal cost, we evaluate the optimal solution of the convex Program~\eqref{eq:obj}-\eqref{eq:edgeConstraint} for each CBCP scheme $(\ttau, B)$ in the set $\C_s$ and return the CBCP scheme with an equilibrium flow with the least societal cost\ifarxiv among all CBCP schemes in the set $\C_s$.\fi That is, we return a CBCP scheme $(\ttau_s^*, B_s^*) \in \C_s$ with a corresponding equilibrium flow $\y(\ttau_s^*, B_s^*)$, such that $f(\y(\ttau_s^*, B_s^*)) \leq f(\y(\ttau, B))$ for all $(\ttau, B) \in \C_s$ with equilibrium flows $\y(\ttau, B)$.

\ifarxiv

A few comments about using dense sampling to solve the bi-level Problem~\eqref{eq:Bi-level-Obj}-\eqref{eq:LowerLevelProb} are in order. First, noting that a CBCP scheme $(\ttau, B)$ belongs to a $T+1$ dimensional space, the computational complexity of dense sampling scales with the number of points in the discretized grid given by $|\C_s| = O(\frac{\Bar{\tau}-\underline{\tau})^T (\Bar{B}-\underline{B}) }{s^{T+1}})$.  In particular, our dense sampling approach involves solving the convex Program~\eqref{eq:obj}-\eqref{eq:edgeConstraint} $|\C_s|$ times. In other words, the computational complexity of the dense sampling approach scales exponentially in the step-size $s$ with the number of periods $T$. However, in practical traffic routing settings, the number of periods $T$ is typically moderately sized, e.g., $T=30$ if the CBCP scheme is run for a month, and the tolls imposed on the express lanes tend to be static and thus fixed over a certain period. Given the time-invariance of practically deployed tolling schemes, as in the setting in Example~\ref{eg:timeInvariantTolls}, the dense sampling approach can be reduced from one over a $T+1$ dimensional space to one over two dimensions, as the toll must be kept constant on the express lane across all periods. Thus, in the setting where the tolls on the express lane at all periods must be kept constant, our proposed dense sampling approach provides a computationally tractable method to compute an optimal CBCP scheme $(\ttau_s^*, B_s^*) \in \C_s$. Furthermore, while several other methods to solve bi-level programs exist, e.g., KKT reformulations~\cite{ruhi2018opportunities}, or second-order methods~\cite{dyro2022second}, dense sampling serves as a clear and transparent methodology for a central planner to evaluate the set of all possible CBCP schemes in the set $\C_s$ and select the one that performs the best, i.e., achieves the least societal cost $f$. In particular, in practical traffic routing settings, a central planner may prioritize finding an (approximately) optimal CBCP scheme with the least possible societal cost $f$ even at the expense of larger computational runtimes, which further motivates the practicality of the dense sampling approach.



\paragraph{Continuity Properties:}

While dense sampling provides a method to evaluate the optimal CBCP scheme in a discretized set $\C_s$, such a scheme may be sub-optimal for the feasibility set $\F_U$. To this end, we now present continuity properties of the equilibrium flows (and the aggregate edge flows) in the toll and budget parameters which highlight that performing dense sampling helps achieve approximately optimal solutions by the derived continuity relations. In particular, the continuity relations motivate the efficacy of a dense sampling approach as the equilibrium flows and the corresponding societal cost are unlikely to change much between two subsequent points in $\C_s$ for small step sizes $s$. As a result, the optimal scheme in the set $\C_s$ serves as a good approximation to the optimal solution to the bi-level Program~\eqref{eq:Bi-level-Obj}-\eqref{eq:LowerLevelProb} for a small enough step-size $s$.

We now present our continuity result that relates both the equilibrium flows and the aggregate edge flows to the corresponding toll and budget parameters that characterize a CBCP scheme, as elucidated through the following lemma. Here, we let $\Y(\ttau, B)$ denote the set of equilibrium flows corresponding to the solution to Problem~\eqref{eq:obj}-\eqref{eq:edgeConstraint} since the equilibrium flows are, in general, non-unique for any given CBCP scheme $(\ttau, B)$.\footnote{Given the potential non-uniqueness of equilibrium flows, in the statement of Lemma~\ref{lem:contRelations}, we use the double arrow ``$\implies$'' to denote a correspondence, which is a map that associates every point in the domain of the correspondence to a subset in its range.  Furthermore, we present formal definitions of an upper semi-continuous and locally bounded correspondence mentioned in Lemma~\ref{lem:contRelations} in Appendix~\ref{apdx:defs}.}


\begin{lemma} [Continuity of Equilibrium Flows] \label{lem:contRelations}
Suppose that all eligible users have time-invariant values of time and the feasible set $\F_U$ is such that the tolls $\ttau>\0$. Further, let $\Y(\ttau, B)$ denote the set of solutions to Problem~\eqref{eq:obj}-\eqref{eq:edgeConstraint} and $\x(\ttau, B)$ denote the corresponding unique edge flow. Then, the correspondence $(\ttau, B) \implies \Y(\ttau, B)$ is upper semi-continuous and locally bounded and $\x(\ttau, B)$ is a continuous function in $(\ttau, B)$ over any open set of toll and budget parameters in $\F_U$. Furthermore, if the set of equilibrium flows $\Y(\ttau, B)$ is singleton, i.e., $\Y(\ttau, B) = \{ \y(\ttau, B) \}$, then the equilibrium flow $\y(\ttau, B)$ is continuous in $(\ttau, B)$ for any open set of toll and budget parameters in $\F_U$.
\end{lemma}

Lemma~\ref{lem:contRelations} establishes that the correspondence $\Y(\ttau, B)$ of equilibrium flows is upper semi-continuous and locally bounded (we refer to Appendix~\ref{apdx:defs} for definitions of these terms) and the aggregate edge flow $\x(\ttau, B)$ is a continuous function in the toll and budget parameters characterizing a CBCP scheme. We prove Lemma~\ref{lem:contRelations} through an application of Berge's theorem of the maximum~\cite{kreps-book} and present a complete proof of this claim in Appendix~\ref{apdx:contRelations}. Finally, we note an immediate consequence of Lemma~\ref{lem:contRelations}, which implies that the CBCP scheme found using dense sampling is a good approximation for the optimal solution to the bi-level Program~\eqref{eq:obj}-\eqref{eq:edgeConstraint} when the function $f$ depends solely on the aggregate edge flows $\x$, i.e., the sum of the flows of all users. 


\begin{corollary} \label{cor:FunctionCont}
Suppose that all eligible users have time-invariant values of time, the feasible set $\F_U$ is such that the tolls $\ttau>\0$, and the societal cost function $f$ depends solely on the aggregate edge flows $\x$, i.e., the bi-level optimization Objective~\eqref{eq:Bi-level-Obj} can be expressed as $f(\x(\tau, B))$, where $\x(\ttau, B)$ is the edge flow corresponding to the solution of the convex Program~\eqref{eq:obj}-\eqref{eq:edgeConstraint}. Then, if the function $f$ is continuous in the edge flows $\x$ it holds that $f$ is also continuous in $(\ttau, B)$.
\end{corollary}

The proof of Corollary~\ref{cor:FunctionCont} is immediate from the continuity of function compositions. Observe that an example of a function $f$ that satisfies the condition of Corollary~\ref{cor:FunctionCont} is the total travel time of all users, which can be expressed as $\sum_{t = 1}^T \sum_{e = 1}^2 x_{e,t}(\ttau, B) l_e(x_{e,t}(\ttau, B))$. By establishing the continuity of the societal cost function $f$ in $(\ttau, B)$, Corollary~\ref{cor:FunctionCont} implies that small changes in the toll and budget parameters will result in small changes in the societal cost $f$. As a result, for a small enough step size $s$, the CBCP scheme output by the dense sampling provides a good approximation to the optimal CBCP scheme. 

\fi

\vspace{-5pt}

\section{Numerical Experiments} \label{sec:experiments}

\ifarxiv

We now investigate the influence of CBCP schemes on traffic patterns and study their optimal design through a real-world application study of the San Mateo 101 Express Lanes Project. Our numerical results not only validate our comparative statics analysis results but also show that the optimal CBCP scheme can vary widely based on the central planner's objective when deploying a CBCP scheme, demonstrating that a principled approach using bi-level optimization proposed in this work is key to realizing the benefits of CBCP schemes. In this section, we first present the numerical implementation details for our experiments and our procedure to calibrate the model parameters of the multi-lane freeway segment, such as the edge travel time functions and users' values of time (Section~\ref{subsec:calibration}). Next, in Section~\ref{subsec:sensitivityNumerical}, we present sensitivity results that demonstrate the variation in the travel time and proportion of users on the express lane with the tolls and budgets, reflective of real-world multi-lane freeway segments with express lanes. Then, in Section~\ref{subsec:ObjectiveNumerical}, we apply the dense sampling approach in Section~\ref{subsec:continuity} to solve the bi-level optimization Problem~\eqref{eq:Bi-level-Obj}-\eqref{eq:LowerLevelProb} to obtain optimal CBCP schemes under different societal objectives. Finally, we discuss policy implications of this work in Section~\ref{subsec:policy}.

\fi


\ifarxiv

\subsection{Model Calibration and Implementation Details} \label{subsec:calibration}

In this section, we describe the implementation details for our experiments and the method used to calibrate the edge travel time functions, user demand, and the user values of time based on a real-world application study of the San Mateo 101 Express Lanes Project. We note that as part of this project, San Mateo County is constructing 22 miles of express lanes on the US 101 highway in the San Francisco Peninsula, which connects Santa Clara and San Mateo counties with the City and County of San Francisco. For our experiments, we modelled the northbound portion of the US 101 highway involved in this project. We further obtained data regarding the usage of the newly opened express lane relative to the untolled general purpose (GP) lanes in San Mateo.

\paragraph{Travel time calibration;} As in the US 101 highway, which has one express lane and three GP lanes, we set up our two-edge Pigou network model, wherein the first edge represents the express lane and the second edge corresponds to the three GP lanes without tolls (see Section~\ref{subsec:preliminaries} for a discussion on this model simplification). To calibrate the travel time functions on the two edges, we queried average hourly weekday travel speed and vehicle flow data from September 2 - September 30, 2019 from Caltrans' Performance Measurement System (PeMS) database~\cite{pems-database}, depicted in Figure~\ref{fig:traveltime}. Using this data, we then fit the parameters of the commonly used BPR travel time function~\cite{utraffic}, depicted in the blue curve in Figure~\ref{fig:traveltime}, defined as 
\begin{align*}
    l_e(x_e) = \xi_e \left( 1 + a\left(\frac{x_e}{\kappa_e} \right)^b \right),
\end{align*}
where $a, b$ are constants, $\xi_e$ is the free-flow travel time on edge e, and $\kappa_e$ is the capacity, i.e., the number of users beyond which the travel time on the edge rapidly increases, of edge e. The calibrated parameter values are as follows: $a = 0.2, b = 6, \xi_e = 19.4 \text{ minutes, and } \kappa_e = 1650$ veh/hr. Given the non-linearity of the BPR function, to tractably solve the convex Program~\eqref{eq:obj}-\eqref{eq:edgeConstraint} for large problem instances, we further employed the commonly used piecewise affine approximation of the BPR function~\cite{SalazarTsaoEtAl2019} depicted in the orange line in Figure~\ref{fig:traveltime}. In particular, the piecewise linear approximation used is given by
\[ l_e(x_e) = l_0 + \begin{cases} 0 & \text{if}\ x_e \le \lambda \kappa_e \\ 
\beta_e (x_{e,t} - \lambda \kappa_e) & \text{otherwise}
\end{cases}\]
where $l_0$ is the height of the horizontal line, $\beta_e$ is the slope of the second line, and $\lambda \kappa_e$ represents the threshold at which the travel time changes in the piecewise linear approximation. We reiterate that since edge one represents the express lane while edge two represents the three general-purpose lanes, the capacities satisfy $\kappa_2 = 3 \kappa_1$ and the slope of the second segment of the piecewise linear approximation satisfies $\beta_2 = \frac{\beta_1}{3}$. Finally, we note that after minimizing the mean squared error between the estimated travel times of the piecewise linear approximation and that obtained from PeMS, the final values for the parameters of the piecewise linear approximation to the BPR function are $l_0 = 19.4$ minutes, $\lambda=0.786$, $\kappa_1 = 1,650$ veh/hr, and $\beta_1 = 0.1256$. For these parameters, the piecewise linear approximation in orange closely approximates the BPR function in blue in Figure~\ref{fig:traveltime}, which further motivates the use of a piecewise linear approximation in modeling travel times for our experiments.

\begin{figure}[!htb]
    \centering
    \begin{minipage}{.45\textwidth}
        \centering
        \includegraphics[scale=0.5]{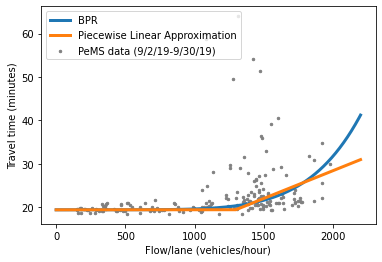} 
        \caption{{\small \sf Calibration of a piecewise linear approximation of the travel time function for a single-lane along the San Mateo 101 Express Lanes Project study area.}}
    \label{fig:traveltime}
    \end{minipage} \hspace{5pt}
    \begin{minipage}{0.45\textwidth}
        \centering
        \includegraphics[scale=0.5]{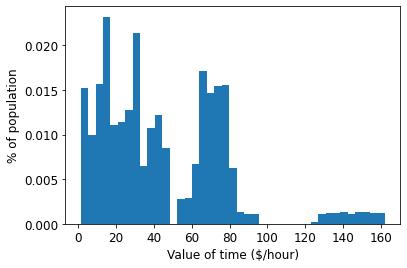} 
        \caption{{ \small \sf Approximated distribution of the values of time of users (in \$/hr) across the San Mateo and Santa Clara Counties.}}
        \label{fig:vot_distr}
    \end{minipage}
\end{figure}

\paragraph{Demand profile calibration:} The total demand for the 4-lane highway was estimated to be around 8,000 vehicles per hour, about the maximum flow observed in the PeMS data. Furthermore, for our experiments, we assumed the CBCP scheme is run over the course of a week with five working days, i.e., $T=5$. To determine the eligibility of users for free express lane credits, we obtained the distribution of household incomes across users using the 2020 US Census American Community Survey (ACS) annual household earnings estimates for Santa Clara and San Mateo counties \cite{ACS2021}. Furthermore, as in the San Mateo 101 Express Lanes Project, we applied a threshold of 200\% of the federal poverty limit, resulting in approximately 17\% of road users in the eligible group. Next, to approximate the VoT distribution across the users, we used the income distribution from the ACS annual earnings data as a surrogate representation of their VoTs. We reiterate that several other valid representations for users' VoTs are also possible, and we assume proportionality between users' income and their VoT for simplicity. In particular, the VoT of each user at each of the five periods was generated by 1) generating a baseline average VoT for each user using the ACS annual earnings data\footnote{The ACS annual household income data reports the estimated percentage of the population in each county in each of 10 income groups (less than $\$10,000$, $\$10,000$ to $\$14,999$, ..., $\$150,000$ to $\$199,999$, and $\$200,000$ or more). The probability density function of individual hourly wages, which is used as a proxy for users' VoTs, across the study population was estimated from this data by attributing the middle of each annual income interval to each of the corresponding groups and dividing by total work hours in a year (i.e., 40 hours/week $\times$ 52 weeks). Furthermore, since the ACS annual earnings data distinguishes between types of households, including family and non-family households, we divided the estimated wage for family households by two in order to more closely represent \textit{individual} wage levels. The resulting data, consisting of the population size at each hourly wage level, was used to estimate a step-wise probability distribution function for the values of time of users in the case study of the San Mateo 101 Express Lanes Project.} and 2) for ineligible users, generating deviations from the baseline for each period from a uniform distribution from $-0.125$ to $0.125$ such that the value of time for each user in group $g$ at period $t$ is $v_{t,g} = 1+v_g \delta_{t,g}$, where $v_g$ is drawn from the probability density function of hourly wages estimated from the ACS data and $\delta_{t,g} \sim \mathcal{U}[-0.125,0.125]$ for ineligible users and $\delta_{t,g}=1$ for eligible users. The resulting VoT distribution is displayed in Figure \ref{fig:vot_distr}, which has a mean of \$44/hour and median of \$37/hour.

\paragraph{Validation data:} In order to validate that the equilibrium edge flows and travel times produced by our model are within the same order of magnitude as in real-world applications, we queried hourly weekday lane-level vehicle flow and travel speed data from September 1 to September 30, 2022 along the first segment of the San Mateo 101 Express Lanes Project that launched in February 2022\footnote{The first of two phases of the San Mateo 101 Express Lanes Project was completed in February, 2022, rolling out the first 6-mile segment of express lane in the Southernmost portion of San Mateo County. Phase two, consisting of the remaining segment of the express lane is scheduled to open in early 2023}. Summary statistics are presented in Table~\ref{tab:US101UtilizationData}. On average, during morning peak hours (7-10 am) about 14\% of the total flow was on the express lane, resulting in travel time savings of about 33\% compared to the GP lanes. Historical data on toll levels for the US 101 Express Lanes Project are not yet available and thus were not included in the validation. However, aggregate toll data was available for the I-680 freeway in Contra Costa County in the East Bay Area~\cite{I680Quarterly}. In the first quarter of 2022, the average toll paid on the 20-mile I-680 express lane peaked at about \$4.80 per trip for travel time savings of about 2.6 minutes (13\%).

\paragraph{Implementation details} We ran our experiments on an i5-3570 processor with 32 GB of RAM and our corresponding implementation is available at \href{https://anonymous.4open.science/r/CBCP-DB80}{https://anonymous.4open.science/r/CBCP-DB80} We used the Python Gurobi Optimizer (\texttt{gurobipy} version 9.5.2) to solve the Convex Program~\eqref{eq:obj}-\eqref{eq:edgeConstraint} with the above travel time and VoT parameters. On average, it took about 9 seconds to solve this convex program for a given toll and budget combination.

\fi

\vspace{-7pt}

\ifarxiv
\subsection{Express Lane Usage and User Travel Times at CBCP Equilibria} \label{subsec:sensitivityNumerical}
\fi

\ifarxiv
In this section, we present the variation in the travel time and proportion of users on the express lane as the express lane tolls and the budgets distributed to eligible users are varied. In particular, we focus on the setting when eligible users have time-invariant VoTs, and the tolls are the same across the five periods, as in Example~\ref{eg:timeInvariantTolls}, over which the CBCP scheme is run. Further, we discretize the tolls to lie between \$0 to \$20, with \$1 increments, and budgets to lie between \$0 to \$90, with \$5 increments, and compute the solution to the convex Program~\eqref{eq:obj}-\eqref{eq:edgeConstraint} at each of the toll and budget combinations in the discretized grid. The resulting distributions of equilibrium lane choices and travel times corresponding to the optimal solution of Problem~\eqref{eq:obj}-\eqref{eq:edgeConstraint} are presented in Figures \ref{fig:express_shares} and \ref{fig:travel_times}, respectively.

\fi

\ifarxiv
\paragraph{Express Lane Usage:}
\fi
As seen in Figure~\ref{fig:express_share_all}, users are split evenly across lanes for \$0 tolls, with one-quarter of all users on the express lane and the remaining three-quarters on the three GP lanes. This observation aligns with equilibrium formation in congestion games without tolls, wherein all users traveling between the same origin and destination incur the same travel time. Further, the proportion of eligible users using the express lane ranges from 0\% when the budget is \$0 to 100\% when the budget exceeds the total cost of tolls over the five periods (i.e., for a toll $\tau$ and budget $B$ it holds that $5\tau \le B$), as reflected by the yellow portion in Figure~\ref{fig:express_share_el}. On the other hand, the share of ineligible users on the express lane is at a maximum of 29\% at the smallest non-zero toll of \$1 and \$0 budget and decreases with either increasing toll or \ifarxiv increasing eligible user \fi budget (see Figure \ref{fig:express_share_in}). From Figure~\ref{fig:express_share_all}, we also observe that the overall share of users on the express lane monotonically decreases (increases) with increasing toll (budget) values, \ifarxiv reaching a minimum of 7\% at the maximum toll of \$20 and \$0 budget. These monotonic relations between the toll and budget parameters and the corresponding proportion of users on the express lane \fi aligns with our comparative statics analysis results in Section~\ref{sec:CompStatics} \ifarxiv (see Lemmas~\ref{lem:TollMonotonicity} and~\ref{lem:BudgetMonotonicity})\fi. \ifarxiv In addition, the proportion of all users using the express lane smoothly varies with the change in the toll and budget parameters, which aligns with the continuity relation established in Section~\ref{subsec:continuity} (see Lemma~\ref{lem:contRelations}). \fi



\ifarxiv
\begin{figure}
\centering
\begin{subfigure}{.32\textwidth}
\centering
    \includegraphics[width = \textwidth]{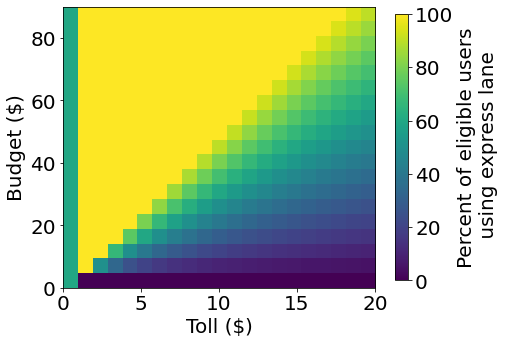}
    \caption{{\small \sf Eligible users }}
    \label{fig:express_share_el}
\end{subfigure}
    \begin{subfigure}{.32\textwidth}
    \centering
    \includegraphics[width = \textwidth]{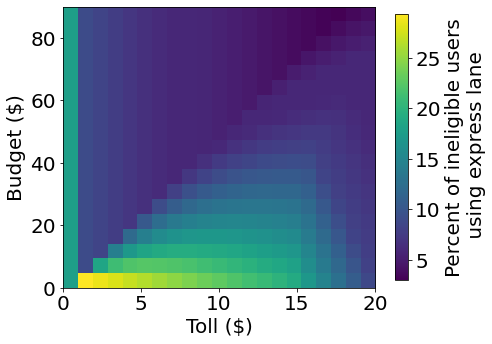}
    \caption{{\small \sf Ineligible users }}
    \label{fig:express_share_in}
\end{subfigure}
\begin{subfigure}{.32\textwidth}
    \centering
    \includegraphics[width = \textwidth]{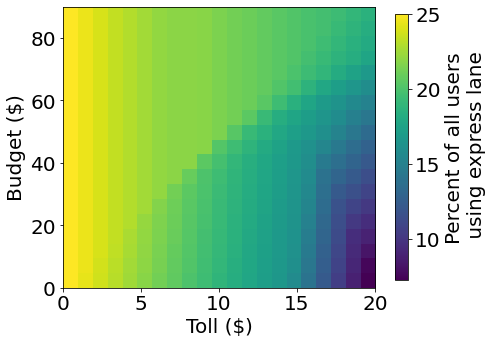}
    \caption{{\small \sf All users }} 
    \label{fig:express_share_all}
\end{subfigure}
    \caption{{\small \sf Percentage of users on the express lane corresponding to the optimal solution of the convex Program~\eqref{eq:obj}-\eqref{eq:edgeConstraint} for different toll and budget combinations. Here the tolls range from \$0 to \$20, with \$1 increments, and budgets range from \$0 to \$90, with \$5 increments.}}
    \label{fig:express_shares}
\end{figure}

\fi

\ifarxiv

\paragraph{User Travel Times:}
From Figures \ref{fig:avg_tt_ex} and \ref{fig:avg_tt_gp}, we observe that the travel times on the express and GP lanes decrease and increase, respectively, with the overall share of users on the express lane. Further, the travel time savings on the express lane increases monotonically with increasing tolls, with a maximum of about 14.8 minutes (a 43\% difference) with a \$20 toll and \$0 budget. 

We also note that the overall express lane usage and travel time savings depicted in Figure~\ref{fig:express_shares} are comparable to the data obtained from Caltrans' PeMS database~\cite{pems-database} for US 101 express lanes in September 2022. In particular, for CBCP schemes with a \$16 toll and budgets between \$0 and \$10, the travel time savings in Figure~\ref{fig:express_shares} are about 39\%, and express lane usage is around 15\%, comparable to that of US 101 express lanes (about 33\% and 14\%, respectively). Further, the valuation of \$16 for about 12.5 minutes of time savings (about \$1.28/minute) is also comparable to that of the I-680 highway in the San Francisco Bay Area (about \$1.85/minute). Thus, our numerical results reflect real-world multi-lane highways with express lanes.

\fi

\ifarxiv
\begin{figure}
\centering
\begin{subfigure}{.32\textwidth}
\centering
    \includegraphics[width = \textwidth]{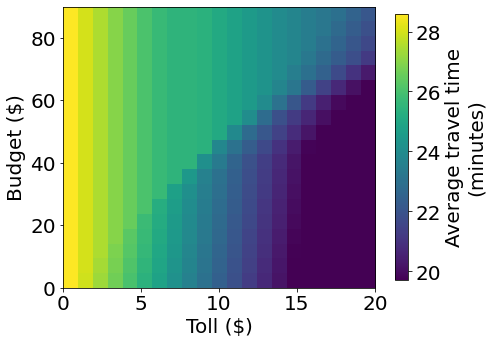}
    \caption{{\small \sf Express lane }}
    \label{fig:avg_tt_ex}
\end{subfigure}
    \begin{subfigure}{.32\textwidth}
    \centering
    \includegraphics[width = \textwidth]{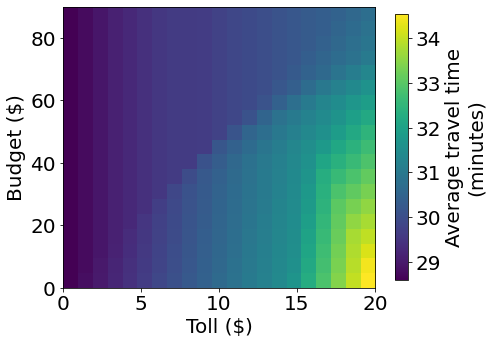}
    \caption{{\small \sf General purpose (GP) lanes }}
    \label{fig:avg_tt_gp}
\end{subfigure}
\begin{subfigure}{.32\textwidth}
    \centering
    \includegraphics[width = \textwidth]{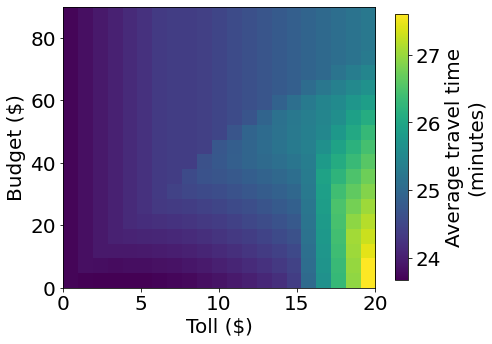}
    \caption{{\small \sf All users }}
    \label{fig:avg_tt_overall}
\end{subfigure}

    \caption{{\small \sf Average travel times corresponding to the optimal the optimal solution of the convex Program~\eqref{eq:obj}-\eqref{eq:edgeConstraint} for different toll and budget combinations. Here the tolls range from \$0 to \$20, with \$1 increments, and budgets range from \$0 to \$90, with \$5 increments.}}
    \label{fig:travel_times}
\end{figure}
\fi




\subsection{Optimal CBCP Schemes} \label{subsec:ObjectiveNumerical}
\ifarxiv
We now study the design of optimal CBCP schemes for the case study of the San Mateo 101 Express Lanes Project for a well-studied societal objective (i.e., the Pareto weighted combination of different cost (or welfare) measures) in the redistributive market design literature~\cite{RAM-Akbarpour,Stantcheva-2016}. In particular, we consider the societal objective
\begin{align*} 
    f_{\llambda}(\textbf{y}(\tau,B))&=\lambda_E \sum_{g \in \G_E}  v_{g} \sum_{t = 1}^T \sum_{e = 1}^2 l_e(x_{e,t}) y_{e,t}^g - \lambda_R \sum_{g \in \G_I} \sum_{t = 1}^T \tau_t y_{1,t}^g \\&+ \lambda_I \sum_{g \in \G_I} \sum_{t = 1}^T \sum_{e = 1}^2 (v_{t,g} l_e(x_{e,t}) +  \mathbbm{1}_{e = 1} \tau_t) y_{e,t}^g,
\end{align*} 
\fi

\ifarxiv

Table~\ref{tab:ParetoResults} presents the optimal CBCP schemes that induce an equilibrium with the minimum societal cost for each Pareto weight $\llambda$ and lists the proportion of users on the express lane and corresponding average travel times under the optimal scheme. Furthermore, Figure~\ref{fig:pareto_1x} displays the distribution of the objective function (see Section~\ref{subsec:ObjectiveNumerical}) for each of nine different Pareto weighted schemes in Table~\ref{tab:ParetoResults}. From Table~\ref{tab:ParetoResults}, we observe that the optimal CBCP scheme can vary widely based on the central planner’s objective, thus demonstrating that a principled approach using bi-level optimization proposed in this work is key to realizing the benefits of CBCP schemes. For instance, if the central planner solely optimizes for the travel costs of the eligible users, i.e., $\llambda = (1, 0, 0)$, then the optimal CBCP scheme involves providing high budgets (for eligible users to use the express lane) and setting high tolls (to push most ineligible users out of using the express lane). On the other hand, when optimizing for toll revenues, the optimal CBCP scheme corresponds to providing no budgets and setting a slightly lower toll of \$15 (to incentivize enough eligible users to use the express lane).

\fi

\ifarxiv

\begin{table}[t]
\centering
\caption{{\small \sf Optimal CBCP schemes for various Pareto weights with the corresponding travel times (TTs) and fraction of users on the express lane }  }
\vspace{-8pt}
\begin{tabular}{c|cc|ccc|cc}
\toprule
        \multicolumn{1}{c|}{Weights} & \multicolumn{2}{|c|}{Optimal CBCP} & \multicolumn{3}{|c|}{\% using express lane} & \multicolumn{2}{|c}{Average TT}  \\ 
        ($\lambda_E$,$\lambda_I$, $\lambda_R$) &$\tau$ & B &Overall & Eligible&Ineligible&Express&GP\\
        \midrule 
        (1, 0, 0) & 19 & 90 & 19 & 95 & 3 & 22.1 & 30.3 \\ 
        (0, 1, 0) & 0 & 0 & 25 & 60 & 18 & 28.2 & 28.2 \\ 
        (0, 0, 1) & 15 & 0 & 16 & 0 & 19 & 19.4 & 31.4 \\ 
        (1, 1, 1) & 13 & 0 & 17 & 0 & 21 & 20.3 & 30.9 \\ 
        (5, 1, 1) & 11 & 0 & 18 & 0 & 22 & 21.5 & 30.5 \\ 
        (10, 1, 1) & 10 & 0 & 19 & 0 & 23 & 22.1 & 30.3 \\ 
        (11, 1, 1) & 10 & 15 & 19 & 30 & 17 & 22.3 & 30.2 \\ 
        (12, 1, 1) & 11 & 45 & 19 & 82 & 7 & 22.6 & 30.1 \\ 
        (15, 1, 1) & 13 & 55 & 19 & 85 & 6 & 22.2 & 30.3 \\ 
        \bottomrule
    \end{tabular} \label{tab:ParetoResults}
\vspace{-8pt}
\end{table}

\fi

We also observe from Table~\ref{tab:ParetoResults} that as the Pareto weight \ifarxiv corresponding to \fi eligible users is increased from 1 to 10 while keeping $\lambda_I, \lambda_R = 1$, the optimal budget remains \$0 while the optimal toll decreases from \$13 to \$10. A \$0 budget is optimal for Pareto weights $\lambda_E \leq 10$ (with $\lambda_I, \lambda_R = 1$), as eligible users have much lower VoTs due to their lower incomes (see Appendix~\ref{subsec:calibration}) and thus their travel costs do not have a high enough weight relative to that of ineligible users. 
However, any increase in $\lambda_E$ beyond 10 results in optimal CBCP schemes with increasing budgets\ifarxiv, as eligible users' travel costs are given a higher weight (or priority)\fi. Further, note from Table~\ref{tab:ParetoResults} that increasing $\lambda_E$ from 10 to 12 increases the proportion of eligible users on the express lane from 0\% to 82\%, corresponding to a 21\% decrease in travel costs for these users (see Figure \ref{fig:OptCosts} in Appendix~\ref{apdx:additionalExpResults}).

\ifarxiv
\begin{figure}
\centering
\begin{subfigure}{.32\textwidth}
\centering
    \includegraphics[width = \textwidth]{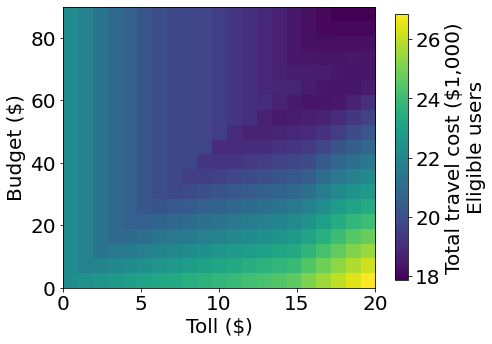}
    \caption{$(\lambda_E, \lambda_I, \lambda_T)=(1,0,0)$}
    \label{fig:total_genCost_el}
\end{subfigure}
    \begin{subfigure}{.32\textwidth}
    \centering
    \includegraphics[width = \textwidth]{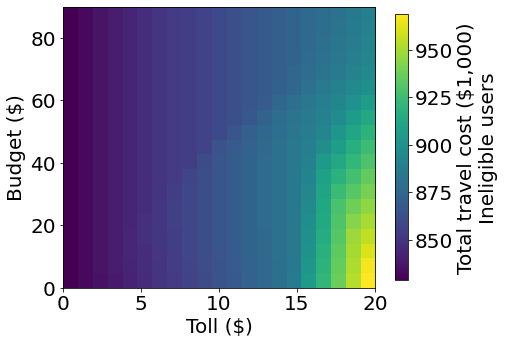}
    \caption{$(\lambda_E, \lambda_I, \lambda_T)=(0,1,0)$}
    \label{fig:total_genCost_in}
\end{subfigure}
\begin{subfigure}{.32\textwidth}
    \centering
    \includegraphics[width = \textwidth]{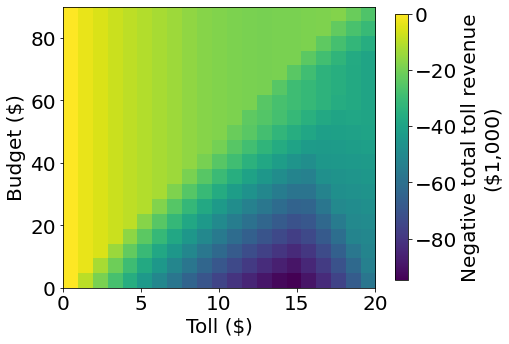}
    \caption{$(\lambda_E, \lambda_I, \lambda_T)=(0,0,1)$}
    \label{fig:total_TollRevenue}
\end{subfigure}
\begin{subfigure}{.32\textwidth}
    \centering
    \includegraphics[width = \textwidth]{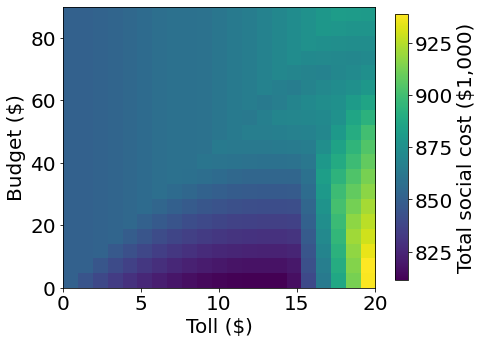}
    \caption{$(\lambda_E, \lambda_I, \lambda_T)=(1,1,1)$}
    \label{fig:1_1_1}
\end{subfigure}
\begin{subfigure}{.32\textwidth}
    \centering
    \includegraphics[width = \textwidth]{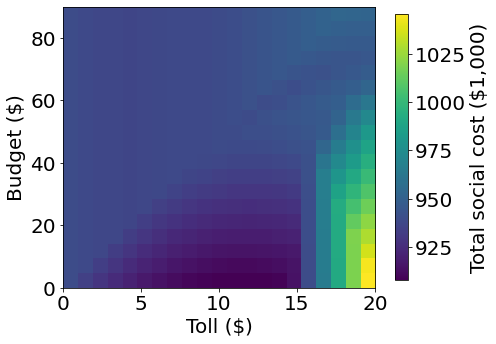}
    \caption{$(\lambda_E, \lambda_I, \lambda_T)=(5,1,1)$}
    \label{fig:5_1_1}
\end{subfigure}
\begin{subfigure}{.32\textwidth}
    \centering
    \includegraphics[width = \textwidth]{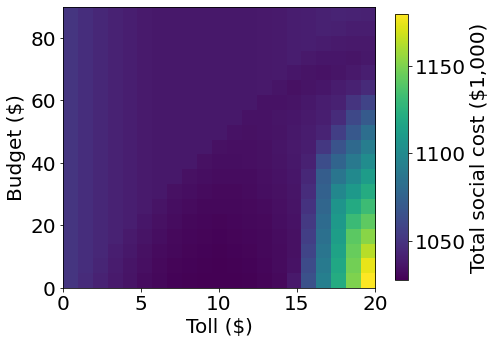}
    \caption{$(\lambda_E, \lambda_I, \lambda_T)=(10,1,1)$}
    \label{fig:10_1_1}
\end{subfigure}
\begin{subfigure}{.32\textwidth}
    \centering
    \includegraphics[width = \textwidth]{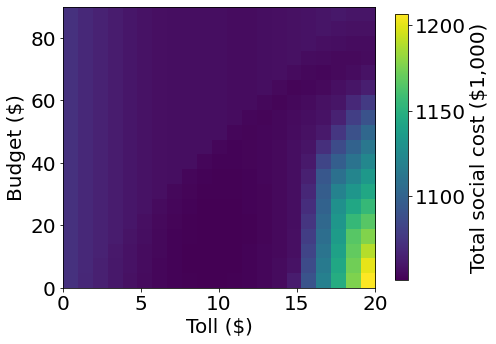}
    \caption{$(\lambda_E, \lambda_I, \lambda_T)=(11,1,1)$}
    \label{fig:11_1_1}
\end{subfigure}
\begin{subfigure}{.32\textwidth}
    \centering
    \includegraphics[width = \textwidth]{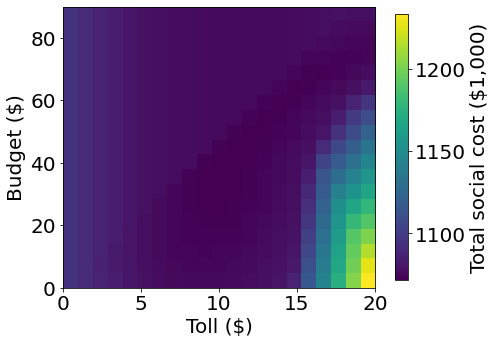}
    \caption{$(\lambda_E, \lambda_I, \lambda_T)=(12,1,1)$}
    \label{fig:12_1_1}
\end{subfigure}
\begin{subfigure}{.32\textwidth}
    \centering
    \includegraphics[width = \textwidth]{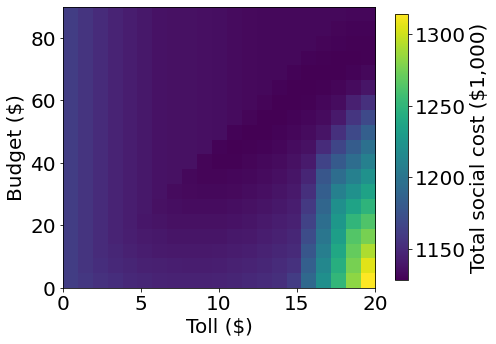}
    \caption{$(\lambda_E, \lambda_I, \lambda_T)=(15,1,1)$}
    \label{fig:15_1_1}
\end{subfigure}

    \caption{{\small \sf Objective function values (see Section~\ref{subsec:ObjectiveNumerical}) with varying pareto weights on the total travel costs of eligible and ineligible users and the negative total toll revenue}}
    \label{fig:pareto_1x}
\end{figure}
\fi

\vspace{-5pt}

\subsection{Policy Implications} \label{subsec:policy}

Our results have several implications for designing optimal CBCP schemes in practice, e.g., for the San Mateo 101 Express Lanes Project. First, since the optimal CBCP scheme can vary widely based on the central planner’s objective (see Table~\ref{tab:ParetoResults}), our results show that a principled approach using bi-level optimization is key to realizing the benefits of CBCP. Further, our experimental results suggest that CBCP schemes, at the expense of reduced toll revenues, can help alleviate the equity concerns of congestion pricing and this potential benefit to eligible users outweighs the negative impacts to ineligible users. In particular, compared to the optimal CBCP scheme for $\llambda = (1, 1, 1)$, the optimal CBCP scheme for $\llambda = (11,1,1)$ 
results in an equilibrium flow at which the toll revenues decrease by about 37\% while the average travel cost for eligible and ineligible users decreases by 10\% and about 1\%, respectively (see Figure~\ref{fig:pareto_1x} in Appendix~\ref{apdx:additionalExpResults}). The reduction in the average travel costs of the ineligible users follows due to the reduced travel time on the GP lanes, on which most ineligible users travel. Further, the optimal CBCP scheme corresponding to $\llambda = (11,1,1)$ amounts to a budget of \$15 every five days, which is a much higher credit allocation than the one-time \$100 allocation currently provided in the San Mateo Community Benefits Program. Hence, our results recommend a more frequent and recurring allocation of credits, e.g., every month, as this credit budget is likely to be used up in a matter of weeks or months, after which the equity benefits of CBCP will be foregone.

While the lowest-income users (i.e., those eligible to receive credits) stand to gain from CBCP, \ifarxiv we reiterate, as noted in section \ref{subsec:sensitivityNumerical}, that \fi the worst-off users \ifarxiv under a CBCP scheme \fi are the ineligible users with the lowest VoTs. \ifarxiv In particular, these \fi users incur higher travel costs due to the imposition of tolls that increases the travel time on GP lanes, yet typically do not have high enough VoTs to merit paying for the express lanes. Given that the lowest VoT users without credits will bear the brunt of \ifarxiv the travel time increases on GP lanes, \fi we believe that a multi-level budget structure, in which credits are allocated to users as some function of their income (and/or other eligibility criteria), may result in a more equitable distribution of the benefits of CBCP.



\ifarxiv
Beyond the aforementioned CBCP policy implications and recommendations of our work, the success of a CBCP scheme in achieving improved equity outcomes significantly hinges on the consideration of high-occupancy modes of travel, which is beyond the scope of this work. In particular, many express lanes, including those on the US 101 in San Mateo, are actually HOT lanes, in which HOVs \ifarxiv(i.e., carpools of two, three, or more persons, shuttles, buses) are also permitted to access the express lane.\else are also permitted access. \fi With free or discounted access for HOVs, tolls may also reduce the total flow in the network by incentivizing a shift from \ifarxiv single occupant vehicles (SOVs) \else SOVs \fi to HOVs. Depending on the prevailing rate of carpooling and willingness to adopt carpooling, the equilibrium tolls and budgets are likely to differ from those in the present study.

\fi

Finally, \ifarxiv as part of the Community Transportation Benefits Program of \fi the San Mateo 101 Express Lanes Project, eligible residents also have the choice of receiving \$100 in public transit fares instead of express lane credits. Such an option benefits low-income residents who may not be able to drive or are willing to forego driving (e.g., by using another mode), which may further reduce vehicle flows and amplify the benefits of the CBCP scheme. We defer the thorough examination of CBCP schemes in the context of HOVs and more complex credit schemes to future research.

\vspace{-7pt}

\section{Conclusion and Future Work} \label{sec:conclusion}

In this paper, we studied CBCP schemes, analogous to those implemented in the San Mateo 101 Express Lanes Project, to route heterogeneous users with different VoTs in a multi-lane highway. \ifarxiv Since CBCP schemes involve giving travel credits to a fraction of the users, we introduced a new mixed economy model, wherein eligible users receive travel credits while the remaining ineligible users spend money out-of-pocket to use the express lane. In this mixed economy setting, we first investigated the effect of CBCP schemes on the resulting traffic patterns by characterizing the properties of the corresponding Nash equilibria. In particular, we established the existence of CBCP equilibria and, in the setting when eligible users have time-invariant values of time, presented a convex program to compute CBCP equilibria. We further performed a comparative statics analysis to establish monotonicity relations between the equilibria induced by CBCP schemes in response to a change in the tolls set on an express lane or budgets distributed to eligible users. Next, we introduced a bi-level optimization framework to design optimal CBCP schemes to achieve particular societal objectives and presented a dense sampling approach to compute an approximation to the optimal CBCP scheme. Finally, we presented numerical experiments based on an application study of the San Mateo 101 Express Lanes Project, which demonstrated that a principled approach using bi-level optimization proposed in this work is key to realizing the benefits of CBCP schemes. \fi

\ifarxiv

There are several directions for future research. First, it would be worthwhile to investigate whether equilibria can be computed efficiently in the more general setting when eligible users' VoTs are time-varying. Next, it would be interesting to explore methods to improve the computational complexity of the dense sampling approach. For instance, as opposed to the general dense sampling approach developed in this paper, which is independent of the societal cost function $f$, using properties of $f$ combined with the comparative statics analysis results (see Section~\ref{sec:CompStatics}) can be used to narrow the search space, i.e., only a subset of points $\C_s$ need to be searched over. Furthermore, several model extensions, e.g., considering time-varying travel demand and the sensitivity of departure time choices to tolling schedules, are of interest to further mirror the real-world operation of express lanes. Moreover, including HOVs and incorporating mode choices would further improve understanding of the role of modal shift incentives in optimal CBCP. Lastly, investigating more general budget allocation structures, e.g., as a step-wise function of incomes, is a promising direction for further improving the equity of CBCP.

\fi







\bibliographystyle{ACM-Reference-Format}
\bibliography{main}


\begin{thebibliography}{76}


\ifx \showCODEN    \undefined \def \showCODEN     #1{\unskip}     \fi
\ifx \showDOI      \undefined \def \showDOI       #1{#1}\fi
\ifx \showISBNx    \undefined \def \showISBNx     #1{\unskip}     \fi
\ifx \showISBNxiii \undefined \def \showISBNxiii  #1{\unskip}     \fi
\ifx \showISSN     \undefined \def \showISSN      #1{\unskip}     \fi
\ifx \showLCCN     \undefined \def \showLCCN      #1{\unskip}     \fi
\ifx \shownote     \undefined \def \shownote      #1{#1}          \fi
\ifx \showarticletitle \undefined \def \showarticletitle #1{#1}   \fi
\ifx \showURL      \undefined \def \showURL       {\relax}        \fi
\providecommand\bibfield[2]{#2}
\providecommand\bibinfo[2]{#2}
\providecommand\natexlab[1]{#1}
\providecommand\showeprint[2][]{arXiv:#2}

\bibitem[\protect\citeauthoryear{Adler and Cetin}{Adler and Cetin}{2001}]%
        {ADLER2001447}
\bibfield{author}{\bibinfo{person}{Jeffrey~L. Adler} {and}
  \bibinfo{person}{Mecit Cetin}.} \bibinfo{year}{2001}\natexlab{}.
\newblock \showarticletitle{A direct redistribution model of congestion
  pricing}.
\newblock \bibinfo{journal}{\emph{Transportation Research Part B:
  Methodological}} \bibinfo{volume}{35}, \bibinfo{number}{5}
  (\bibinfo{year}{2001}), \bibinfo{pages}{447--460}.
\newblock


\bibitem[\protect\citeauthoryear{Akbarpour, Budish, Dworczak, and
  Kominers}{Akbarpour et~al\mbox{.}}{2022}]%
        {Vaccine-abdk}
\bibfield{author}{\bibinfo{person}{Mohammad Akbarpour}, \bibinfo{person}{Eric
  Budish}, \bibinfo{person}{Piotr Dworczak}, {and} \bibinfo{person}{Scott~Duke
  Kominers}.} \bibinfo{year}{2022}\natexlab{}.
\newblock \showarticletitle{An Economic Framework for Vaccine Prioritization}.
  In \bibinfo{booktitle}{\emph{Proceedings of the 23rd ACM Conference on
  Economics and Computation}} (Boulder, CO, USA) \emph{(\bibinfo{series}{EC
  '22})}. \bibinfo{publisher}{Association for Computing Machinery},
  \bibinfo{address}{New York, NY, USA}, \bibinfo{pages}{1181}.
\newblock
\showISBNx{9781450391504}
\urldef\tempurl%
\url{https://doi.org/10.1145/3490486.3538241}
\showDOI{\tempurl}


\bibitem[\protect\citeauthoryear{Akbarpour, Dworczak, and Kominers}{Akbarpour
  et~al\mbox{.}}{2020}]%
        {RAM-Akbarpour}
\bibfield{author}{\bibinfo{person}{Mohammad Akbarpour}, \bibinfo{person}{Piotr
  Dworczak}, {and} \bibinfo{person}{Scott~Duke Kominers}.}
  \bibinfo{year}{2020}\natexlab{}.
\newblock \bibinfo{booktitle}{\emph{{Redistributive allocation mechanisms}}}.
\newblock \bibinfo{type}{GRAPE Working Papers}~40. \bibinfo{institution}{GRAPE
  Group for Research in Applied Economics}.
\newblock
\urldef\tempurl%
\url{https://ideas.repec.org/p/fme/wpaper/40.html}
\showURL{%
\tempurl}


\bibitem[\protect\citeauthoryear{Alaei, Jalaly~Khalilabadi, and Tardos}{Alaei
  et~al\mbox{.}}{2017}]%
        {eq-matching-markets}
\bibfield{author}{\bibinfo{person}{Saeed Alaei}, \bibinfo{person}{Pooya
  Jalaly~Khalilabadi}, {and} \bibinfo{person}{Eva Tardos}.}
  \bibinfo{year}{2017}\natexlab{}.
\newblock \showarticletitle{Computing Equilibrium in Matching Markets}. In
  \bibinfo{booktitle}{\emph{Proceedings of the 2017 ACM Conference on Economics
  and Computation}} (Cambridge, Massachusetts, USA) \emph{(\bibinfo{series}{EC
  '17})}. \bibinfo{publisher}{Association for Computing Machinery},
  \bibinfo{address}{New York, NY, USA}, \bibinfo{pages}{245–261}.
\newblock
\showISBNx{9781450345279}
\urldef\tempurl%
\url{https://doi.org/10.1145/3033274.3085150}
\showDOI{\tempurl}


\bibitem[\protect\citeauthoryear{Arnott, de~Palma, and Lindsey}{Arnott
  et~al\mbox{.}}{1994}]%
        {arnott1994}
\bibfield{author}{\bibinfo{person}{Richard Arnott}, \bibinfo{person}{André de
  Palma}, {and} \bibinfo{person}{Robin Lindsey}.}
  \bibinfo{year}{1994}\natexlab{}.
\newblock \showarticletitle{The Welfare Effects of Congestion Tolls with
  Heterogeneous Commuters}.
\newblock \bibinfo{journal}{\emph{Journal of Transport Economics and Policy}}
  \bibinfo{volume}{28}, \bibinfo{number}{2} (\bibinfo{year}{1994}),
  \bibinfo{pages}{139--161}.
\newblock


\bibitem[\protect\citeauthoryear{Azizan~Ruhi, Dvijotham, Chen, and
  Wierman}{Azizan~Ruhi et~al\mbox{.}}{2018}]%
        {ruhi2018opportunities}
\bibfield{author}{\bibinfo{person}{Navid Azizan~Ruhi},
  \bibinfo{person}{Krishnamurthy Dvijotham}, \bibinfo{person}{Niangjun Chen},
  {and} \bibinfo{person}{Adam Wierman}.} \bibinfo{year}{2018}\natexlab{}.
\newblock \showarticletitle{Opportunities for price manipulation by aggregators
  in electricity markets}.
\newblock \bibinfo{journal}{\emph{IEEE Transactions on Smart Grid}}
  \bibinfo{volume}{9}, \bibinfo{number}{6} (\bibinfo{year}{2018}),
  \bibinfo{pages}{5687--5698}.
\newblock


\bibitem[\protect\citeauthoryear{Bertsimas, Farias, and Trichakis}{Bertsimas
  et~al\mbox{.}}{2011}]%
        {bertsimas-PoF}
\bibfield{author}{\bibinfo{person}{Dimitris Bertsimas},
  \bibinfo{person}{Vivek~F. Farias}, {and} \bibinfo{person}{Nikolaos
  Trichakis}.} \bibinfo{year}{2011}\natexlab{}.
\newblock \showarticletitle{The Price of Fairness}.
\newblock \bibinfo{journal}{\emph{Operations Research}} \bibinfo{volume}{59},
  \bibinfo{number}{1} (\bibinfo{year}{2011}), \bibinfo{pages}{17--31}.
\newblock
\urldef\tempurl%
\url{https://doi.org/10.1287/opre.1100.0865}
\showDOI{\tempurl}
\showeprint{https://doi.org/10.1287/opre.1100.0865}


\bibitem[\protect\citeauthoryear{Brotcorne, Labb\'{e}, Marcotte, and
  Savard}{Brotcorne et~al\mbox{.}}{2001}]%
        {bilevel-tolls}
\bibfield{author}{\bibinfo{person}{Luce Brotcorne}, \bibinfo{person}{Martine
  Labb\'{e}}, \bibinfo{person}{Patrice Marcotte}, {and} \bibinfo{person}{Gilles
  Savard}.} \bibinfo{year}{2001}\natexlab{}.
\newblock \showarticletitle{A Bilevel Model for Toll Optimization on a
  Multicommodity Transportation Network}.
\newblock \bibinfo{journal}{\emph{Transportation Science}}
  \bibinfo{volume}{35}, \bibinfo{number}{4} (\bibinfo{year}{2001}),
  \bibinfo{pages}{345--358}.
\newblock
\urldef\tempurl%
\url{https://doi.org/10.1287/trsc.35.4.345.10433}
\showDOI{\tempurl}
\showeprint{https://doi.org/10.1287/trsc.35.4.345.10433}


\bibitem[\protect\citeauthoryear{Budish}{Budish}{2011}]%
        {Budish}
\bibfield{author}{\bibinfo{person}{Eric Budish}.}
  \bibinfo{year}{2011}\natexlab{}.
\newblock \showarticletitle{The Combinatorial Assignment Problem: Approximate
  Competitive Equilibrium from Equal Incomes}.
\newblock \bibinfo{journal}{\emph{Journal of Political Economy}}
  \bibinfo{volume}{119}, \bibinfo{number}{6} (\bibinfo{year}{2011}),
  \bibinfo{pages}{1061--1103}.
\newblock
\urldef\tempurl%
\url{https://doi.org/10.1086/664613}
\showDOI{\tempurl}


\bibitem[\protect\citeauthoryear{Bureau}{Bureau}{2021}]%
        {ACS2021}
\bibfield{author}{\bibinfo{person}{US~Census Bureau}.}
  \bibinfo{year}{2021}\natexlab{}.
\newblock \bibinfo{title}{2021 American Community Survey (ACS) 1-Year Estimates
  - S2001: Earnings in the Past 12 Months (In 2021 Inflation-Adjusted
  Dollars)}.
\newblock
\newblock
\urldef\tempurl%
\url{https://data.census.gov/cedsci/table?q=S2001&g=0500000US06081,06085&tid=ACSST1Y2021.S2001}
\showURL{%
\tempurl}


\bibitem[\protect\citeauthoryear{Cano}{Cano}{2021}]%
        {CP-SF}
\bibfield{author}{\bibinfo{person}{Ricardo Cano}.}
  \bibinfo{year}{2021}\natexlab{}.
\newblock \bibinfo{title}{S.F. is considering downtown ‘congestion
  pricing.’ Here’s how much it would cost}.
\newblock \bibinfo{howpublished}{SF Chronicle}.
\newblock
\urldef\tempurl%
\url{https://www.sfchronicle.com/sf/article/S-F-is-considering-downtown-congestion-16336449.php}
\showURL{%
Retrieved October 14, 2022 from \tempurl}


\bibitem[\protect\citeauthoryear{Carlos F.~Pardo}{Carlos F.~Pardo}{2022}]%
        {BOGOTA}
\bibfield{author}{\bibinfo{person}{Frederic~Charlier Carlos F.~Pardo}.}
  \bibinfo{year}{2022}\natexlab{}.
\newblock \bibinfo{title}{Bogotá Tries to Make Congestion Pricing Flexible and
  Equitable}.
\newblock \bibinfo{howpublished}{Streets Blog USA}.
\newblock
\urldef\tempurl%
\url{https://usa.streetsblog.org/2022/10/17/bogota-tries-to-make-congestion-pricing-flexible-and-equitable/}
\showURL{%
Retrieved October 21, 2022 from \tempurl}


\bibitem[\protect\citeauthoryear{Chen and Varaiya}{Chen and Varaiya}{2002}]%
        {pems-database}
\bibfield{author}{\bibinfo{person}{Chao Chen} {and} \bibinfo{person}{Pravin
  Varaiya}.} \bibinfo{year}{2002}\natexlab{}.
\newblock \emph{\bibinfo{title}{Freeway Performance Measurement System
  (Pems)}}.
\newblock \bibinfo{thesistype}{Ph.D. Dissertation}.
\newblock
\newblock
\shownote{AAI3082138.}


\bibitem[\protect\citeauthoryear{Chien and Sinclair}{Chien and
  Sinclair}{2011}]%
        {CHIEN2011315}
\bibfield{author}{\bibinfo{person}{Steve Chien} {and} \bibinfo{person}{Alistair
  Sinclair}.} \bibinfo{year}{2011}\natexlab{}.
\newblock \showarticletitle{Convergence to approximate Nash equilibria in
  congestion games}.
\newblock \bibinfo{journal}{\emph{Games and Economic Behavior}}
  \bibinfo{volume}{71}, \bibinfo{number}{2} (\bibinfo{year}{2011}),
  \bibinfo{pages}{315--327}.
\newblock
\showISSN{0899-8256}
\urldef\tempurl%
\url{https://doi.org/10.1016/j.geb.2009.05.004}
\showDOI{\tempurl}


\bibitem[\protect\citeauthoryear{Cole, Dodis, and Roughgarden}{Cole
  et~al\mbox{.}}{2003}]%
        {heterogeneous-pricing-roughgarden}
\bibfield{author}{\bibinfo{person}{Richard Cole}, \bibinfo{person}{Yevgeniy
  Dodis}, {and} \bibinfo{person}{Tim Roughgarden}.}
  \bibinfo{year}{2003}\natexlab{}.
\newblock \showarticletitle{Pricing Network Edges for Heterogeneous Selfish
  Users}. In \bibinfo{booktitle}{\emph{Proceedings of the Thirty-Fifth Annual
  ACM Symposium on Theory of Computing}} (San Diego, CA, USA)
  \emph{(\bibinfo{series}{STOC '03})}. \bibinfo{publisher}{Association for
  Computing Machinery}, \bibinfo{address}{New York, NY, USA},
  \bibinfo{pages}{521–530}.
\newblock
\showISBNx{1581136749}
\urldef\tempurl%
\url{https://doi.org/10.1145/780542.780618}
\showDOI{\tempurl}


\bibitem[\protect\citeauthoryear{Colson, Marcotte, and Savard}{Colson
  et~al\mbox{.}}{2007}]%
        {colson2007overview}
\bibfield{author}{\bibinfo{person}{Beno{\^\i}t Colson},
  \bibinfo{person}{Patrice Marcotte}, {and} \bibinfo{person}{Gilles Savard}.}
  \bibinfo{year}{2007}\natexlab{}.
\newblock \showarticletitle{An overview of bilevel optimization}.
\newblock \bibinfo{journal}{\emph{Annals of operations research}}
  \bibinfo{volume}{153}, \bibinfo{number}{1} (\bibinfo{year}{2007}),
  \bibinfo{pages}{235--256}.
\newblock


\bibitem[\protect\citeauthoryear{Condorelli}{Condorelli}{2013}]%
        {CONDORELLI2013582}
\bibfield{author}{\bibinfo{person}{Daniele Condorelli}.}
  \bibinfo{year}{2013}\natexlab{}.
\newblock \showarticletitle{Market and non-market mechanisms for the optimal
  allocation of scarce resources}.
\newblock \bibinfo{journal}{\emph{Games and Economic Behavior}}
  \bibinfo{volume}{82} (\bibinfo{year}{2013}), \bibinfo{pages}{582--591}.
\newblock


\bibitem[\protect\citeauthoryear{Corporation}{Corporation}{2022}]%
        {bay-area-express-lane}
\bibfield{author}{\bibinfo{person}{Metropolitan~Transportation Corporation}.}
  \bibinfo{year}{2022}\natexlab{}.
\newblock \bibinfo{title}{Bay Area Express Lanes}.
\newblock
\newblock
\urldef\tempurl%
\url{https://mtc.ca.gov/operations/traveler-services/bay-area-express-lanes}
\showURL{%
\tempurl}


\bibitem[\protect\citeauthoryear{Dafermos}{Dafermos}{1972}]%
        {Dafermos-multiclass}
\bibfield{author}{\bibinfo{person}{Stella~C. Dafermos}.}
  \bibinfo{year}{1972}\natexlab{}.
\newblock \showarticletitle{The Traffic Assignment Problem for Multiclass-User
  Transportation Networks}.
\newblock \bibinfo{journal}{\emph{Transportation Science}} \bibinfo{volume}{6},
  \bibinfo{number}{1} (\bibinfo{year}{1972}), \bibinfo{pages}{73--87}.
\newblock
\urldef\tempurl%
\url{https://doi.org/10.1287/trsc.6.1.73}
\showDOI{\tempurl}
\showeprint{https://doi.org/10.1287/trsc.6.1.73}


\bibitem[\protect\citeauthoryear{Daganzo}{Daganzo}{1995}]%
        {DAGANZO1995139}
\bibfield{author}{\bibinfo{person}{Carlos~F. Daganzo}.}
  \bibinfo{year}{1995}\natexlab{}.
\newblock \showarticletitle{A pareto optimum congestion reduction scheme}.
\newblock \bibinfo{journal}{\emph{Transportation Research Part B:
  Methodological}} \bibinfo{volume}{29}, \bibinfo{number}{2}
  (\bibinfo{year}{1995}), \bibinfo{pages}{139--154}.
\newblock


\bibitem[\protect\citeauthoryear{Dempe, Kalashnikov, Prez-Valds, and
  Kalashnykova}{Dempe et~al\mbox{.}}{2015}]%
        {bi-level-hardness}
\bibfield{author}{\bibinfo{person}{Stephan Dempe}, \bibinfo{person}{Vyacheslav
  Kalashnikov}, \bibinfo{person}{Gerardo~A. Prez-Valds}, {and}
  \bibinfo{person}{Nataliya Kalashnykova}.} \bibinfo{year}{2015}\natexlab{}.
\newblock \bibinfo{booktitle}{\emph{Bilevel Programming Problems: Theory,
  Algorithms and Applications to Energy Networks}}.
\newblock \bibinfo{publisher}{Springer Publishing Company, Incorporated}.
\newblock
\showISBNx{3662458268}


\bibitem[\protect\citeauthoryear{Dyro, Schmerling, Arechiga, and Pavone}{Dyro
  et~al\mbox{.}}{2022}]%
        {dyro2022second}
\bibfield{author}{\bibinfo{person}{Robert Dyro}, \bibinfo{person}{Edward
  Schmerling}, \bibinfo{person}{Nikos Arechiga}, {and} \bibinfo{person}{Marco
  Pavone}.} \bibinfo{year}{2022}\natexlab{}.
\newblock \showarticletitle{Second-Order Sensitivity Analysis for Bilevel
  Optimization}. In \bibinfo{booktitle}{\emph{International Conference on
  Artificial Intelligence and Statistics}}. PMLR, \bibinfo{pages}{9166--9181}.
\newblock


\bibitem[\protect\citeauthoryear{Editor}{Editor}{2022}]%
        {CBCP-SanMateo}
\bibfield{author}{\bibinfo{person}{Editor}.} \bibinfo{year}{2022}\natexlab{}.
\newblock \bibinfo{title}{San Mateo 101 Express Lanes Open with a
  First-of-its-Kind Equity Program}.
\newblock \bibinfo{howpublished}{Everything South City}.
\newblock
\urldef\tempurl%
\url{https://everythingsouthcity.com/2022/04/san-mateo-101-express-lanes-open-with-a-first-of-its-kind-equity-program/}
\showURL{%
Retrieved April 27, 2022 from \tempurl}


\bibitem[\protect\citeauthoryear{Eliasson and Mattsson}{Eliasson and
  Mattsson}{2006}]%
        {ELIASSON2006602}
\bibfield{author}{\bibinfo{person}{Jonas Eliasson} {and}
  \bibinfo{person}{Lars-Göran Mattsson}.} \bibinfo{year}{2006}\natexlab{}.
\newblock \showarticletitle{Equity effects of congestion pricing: Quantitative
  methodology and a case study for Stockholm}.
\newblock \bibinfo{journal}{\emph{Transportation Research Part A: Policy and
  Practice}} \bibinfo{volume}{40}, \bibinfo{number}{7} (\bibinfo{year}{2006}),
  \bibinfo{pages}{602--620}.
\newblock
\showISSN{0965-8564}
\urldef\tempurl%
\url{https://doi.org/10.1016/j.tra.2005.11.002}
\showDOI{\tempurl}


\bibitem[\protect\citeauthoryear{Elokda, Cenedese, Zhang, Lygeros, and
  Dörfler}{Elokda et~al\mbox{.}}{2022}]%
        {ezzat-karma}
\bibfield{author}{\bibinfo{person}{Ezzat Elokda}, \bibinfo{person}{Carlo
  Cenedese}, \bibinfo{person}{Kenan Zhang}, \bibinfo{person}{John Lygeros},
  {and} \bibinfo{person}{Florian Dörfler}.} \bibinfo{year}{2022}\natexlab{}.
\newblock \bibinfo{title}{CARMA: Fair and efficient bottleneck congestion
  management with karma}.
\newblock
\newblock
\urldef\tempurl%
\url{https://doi.org/10.48550/ARXIV.2208.07113}
\showDOI{\tempurl}


\bibitem[\protect\citeauthoryear{Fabrikant, Papadimitriou, and
  Talwar}{Fabrikant et~al\mbox{.}}{2004}]%
        {fabrikant2004complexity}
\bibfield{author}{\bibinfo{person}{Alex Fabrikant}, \bibinfo{person}{Christos
  Papadimitriou}, {and} \bibinfo{person}{Kunal Talwar}.}
  \bibinfo{year}{2004}\natexlab{}.
\newblock \showarticletitle{The complexity of pure Nash equilibria}. In
  \bibinfo{booktitle}{\emph{Proceedings of the thirty-sixth annual ACM
  symposium on Theory of computing}}. \bibinfo{pages}{604--612}.
\newblock


\bibitem[\protect\citeauthoryear{Facchinei and Pang}{Facchinei and
  Pang}{2003}]%
        {facchinei2003finite}
\bibfield{author}{\bibinfo{person}{Francisco Facchinei} {and}
  \bibinfo{person}{Jong-Shi Pang}.} \bibinfo{year}{2003}\natexlab{}.
\newblock \bibinfo{booktitle}{\emph{Finite-dimensional variational inequalities
  and complementarity problems}}.
\newblock \bibinfo{publisher}{Springer}.
\newblock


\bibitem[\protect\citeauthoryear{{Fleischer}, {Jain}, and
  {Mahdian}}{{Fleischer} et~al\mbox{.}}{2004}]%
        {multicommodity-extension}
\bibfield{author}{\bibinfo{person}{Lisa {Fleischer}}, \bibinfo{person}{Kamal
  {Jain}}, {and} \bibinfo{person}{Mohammad {Mahdian}}.}
  \bibinfo{year}{2004}\natexlab{}.
\newblock \showarticletitle{Tolls for heterogeneous selfish users in
  multicommodity networks and generalized congestion games}. In
  \bibinfo{booktitle}{\emph{IEEE Symposium on Foundations of Computer
  Science}}.
\newblock
\urldef\tempurl%
\url{https://doi.org/10.1109/FOCS.2004.69}
\showDOI{\tempurl}


\bibitem[\protect\citeauthoryear{Gemici, Koutsoupias, Monnot, Papadimitriou,
  and Piliouras}{Gemici et~al\mbox{.}}{2019}]%
        {gemici_et_al:LIPIcs:2019:10270}
\bibfield{author}{\bibinfo{person}{Kurtulus Gemici}, \bibinfo{person}{Elias
  Koutsoupias}, \bibinfo{person}{Barnab{\'e} Monnot},
  \bibinfo{person}{Christos~H. Papadimitriou}, {and} \bibinfo{person}{Georgios
  Piliouras}.} \bibinfo{year}{2019}\natexlab{}.
\newblock \showarticletitle{{Wealth Inequality and the Price of Anarchy}}. In
  \bibinfo{booktitle}{\emph{International Symposium on Theoretical Aspects of
  Computer Science}} \emph{(\bibinfo{series}{Leibniz International Proceedings
  in Informatics}, Vol.~\bibinfo{volume}{126})}. \bibinfo{publisher}{Schloss
  Dagstuhl--Leibniz-Zentrum fuer Informatik}, \bibinfo{pages}{31:1--31:16}.
\newblock
\showISBNx{978-3-95977-100-9}
\showISSN{1868-8969}


\bibitem[\protect\citeauthoryear{G{\'o}mez-Ib{\'a}{\~n}ez and
  Small}{G{\'o}mez-Ib{\'a}{\~n}ez and Small}{1994}]%
        {gomez1994road}
\bibfield{author}{\bibinfo{person}{Jose~A G{\'o}mez-Ib{\'a}{\~n}ez} {and}
  \bibinfo{person}{Kenneth~A Small}.} \bibinfo{year}{1994}\natexlab{}.
\newblock \bibinfo{booktitle}{\emph{Road pricing for congestion management: A
  survey of international practice}}. Vol.~\bibinfo{volume}{210}.
\newblock \bibinfo{publisher}{Transportation Research Board}.
\newblock


\bibitem[\protect\citeauthoryear{Gorokh, Banerjee, and Iyer}{Gorokh
  et~al\mbox{.}}{2021}]%
        {gorokh2020nonmonetary}
\bibfield{author}{\bibinfo{person}{Artur Gorokh}, \bibinfo{person}{Siddhartha
  Banerjee}, {and} \bibinfo{person}{Krishnamurthy Iyer}.}
  \bibinfo{year}{2021}\natexlab{}.
\newblock \showarticletitle{From Monetary to Nonmonetary Mechanism Design via
  Artificial Currencies}.
\newblock \bibinfo{journal}{\emph{Mathematics of Operations Research}}
  \bibinfo{volume}{46}, \bibinfo{number}{3} (\bibinfo{year}{2021}),
  \bibinfo{pages}{835--855}.
\newblock
\urldef\tempurl%
\url{https://doi.org/10.1287/moor.2020.1098}
\showDOI{\tempurl}


\bibitem[\protect\citeauthoryear{Guo and Yang}{Guo and Yang}{2010}]%
        {GUO2010972}
\bibfield{author}{\bibinfo{person}{Xiaolei Guo} {and} \bibinfo{person}{Hai
  Yang}.} \bibinfo{year}{2010}\natexlab{}.
\newblock \showarticletitle{Pareto-improving congestion pricing and revenue
  refunding with multiple user classes}.
\newblock \bibinfo{journal}{\emph{Transportation Research Part B:
  Methodological}} \bibinfo{volume}{44}, \bibinfo{number}{8}
  (\bibinfo{year}{2010}), \bibinfo{pages}{972--982}.
\newblock


\bibitem[\protect\citeauthoryear{Hall}{Hall}{2018}]%
        {HALL2018113}
\bibfield{author}{\bibinfo{person}{Jonathan~D. Hall}.}
  \bibinfo{year}{2018}\natexlab{}.
\newblock \showarticletitle{Pareto improvements from Lexus Lanes: The effects
  of pricing a portion of the lanes on congested highways}.
\newblock \bibinfo{journal}{\emph{Journal of Public Economics}}
  \bibinfo{volume}{158} (\bibinfo{year}{2018}), \bibinfo{pages}{113--125}.
\newblock
\showISSN{0047-2727}
\urldef\tempurl%
\url{https://doi.org/10.1016/j.jpubeco.2018.01.003}
\showDOI{\tempurl}


\bibitem[\protect\citeauthoryear{Hannan}{Hannan}{2016}]%
        {Hannan+2016+97+140}
\bibfield{author}{\bibinfo{person}{James Hannan}.}
  \bibinfo{year}{2016}\natexlab{}.
\newblock \bibinfo{booktitle}{\emph{4. APPROXIMATION TO RAYES RISK IN REPEATED
  PLAY}}.
\newblock \bibinfo{publisher}{Princeton University Press},
  \bibinfo{address}{Princeton}, \bibinfo{pages}{97--140}.
\newblock
\urldef\tempurl%
\url{https://doi.org/doi:10.1515/9781400882151-006}
\showDOI{\tempurl}


\bibitem[\protect\citeauthoryear{Hatfield and Milgrom}{Hatfield and
  Milgrom}{2005}]%
        {hatfield2005matching}
\bibfield{author}{\bibinfo{person}{John~W. Hatfield} {and}
  \bibinfo{person}{Paul~R Milgrom}.} \bibinfo{year}{2005}\natexlab{}.
\newblock \showarticletitle{Matching with contracts}.
\newblock \bibinfo{journal}{\emph{American Economic Review}}
  \bibinfo{volume}{95}, \bibinfo{number}{4} (\bibinfo{year}{2005}),
  \bibinfo{pages}{913--935}.
\newblock


\bibitem[\protect\citeauthoryear{Hylland and Zeckhauser}{Hylland and
  Zeckhauser}{1979}]%
        {HZ}
\bibfield{author}{\bibinfo{person}{Aanund Hylland} {and}
  \bibinfo{person}{Richard Zeckhauser}.} \bibinfo{year}{1979}\natexlab{}.
\newblock \showarticletitle{The Efficient Allocation of Individuals to
  Positions}.
\newblock \bibinfo{journal}{\emph{Journal of Political Economy}}
  \bibinfo{volume}{87}, \bibinfo{number}{2} (\bibinfo{year}{1979}),
  \bibinfo{pages}{293--314}.
\newblock
\showISSN{00223808, 1537534X}
\urldef\tempurl%
\url{https://doi.org/10.1086/260757}
\showDOI{\tempurl}


\bibitem[\protect\citeauthoryear{Jaensirisak, Wardman, and May}{Jaensirisak
  et~al\mbox{.}}{2005}]%
        {CP-Low-acceptance}
\bibfield{author}{\bibinfo{person}{S. Jaensirisak}, \bibinfo{person}{M.
  Wardman}, {and} \bibinfo{person}{A.~D. May}.}
  \bibinfo{year}{2005}\natexlab{}.
\newblock \showarticletitle{Explaining Variations in Public Acceptability of
  Road Pricing Schemes}.
\newblock \bibinfo{journal}{\emph{Journal of Transport Economics and Policy}}
  \bibinfo{volume}{39}, \bibinfo{number}{2} (\bibinfo{year}{2005}),
  \bibinfo{pages}{127--153}.
\newblock


\bibitem[\protect\citeauthoryear{Jahn, Möhring, Schulz, and Stier-Moses}{Jahn
  et~al\mbox{.}}{2005}]%
        {so-routing-seminal}
\bibfield{author}{\bibinfo{person}{O. Jahn}, \bibinfo{person}{R. Möhring},
  \bibinfo{person}{A. Schulz}, {and} \bibinfo{person}{N. Stier-Moses}.}
  \bibinfo{year}{2005}\natexlab{}.
\newblock \showarticletitle{System-Optimal Routing of Traffic Flows with User
  Constraints in Networks with Congestion}.
\newblock \bibinfo{journal}{\emph{Operations Research}} \bibinfo{volume}{53},
  \bibinfo{number}{4} (\bibinfo{year}{2005}), \bibinfo{pages}{600--616}.
\newblock
\urldef\tempurl%
\url{https://doi.org/10.1287/opre.1040.0197}
\showDOI{\tempurl}


\bibitem[\protect\citeauthoryear{Jalota, Pavone, Qi, and Ye}{Jalota
  et~al\mbox{.}}{2021a}]%
        {jalota2021fisher}
\bibfield{author}{\bibinfo{person}{Devansh Jalota}, \bibinfo{person}{Marco
  Pavone}, \bibinfo{person}{Qi Qi}, {and} \bibinfo{person}{Yinyu Ye}.}
  \bibinfo{year}{2021}\natexlab{a}.
\newblock \bibinfo{title}{Fisher Markets with Linear Constraints: Equilibrium
  Properties and Efficient Distributed Algorithms}.
\newblock
\newblock
\showeprint[arxiv]{2106.10412}~[cs.GT]


\bibitem[\protect\citeauthoryear{Jalota, Solovey, Gopalakrishnan, Zoepf,
  Balakrishnan, and Pavone}{Jalota et~al\mbox{.}}{2021b}]%
        {jalota-acm-eaamo}
\bibfield{author}{\bibinfo{person}{Devansh Jalota}, \bibinfo{person}{Kiril
  Solovey}, \bibinfo{person}{Karthik Gopalakrishnan}, \bibinfo{person}{Stephen
  Zoepf}, \bibinfo{person}{Hamsa Balakrishnan}, {and} \bibinfo{person}{Marco
  Pavone}.} \bibinfo{year}{2021}\natexlab{b}.
\newblock \showarticletitle{When Efficiency Meets Equity in Congestion Pricing
  and Revenue Refunding Schemes}. In \bibinfo{booktitle}{\emph{Equity and
  Access in Algorithms, Mechanisms, and Optimization}}.
  \bibinfo{publisher}{Association for Computing Machinery},
  \bibinfo{address}{New York, NY, USA}, Article \bibinfo{articleno}{9},
  \bibinfo{numpages}{11}~pages.
\newblock
\showISBNx{9781450385534}
\urldef\tempurl%
\url{https://doi.org/10.1145/3465416.3483296}
\showURL{%
\tempurl}


\bibitem[\protect\citeauthoryear{Jalota, Solovey, Tsao, Zoepf, and
  Pavone}{Jalota et~al\mbox{.}}{2022}]%
        {jalota-balancing}
\bibfield{author}{\bibinfo{person}{Devansh Jalota}, \bibinfo{person}{Kiril
  Solovey}, \bibinfo{person}{Matthew Tsao}, \bibinfo{person}{Stephen Zoepf},
  {and} \bibinfo{person}{Marco Pavone}.} \bibinfo{year}{2022}\natexlab{}.
\newblock \showarticletitle{Balancing Fairness and Efficiency in Traffic
  Routing via Interpolated Traffic Assignment}. In
  \bibinfo{booktitle}{\emph{Proceedings of the 21st International Conference on
  Autonomous Agents and Multiagent Systems}} (Virtual Event, New Zealand)
  \emph{(\bibinfo{series}{AAMAS '22})}. \bibinfo{publisher}{International
  Foundation for Autonomous Agents and Multiagent Systems},
  \bibinfo{address}{Richland, SC}, \bibinfo{pages}{678–686}.
\newblock
\showISBNx{9781450392136}


\bibitem[\protect\citeauthoryear{Kash, Friedman, and Halpern}{Kash
  et~al\mbox{.}}{2007}]%
        {kash2007optimizing}
\bibfield{author}{\bibinfo{person}{Ian~A Kash}, \bibinfo{person}{Eric~J
  Friedman}, {and} \bibinfo{person}{Joseph~Y Halpern}.}
  \bibinfo{year}{2007}\natexlab{}.
\newblock \showarticletitle{Optimizing scrip systems: Efficiency, crashes,
  hoarders, and altruists}. In \bibinfo{booktitle}{\emph{Proceedings of the 8th
  ACM conference on Electronic commerce}}. \bibinfo{pages}{305--315}.
\newblock


\bibitem[\protect\citeauthoryear{Kelso and Crawford}{Kelso and
  Crawford}{1982}]%
        {kelso1982job}
\bibfield{author}{\bibinfo{person}{Alexander~S. Kelso} {and}
  \bibinfo{person}{Vincent~P. Crawford}.} \bibinfo{year}{1982}\natexlab{}.
\newblock \showarticletitle{Job Matching, Coalition Formation, and Gross
  Substitutes}.
\newblock \bibinfo{journal}{\emph{Econometrica}} \bibinfo{volume}{50},
  \bibinfo{number}{6} (\bibinfo{year}{1982}), \bibinfo{pages}{1483--1504}.
\newblock
\showISSN{00129682, 14680262}
\urldef\tempurl%
\url{http://www.jstor.org/stable/1913392}
\showURL{%
\tempurl}


\bibitem[\protect\citeauthoryear{Kinderlehrer and Stampacchia}{Kinderlehrer and
  Stampacchia}{2000}]%
        {variational-book-2000}
\bibfield{author}{\bibinfo{person}{David Kinderlehrer} {and}
  \bibinfo{person}{Guido Stampacchia}.} \bibinfo{year}{2000}\natexlab{}.
\newblock \bibinfo{booktitle}{\emph{An Introduction to Variational Inequalities
  and Their Applications}}.
\newblock \bibinfo{publisher}{Society for Industrial and Applied Mathematics}.
\newblock
\urldef\tempurl%
\url{https://doi.org/10.1137/1.9780898719451}
\showDOI{\tempurl}
\showeprint{https://epubs.siam.org/doi/pdf/10.1137/1.9780898719451}


\bibitem[\protect\citeauthoryear{Kockelman and Kalmanje}{Kockelman and
  Kalmanje}{2005}]%
        {KOCKELMAN2005671}
\bibfield{author}{\bibinfo{person}{Kara~M. Kockelman} {and}
  \bibinfo{person}{Sukumar Kalmanje}.} \bibinfo{year}{2005}\natexlab{}.
\newblock \showarticletitle{Credit-based congestion pricing: a policy proposal
  and the public’s response}.
\newblock \bibinfo{journal}{\emph{Transportation Research Part A: Policy and
  Practice}} \bibinfo{volume}{39}, \bibinfo{number}{7} (\bibinfo{year}{2005}),
  \bibinfo{pages}{671--690}.
\newblock
\showISSN{0965-8564}
\urldef\tempurl%
\url{https://doi.org/10.1016/j.tra.2005.02.014}
\showDOI{\tempurl}


\bibitem[\protect\citeauthoryear{Kreps}{Kreps}{2012}]%
        {kreps-book}
\bibfield{author}{\bibinfo{person}{David~M. Kreps}.}
  \bibinfo{year}{2012}\natexlab{}.
\newblock \bibinfo{booktitle}{\emph{{Microeconomic Foundations I: Choice and
  Competitive Markets}}}.
\newblock Number 9890 in \bibinfo{series}{Economics Books}.
  \bibinfo{publisher}{Princeton University Press}.
\newblock


\bibitem[\protect\citeauthoryear{Labb\'{e}, Marcotte, and Savard}{Labb\'{e}
  et~al\mbox{.}}{1998}]%
        {bilevel-labbe}
\bibfield{author}{\bibinfo{person}{Martine Labb\'{e}}, \bibinfo{person}{Patrice
  Marcotte}, {and} \bibinfo{person}{Gilles Savard}.}
  \bibinfo{year}{1998}\natexlab{}.
\newblock \showarticletitle{A Bilevel Model of Taxation and Its Application to
  Optimal Highway Pricing}.
\newblock \bibinfo{journal}{\emph{Management Science}} \bibinfo{volume}{44},
  \bibinfo{number}{12-part-1} (\bibinfo{year}{1998}),
  \bibinfo{pages}{1608--1622}.
\newblock
\urldef\tempurl%
\url{https://doi.org/10.1287/mnsc.44.12.1608}
\showDOI{\tempurl}
\showeprint{https://doi.org/10.1287/mnsc.44.12.1608}


\bibitem[\protect\citeauthoryear{Lanes}{Lanes}{2022}]%
        {I680Quarterly}
\bibfield{author}{\bibinfo{person}{Bay Area~Express Lanes}.}
  \bibinfo{year}{2022}\natexlab{}.
\newblock \bibinfo{booktitle}{\emph{I-680 Contra Costa Express Lanes
  Performance Report - 1st Quarter 2022: January-March}}.
\newblock \bibinfo{type}{{T}echnical {R}eport}.
  \bibinfo{institution}{Metropolitan Transportation Comission}.
\newblock
\urldef\tempurl%
\url{https://mtc.ca.gov/sites/default/files/documents/2022-08/Q1-2022-BAIFA-Express-Lanes-Quarterly-Performance-Report.pdf}
\showURL{%
\tempurl}


\bibitem[\protect\citeauthoryear{Langmyhr}{Langmyhr}{1999}]%
        {langmyhr1999understanding}
\bibfield{author}{\bibinfo{person}{Tore Langmyhr}.}
  \bibinfo{year}{1999}\natexlab{}.
\newblock \showarticletitle{Understanding innovation: the case of road
  pricing}.
\newblock \bibinfo{journal}{\emph{Transport Reviews}} \bibinfo{volume}{19},
  \bibinfo{number}{3} (\bibinfo{year}{1999}), \bibinfo{pages}{255--271}.
\newblock


\bibitem[\protect\citeauthoryear{Larsson and Patriksson}{Larsson and
  Patriksson}{1998}]%
        {Larsson1998}
\bibfield{author}{\bibinfo{person}{Torbj{\"o}rn Larsson} {and}
  \bibinfo{person}{Michael Patriksson}.} \bibinfo{year}{1998}\natexlab{}.
\newblock \bibinfo{booktitle}{\emph{Side Constrained Traffic Equilibrium
  Models---Traffic Management Through Link Tolls}}.
\newblock \bibinfo{publisher}{Springer US}, \bibinfo{address}{Boston, MA},
  \bibinfo{pages}{125--151}.
\newblock
\showISBNx{978-1-4615-5757-9}
\urldef\tempurl%
\url{https://doi.org/10.1007/978-1-4615-5757-9_7}
\showDOI{\tempurl}


\bibitem[\protect\citeauthoryear{Levinson}{Levinson}{2010}]%
        {Levinson-equityeffects}
\bibfield{author}{\bibinfo{person}{David Levinson}.}
  \bibinfo{year}{2010}\natexlab{}.
\newblock \showarticletitle{Equity Effects of Road Pricing: A Review}.
\newblock \bibinfo{journal}{\emph{Transport Reviews}} \bibinfo{volume}{30},
  \bibinfo{number}{1} (\bibinfo{year}{2010}), \bibinfo{pages}{33--57}.
\newblock
\urldef\tempurl%
\url{https://doi.org/10.1080/01441640903189304}
\showDOI{\tempurl}
\showeprint{https://doi.org/10.1080/01441640903189304}


\bibitem[\protect\citeauthoryear{Li, Kockelman, and Huang}{Li
  et~al\mbox{.}}{2021}]%
        {LiCBCPAustin}
\bibfield{author}{\bibinfo{person}{Weijia~(Vivian) Li},
  \bibinfo{person}{Kara~M. Kockelman}, {and} \bibinfo{person}{Yantao Huang}.}
  \bibinfo{year}{2021}\natexlab{}.
\newblock \showarticletitle{Traffic and Welfare Impacts of Credit-Based
  Congestion Pricing Applications: An Austin Case Study}.
\newblock \bibinfo{journal}{\emph{Transportation Research Record}}
  \bibinfo{volume}{2675}, \bibinfo{number}{1} (\bibinfo{year}{2021}),
  \bibinfo{pages}{10--24}.
\newblock
\urldef\tempurl%
\url{https://doi.org/10.1177/0361198120960139}
\showDOI{\tempurl}
\showeprint{https://doi.org/10.1177/0361198120960139}


\bibitem[\protect\citeauthoryear{Litman}{Litman}{2003}]%
        {litman2003london}
\bibfield{author}{\bibinfo{person}{T Litman}.} \bibinfo{year}{2003}\natexlab{}.
\newblock \bibinfo{title}{London Congestion Pricing: Implications for Other
  Cities. Victoria Transport Policy Institute, Victoria, BC, Canada}.
\newblock
\newblock


\bibitem[\protect\citeauthoryear{Nagurney}{Nagurney}{2000}]%
        {NAGURNEY2000393}
\bibfield{author}{\bibinfo{person}{A. Nagurney}.}
  \bibinfo{year}{2000}\natexlab{}.
\newblock \showarticletitle{A multiclass, multicriteria traffic network
  equilibrium model}.
\newblock \bibinfo{journal}{\emph{Mathematical and Computer Modelling}}
  \bibinfo{volume}{32}, \bibinfo{number}{3} (\bibinfo{year}{2000}),
  \bibinfo{pages}{393--411}.
\newblock
\showISSN{0895-7177}
\urldef\tempurl%
\url{https://doi.org/10.1016/S0895-7177(00)00142-4}
\showDOI{\tempurl}


\bibitem[\protect\citeauthoryear{Nie}{Nie}{2015}]%
        {NieCBCP2015}
\bibfield{author}{\bibinfo{person}{Yu~Marco Nie}.}
  \bibinfo{year}{2015}\natexlab{}.
\newblock \showarticletitle{A New Tradable Credit Scheme for the Morning
  Commute Problem}.
\newblock \bibinfo{journal}{\emph{Networks and Spatial Economics}}
  \bibinfo{volume}{15} (\bibinfo{year}{2015}), \bibinfo{pages}{719--741}.
\newblock
\urldef\tempurl%
\url{https://doi.org/10.1007/s11067-013-9192-8}
\showDOI{\tempurl}


\bibitem[\protect\citeauthoryear{of~Public~Roads}{of~Public~Roads}{1964}]%
        {utraffic}
\bibfield{author}{\bibinfo{person}{Bureau of Public~Roads}.}
  \bibinfo{year}{1964}\natexlab{}.
\newblock \bibinfo{booktitle}{\emph{Traffic assignment manual}}.
\newblock \bibinfo{publisher}{Technical report, U.S. Dept. of Commerce, Urban
  Planning Division}.
\newblock
\showISBNx{9780598494856}
\showLCCN{65062808}
\urldef\tempurl%
\url{https://books.google.com/books?id=t9XhugEACAAJ}
\showURL{%
\tempurl}


\bibitem[\protect\citeauthoryear{Patriksson and Rockafellar}{Patriksson and
  Rockafellar}{2002}]%
        {bilevel-patricksson}
\bibfield{author}{\bibinfo{person}{Michael Patriksson} {and}
  \bibinfo{person}{R.~Tyrrell Rockafellar}.} \bibinfo{year}{2002}\natexlab{}.
\newblock \showarticletitle{A Mathematical Model and Descent Algorithm for
  Bilevel Traffic Management}.
\newblock \bibinfo{journal}{\emph{Transportation Science}}
  \bibinfo{volume}{36}, \bibinfo{number}{3} (\bibinfo{year}{2002}),
  \bibinfo{pages}{271--291}.
\newblock
\urldef\tempurl%
\url{https://doi.org/10.1287/trsc.36.3.271.7826}
\showDOI{\tempurl}
\showeprint{https://doi.org/10.1287/trsc.36.3.271.7826}


\bibitem[\protect\citeauthoryear{Patterson and Levinson}{Patterson and
  Levinson}{2008}]%
        {patterson2008lexus}
\bibfield{author}{\bibinfo{person}{Tyler Patterson} {and}
  \bibinfo{person}{David~M Levinson}.} \bibinfo{year}{2008}\natexlab{}.
\newblock \showarticletitle{Lexus lanes or corolla lanes? Spatial use and
  equity patterns on the I-394 MnPASS lanes}.
\newblock  (\bibinfo{year}{2008}).
\newblock


\bibitem[\protect\citeauthoryear{Paybarah}{Paybarah}{2019}]%
        {nyt-cp}
\bibfield{author}{\bibinfo{person}{Azi Paybarah}.}
  \bibinfo{year}{2019}\natexlab{}.
\newblock \bibinfo{title}{Congestion Pricing: Mass Transit Savior or Tax on the
  Working Class?}
\newblock \bibinfo{howpublished}{New York Times}.
\newblock


\bibitem[\protect\citeauthoryear{Poole}{Poole}{2000}]%
        {Regulation}
\bibfield{author}{\bibinfo{person}{C.~Kenneth Poole, Robert W. Jr.~Orski}.}
  \bibinfo{year}{2000}\natexlab{}.
\newblock \showarticletitle{HOT Lanes: A Better Way to Attack Urban Highway
  Congestion Transportation}.
\newblock \bibinfo{journal}{\emph{Regulation}}  \bibinfo{volume}{23}
  (\bibinfo{year}{2000}), \bibinfo{pages}{15}.
\newblock


\bibitem[\protect\citeauthoryear{Prendergast}{Prendergast}{2017}]%
        {prendergast-2017}
\bibfield{author}{\bibinfo{person}{Canice Prendergast}.}
  \bibinfo{year}{2017}\natexlab{}.
\newblock \showarticletitle{How Food Banks Use Markets to Feed the Poor}.
\newblock \bibinfo{journal}{\emph{Journal of Economic Perspectives}}
  \bibinfo{volume}{31}, \bibinfo{number}{4} (\bibinfo{date}{November}
  \bibinfo{year}{2017}), \bibinfo{pages}{145--62}.
\newblock
\urldef\tempurl%
\url{https://doi.org/10.1257/jep.31.4.145}
\showDOI{\tempurl}


\bibitem[\protect\citeauthoryear{Roughgarden}{Roughgarden}{2005}]%
        {self-routing-PoA}
\bibfield{author}{\bibinfo{person}{Tim Roughgarden}.}
  \bibinfo{year}{2005}\natexlab{}.
\newblock \bibinfo{booktitle}{\emph{Selfish Routing and the Price of Anarchy}}.
\newblock \bibinfo{publisher}{The MIT Press}.
\newblock
\showISBNx{0262182432}


\bibitem[\protect\citeauthoryear{Roughgarden and Tardos}{Roughgarden and
  Tardos}{2002}]%
        {how-bad-is-selfish}
\bibfield{author}{\bibinfo{person}{Tim Roughgarden} {and}
  \bibinfo{person}{\'{E}va Tardos}.} \bibinfo{year}{2002}\natexlab{}.
\newblock \showarticletitle{How Bad is Selfish Routing?}
\newblock \bibinfo{journal}{\emph{J. ACM}} \bibinfo{volume}{49},
  \bibinfo{number}{2} (\bibinfo{date}{March} \bibinfo{year}{2002}),
  \bibinfo{pages}{236–259}.
\newblock
\showISSN{0004-5411}
\urldef\tempurl%
\url{https://doi.org/10.1145/506147.506153}
\showDOI{\tempurl}


\bibitem[\protect\citeauthoryear{Saez and Stantcheva}{Saez and
  Stantcheva}{2016}]%
        {Stantcheva-2016}
\bibfield{author}{\bibinfo{person}{Emmanuel Saez} {and}
  \bibinfo{person}{Stefanie Stantcheva}.} \bibinfo{year}{2016}\natexlab{}.
\newblock \showarticletitle{Generalized Social Marginal Welfare Weights for
  Optimal Tax Theory}.
\newblock \bibinfo{journal}{\emph{American Economic Review}}
  \bibinfo{volume}{106}, \bibinfo{number}{1} (\bibinfo{date}{January}
  \bibinfo{year}{2016}), \bibinfo{pages}{24--45}.
\newblock
\urldef\tempurl%
\url{https://doi.org/10.1257/aer.20141362}
\showDOI{\tempurl}


\bibitem[\protect\citeauthoryear{Salazar, Tsao, Aguiar, Schiffer, and
  Pavone}{Salazar et~al\mbox{.}}{2019}]%
        {SalazarTsaoEtAl2019}
\bibfield{author}{\bibinfo{person}{M. Salazar}, \bibinfo{person}{M. Tsao},
  \bibinfo{person}{I. Aguiar}, \bibinfo{person}{M. Schiffer}, {and}
  \bibinfo{person}{M. Pavone}.} \bibinfo{year}{2019}\natexlab{}.
\newblock \showarticletitle{A Congestion-aware Routing Scheme for Autonomous
  Mobility-on-Demand Systems}. In \bibinfo{booktitle}{\emph{{European Control
  Conference}}}. \bibinfo{address}{Naples, Italy}.
\newblock
\urldef\tempurl%
\url{/wp-content/papercite-data/pdf/Salazar.Tsao.ea.ECC19.pdf}
\showURL{%
\tempurl}


\bibitem[\protect\citeauthoryear{Sheffi}{Sheffi}{1985}]%
        {Sheffi1985}
\bibfield{author}{\bibinfo{person}{Y. Sheffi}.}
  \bibinfo{year}{1985}\natexlab{}.
\newblock \bibinfo{booktitle}{\emph{Urban Transportation Networks: Equilibrium
  Analysis with Mathematical Programming Methods}}.
\newblock \bibinfo{publisher}{Prentice-Hall}.
\newblock


\bibitem[\protect\citeauthoryear{Small}{Small}{1992}]%
        {small1992using}
\bibfield{author}{\bibinfo{person}{Kenneth~A Small}.}
  \bibinfo{year}{1992}\natexlab{}.
\newblock \showarticletitle{Using the revenues from congestion pricing}.
\newblock \bibinfo{journal}{\emph{Transportation}} \bibinfo{volume}{19},
  \bibinfo{number}{4} (\bibinfo{year}{1992}), \bibinfo{pages}{359--381}.
\newblock


\bibitem[\protect\citeauthoryear{Verhoef}{Verhoef}{2002}]%
        {VERHOEF2002281}
\bibfield{author}{\bibinfo{person}{Erik~T. Verhoef}.}
  \bibinfo{year}{2002}\natexlab{}.
\newblock \showarticletitle{Second-best congestion pricing in general static
  transportation networks with elastic demands}.
\newblock \bibinfo{journal}{\emph{Regional Science and Urban Economics}}
  \bibinfo{volume}{32}, \bibinfo{number}{3} (\bibinfo{year}{2002}),
  \bibinfo{pages}{281--310}.
\newblock
\showISSN{0166-0462}
\urldef\tempurl%
\url{https://doi.org/10.1016/S0166-0462(00)00064-8}
\showDOI{\tempurl}


\bibitem[\protect\citeauthoryear{Vickrey}{Vickrey}{1969}]%
        {vickrey1969congestion}
\bibfield{author}{\bibinfo{person}{William~S Vickrey}.}
  \bibinfo{year}{1969}\natexlab{}.
\newblock \showarticletitle{Congestion theory and transport investment}.
\newblock \bibinfo{journal}{\emph{The American Economic Review}}
  \bibinfo{volume}{59}, \bibinfo{number}{2} (\bibinfo{year}{1969}),
  \bibinfo{pages}{251--260}.
\newblock


\bibitem[\protect\citeauthoryear{Wang, Yang, Zhu, and Li}{Wang
  et~al\mbox{.}}{2012}]%
        {WANG2012426}
\bibfield{author}{\bibinfo{person}{Xiaolei Wang}, \bibinfo{person}{Hai Yang},
  \bibinfo{person}{Daoli Zhu}, {and} \bibinfo{person}{Changmin Li}.}
  \bibinfo{year}{2012}\natexlab{}.
\newblock \showarticletitle{Tradable travel credits for congestion management
  with heterogeneous users}.
\newblock \bibinfo{journal}{\emph{Transportation Research Part E: Logistics and
  Transportation Review}} \bibinfo{volume}{48}, \bibinfo{number}{2}
  (\bibinfo{year}{2012}), \bibinfo{pages}{426--437}.
\newblock
\showISSN{1366-5545}
\urldef\tempurl%
\url{https://doi.org/10.1016/j.tre.2011.10.007}
\showDOI{\tempurl}


\bibitem[\protect\citeauthoryear{Weitzman}{Weitzman}{1977}]%
        {weitzman-seminal}
\bibfield{author}{\bibinfo{person}{Martin~L. Weitzman}.}
  \bibinfo{year}{1977}\natexlab{}.
\newblock \showarticletitle{Is the Price System or Rationing More Efficient in
  Getting a Commodity to Those Who Need it Most?}
\newblock \bibinfo{journal}{\emph{Bell Journal of Economics}}
  \bibinfo{volume}{8}, \bibinfo{number}{2} (\bibinfo{year}{1977}),
  \bibinfo{pages}{517--524}.
\newblock


\bibitem[\protect\citeauthoryear{Wu, Yin, Lawphongpanich, and Yang}{Wu
  et~al\mbox{.}}{2012}]%
        {WU20121273}
\bibfield{author}{\bibinfo{person}{Di Wu}, \bibinfo{person}{Yafeng Yin},
  \bibinfo{person}{Siriphong Lawphongpanich}, {and} \bibinfo{person}{Hai
  Yang}.} \bibinfo{year}{2012}\natexlab{}.
\newblock \showarticletitle{Design of more equitable congestion pricing and
  tradable credit schemes for multimodal transportation networks}.
\newblock \bibinfo{journal}{\emph{Transportation Research Part B:
  Methodological}} \bibinfo{volume}{46}, \bibinfo{number}{9}
  (\bibinfo{year}{2012}), \bibinfo{pages}{1273--1287}.
\newblock


\bibitem[\protect\citeauthoryear{Xiao, Huang, and Liu}{Xiao
  et~al\mbox{.}}{2015}]%
        {XiaoTradableCBCP}
\bibfield{author}{\bibinfo{person}{Ling-Ling Xiao}, \bibinfo{person}{Hai-Jun
  Huang}, {and} \bibinfo{person}{Ronghui Liu}.}
  \bibinfo{year}{2015}\natexlab{}.
\newblock \showarticletitle{Tradable credit scheme for rush hour travel choice
  with heterogeneous commuters}.
\newblock \bibinfo{journal}{\emph{Advances in Mechanical Engineering}}
  \bibinfo{volume}{7}, \bibinfo{number}{10} (\bibinfo{year}{2015}),
  \bibinfo{pages}{1687814015612430}.
\newblock
\urldef\tempurl%
\url{https://doi.org/10.1177/1687814015612430}
\showDOI{\tempurl}
\showeprint{https://doi.org/10.1177/1687814015612430}


\bibitem[\protect\citeauthoryear{Yang and Huang}{Yang and Huang}{2004}]%
        {YANG20041}
\bibfield{author}{\bibinfo{person}{Hai Yang} {and} \bibinfo{person}{Hai-Jun
  Huang}.} \bibinfo{year}{2004}\natexlab{}.
\newblock \showarticletitle{The multi-class, multi-criteria traffic network
  equilibrium and systems optimum problem}.
\newblock \bibinfo{journal}{\emph{Transportation Research Part B:
  Methodological}} \bibinfo{volume}{38}, \bibinfo{number}{1}
  (\bibinfo{year}{2004}), \bibinfo{pages}{1--15}.
\newblock
\showISSN{0191-2615}
\urldef\tempurl%
\url{https://doi.org/10.1016/S0191-2615(02)00074-7}
\showDOI{\tempurl}


\bibitem[\protect\citeauthoryear{Yang and Wang}{Yang and Wang}{2011a}]%
        {YangCBCP2011}
\bibfield{author}{\bibinfo{person}{Hai Yang} {and} \bibinfo{person}{Xiaolei
  Wang}.} \bibinfo{year}{2011}\natexlab{a}.
\newblock \showarticletitle{Managing network mobility with tradable credits}.
\newblock \bibinfo{journal}{\emph{Transportation Research Part B:
  Methodological}} \bibinfo{volume}{45}, \bibinfo{number}{3}
  (\bibinfo{year}{2011}), \bibinfo{pages}{580--594}.
\newblock
\showISSN{0191-2615}
\urldef\tempurl%
\url{https://doi.org/10.1016/j.trb.2010.10.002}
\showDOI{\tempurl}


\bibitem[\protect\citeauthoryear{Yang and Wang}{Yang and Wang}{2011b}]%
        {yang-tradable-credits}
\bibfield{author}{\bibinfo{person}{Hai Yang} {and} \bibinfo{person}{Xiaolei
  Wang}.} \bibinfo{year}{2011}\natexlab{b}.
\newblock \showarticletitle{Managing network mobility with tradable credits}.
\newblock \bibinfo{journal}{\emph{Transportation Research Part B:
  Methodological}} \bibinfo{volume}{45}, \bibinfo{number}{3}
  (\bibinfo{year}{2011}), \bibinfo{pages}{580--594}.
\newblock
\urldef\tempurl%
\url{https://EconPapers.repec.org/RePEc:eee:transb:v:45:y:2011:i:3:p:580-594}
\showURL{%
\tempurl}


\end{thebibliography}


\appendix

\section{Further Discussion of CBCP Equilibrium Notion} \label{apdx:cbcp-eq-def}

The notion of CBCP equilibria, introduced in Definition~\ref{def:MainNashEq} captures the travel decisions of both eligible and ineligible users. However, we do remark that, as with prior equilibrium notions in non-atomic congestion games~\cite{heterogeneous-pricing-roughgarden,multicommodity-extension,YANG20041}, CBCP equilibria requires certain informational assumptions to be realized. For instance, all users require information on the resulting edge flows to solve their respective individual optimization problems, i.e., Problem~\eqref{eq:objIOPIn}-\eqref{eq:con2IOpIn} for the ineligible users and Problem~\eqref{eq:objIOP}-\eqref{eq:con3IOp} for the eligible users. Furthermore, due to the coupling of the eligible users' travel decisions across periods, eligible users additionally require information on their values of time over the $T$ periods to make informed travel decisions. As a result, from the perspective of the eligible users, the Nash equilibrium notion in Definition~\ref{def:MainNashEq} can be interpreted as an equilibrium concept wherein eligible users can look back in hindsight and, based on the realized edge flows $\x$ and their values of time over the $T$ periods, determine if they could have benefited through a unilateral deviation by using the express lane in a different set of periods. We reiterate that no such coupling of travel decisions across periods holds for ineligible users. 

Despite the need for certain informational assumptions, we note that the equilibrium notion introduced in Definition~\ref{def:MainNashEq} provides a reasonable approximation into how both groups of users would make travel decisions as it captures the objectives and preferences of both eligible and ineligible users. Thus, as with prior general equilibrium studies, we focus our attention on fundamentally characterizing properties of the Nash equilibrium concept in Definition~\ref{def:MainNashEq} to understand the influence of CBCP schemes on traffic outcomes in practice. The question of whether CBCP equilibria can be realized is beyond the scope of this work, and we refer to~\cite{CHIEN2011315,Hannan+2016+97+140} for a discussion of methods used to study whether the players of a game reach an equilibrium.

Finally, we also note that without loss of generality it suffices to focus on CBCP equilibrium flows $\y$ such that for eligible (ineligible) user groups $g \in \G_E$ ($g \in \G_I$), $y_{e,t}^g = z_{e,t}^{g*}$ for some optimal solution $\z^*$ to Problem~\eqref{eq:objIOP}-\eqref{eq:con3IOp} (Problem~\eqref{eq:objIOPIn}-\eqref{eq:con2IOpIn}), as all users in a given group incur the same travel cost at any CBCP equilibrium. 

\section{Definitions} \label{apdx:defs}

In this section, we present the definitions of locally bounded, upper semi-continuous, and lower semi-continuous correspondences.

\begin{definition} [Locally Bounded Correspondence] \label{def:locallyBdd}
A correspondence $\phi: X \implies Y$ is locally bounded if for ever $x \in X$, there exists an $\epsilon(x)>0$ and a bounded set $Y(x) \subseteq Y$ such that $\phi(x') \subseteq Y(x)$ for all $x'$ that are less than an $\epsilon$ distance from $x$.
\end{definition}

\begin{definition} [Upper Semi-Continuous Correspondence] \label{dec:UpperSemi}
A correspondence $\phi: X \implies Y$ is upper semi-continuous if, whenever $\{ x^n\}$ is a sequence in $X$ with limit $x$ and $\{ y^n \}$ is a sequence in $Y$ such that $y^n \in \phi(x^n)$ for all $n$ and $\lim_{n} y^n$ exists, then this limit point is an element of $\phi(x)$.
\end{definition}

\begin{definition} [Lower Semi-Continuous Correspondence] \label{def:LowerSemi}
A correspondence $\phi: X \implies Y$ is lower semi-continuous if for every $x \in X$, the sequence $\{ x^n \}$ in $X$ with limit $x$, and $y \in \phi(x)$, we can find for all sufficiently large $n$, i.e., all $n>N$ for some large $N$, $y^n \in \phi(x^n)$ such that $\lim_{n} y^n = y$.
\end{definition}

\section{Proofs}

\subsection{Proof of Lemma~\ref{lem:optEligible}} \label{apdx:OptEligiblePf}

To derive the optimal solution of Problem~\eqref{eq:objIOP}-\eqref{eq:con3IOp}, we first formulate its Lagrangian, which can be written as follows
\begin{align*}
    \mathcal{L} = \sum_{t = 1}^T \sum_{e = 1}^2 v_{t,g} l_e(x_{e,t}) z_{e,t}^g + \mu \left( \sum_{t = 1}^T \tau_t z_{1,t}^g - B \right) - \sum_{t = 1}^T \lambda_t (z_{1,t}^g + z_{2,t}^g - 1) - \sum_{t = 1}^T \sum_{e = 1}^2 s_{e,t} z_{e,t}^g,
\end{align*}
where $\mu$ is the dual variable corresponding to the budget Constraint~\eqref{eq:con3IOp}, $\lambda_t$ is the dual variable corresponding to the allocation Constraint~\eqref{eq:con2IOp}, and $s_{e,t}$ is the dual variable corresponding to the non-negativity constraint.

We now derive the first order necessary and sufficient optimality conditions of Problem~\eqref{eq:objIOP}-\eqref{eq:con3IOp} by evaluating the derivative of the above defined Lagrangian. For the first edge we obtain that
\begin{align*}
    \frac{\partial \Ll}{\partial z_{1,t}^g} = v_{t,g} l_1(x_{1,t}) + \mu \tau_t - \lambda_t - s_{1,t} = 0.
\end{align*}
From the sign constraint that $s_{1,t} \geq 0$ and the complimentary slackness conditions, the above equation implies that
\begin{align*}
    v_{t,g} l_1(x_{1,t}) + \mu \tau_t \geq \lambda_t, \text{ if } z_{1,t}^g \geq 0, \\
    v_{t,g} l_1(x_{1,t}) + \mu \tau_t = \lambda_t, \text{ if } z_{1,t}^g > 0.
\end{align*}
Similarly, we have for the eligible users for the second edge that
\begin{align*}
    v_{t,g} l_2(x_{2,t}) \geq \lambda_t, \text{ if } z_{2,t}^g \geq 0, \\
    v_{t,g} l_2(x_{2,t}) = \lambda_t, \text{ if } z_{2,t}^g > 0.
\end{align*}
From the above relations, we observe that $\lambda_t = \min \{ v_{t,g} l_1(x_{1,t}) + \mu \tau_t, v_{t,g} l_2(x_{2,t}) \}$ for all periods $t \in [T]$. In particular, eligible users in group $g$ use the express lane at periods $t \in [T]$ when $v_{t,g} l_1(x_{1,t}) + \mu \tau_t \leq v_{t,g} l_2(x_{2,t})$. Next, observe that since the dual variable $\mu$ is independent of $t$, this inequality can be rearranged to obtain that $\mu \leq \frac{v_{t,g}(l_2(x_{2,t}) - l_1(x_{1,t}))}{\tau_t}$ for all periods $t$ when eligible users use the express lane and $\mu > \frac{v_{t,g}(l_2(x_{2,t}) - l_1(x_{1,t}))}{\tau_t}$ for all periods $t$ when eligible users do not use the express lane. In other words, users use the express lane at periods $t$ in the descending order of the travel \emph{bang-per-buck} ratio $\frac{v_{t,g}(l_2(x_{2,t}) - l_1(x_{1,t}))}{\tau_t}$ until their budget is exhausted, which proves our claim.

\subsection{Proof of Lemma~\ref{lem:variationalinequality}} \label{apdx:mainVariationalIneqPf}

We first prove the forward direction of this claim. Suppose $\y^*$ is a CBCP $(\ttau, B)$-equilibrium flow, then we claim that it satisfies the variational inequality problem. To this end, note that at any CBCP $(\ttau, B)$-equilibrium flow it must hold for all ineligible users in group $g \in \G_I$ at each period $t\in [T]$ that
\begin{align*}
    \sum_{e = 1}^2 (v_{t,g} l_e(x_{e,t}^{*}) + \mathbbm{1}_{e = 1} \tau_t) y_{e,t}^g \geq \sum_{e = 1}^2 (v_{t,g} l_e(x_{e,t}^{*}) + \mathbbm{1}_{e = 1} \tau_t) y_{e,t}^{g*}
\end{align*}
for all feasible $\y^g =  (y_{e,t}^g)_{t \in [T], e \in \{1, 2 \}}$ for user group $g \in \G_I$, since ineligible users will choose the edge that minimizes their travel cost at each period. Similarly, at any Nash equilibrium flow $\y^*$, it holds for all eligible users in group $g \in \G_E$ that
\begin{align*}
    \sum_{t = 1}^T \sum_{e = 1}^2 v_{t,g} l_e(x_{e,t}^{*}) y_{e,t}^g \geq \sum_{t = 1}^T \sum_{e = 1}^2 v_{t,g} l_e(x_{e,t}^{*}) y_{et}^{g*},
\end{align*}
for all feasible $\y^g =  (y_{e,t}^g)_{t \in [T], e \in \{1, 2 \}}$ for user group $g \in \G_E$. Then, summing the above two inequalities for all ineligible and eligible users we obtain that
\begin{align*}
    \sum_{g \in \G_E} \sum_{t = 1}^T \sum_{e = 1}^2  v_{t,g} l_e(x_{e,t}^{*}) y_{e,t}^g + \sum_{g \in \G_I} \sum_{t = 1}^T \sum_{e = 1}^2 (v_{t,g} l_e(x_{e,t}^{*}) + \mathbbm{1}_{e = 1} \tau_t) y_{e,t}^g  \geq \\ \sum_{g \in \G_E} \sum_{t = 1}^T \sum_{e = 1}^2  v_{t,g} l_e(x_{e,t}^{*}) y_{et}^{g*} + \sum_{g \in \G_I} \sum_{t = 1}^T \sum_{e = 1}^2 (v_{t,g} l_e(x_{e,t}^{*}) + \mathbbm{1}_{e = 1} \tau_t) y_{e,t}^{g*},
\end{align*}
which can be rearranged to obtain the variational inequality in the statement of the lemma. This proves that any CBCP $(\ttau, B)$-equilibrium flow must satisfy the given variational inequality.

To prove the other direction, fix the strategies of all but one eligible user group $g' \in \G_E$ and consider a feasible flow $\y$ such that $\y^g = \y^{g*}$ for all $g \neq g' \in \G_E$, where $\y^g =  (y_{e,t}^g)_{t \in [T], e \in \{1, 2 \}}$. Then, it holds by the variational inequality that
\begin{align*}
    \sum_{t = 1}^T \sum_{e = 1}^2 v_{t,g'} l_e(x_{e,t}^{*}) (y_{e,t}^{g'} - y_{e,t}^{g'*}) \geq 0,
\end{align*}
which is the equilibrium condition for any user in the eligible group $g' \in \G_E$. Next, fix an ineligible user group $g' \in \G_I$ and some period $t' \in [T]$, and consider a flow $\y$ such that $\y^g = \y^{g*}$ for all $g \neq g' \in \G_I$ and $\y^{g'}_t = \y^{g'*}_t$ for all $t \neq t'$, where $\y^g_t =  (y_{e,t}^g)_{e \in \{1, 2 \}}$. Then, it holds by the variational inequality that
\begin{align*}
    \sum_{e = 1}^2 (v_{t',g} l_e(x_{e,t'}^{*}) + \mathbbm{1}_{e = 1} \tau_{t'}) (y_{e,t'}^{g'} - y_{e,t'}^{g'*}) \geq 0,
\end{align*}
which is the equilibrium condition for any user in the ineligible group $g' \in \G_I$ at any period $t' \in [T]$, as ineligible users minimize their total travel costs at each period. Since both eligible and ineligible users satisfy their corresponding equilibrium conditions, we have established that any flow $\y^*$ satisfying the variational inequality is a CBCP $(\ttau, B)$-equilibrium, which proves our claim.

\subsection{Proof of Lemma~\ref{lem:uniqueness-edge}} \label{apdx:edge-uniquenessPf}

To prove this claim, we proceed by contradiction. In particular, suppose that there are two equilibrium flows corresponding to a CBCP scheme $(\ttau, B)$ with two distinct aggregate edge flows $\x$ and $\Tilde{\x}$. Since the aggregate edge flows are distinct, we suppose without loss of generality that $x_{1,t} > \Tilde{x}_{1,t}$ for some period $t$. We now show that the number of ineligible users on the express lane at period $t$, denoted as $x_{1,t}^{\I}$ and $\Tilde{x}_{1,t}^{\I}$, respectively, must satisfy $x_{1,t}^{\I} \leq \Tilde{x}_{1,t}^{\I}$, and then use this fact to derive our desired contradiction.

To establish that $x_{1,t}^{\I} \leq \Tilde{x}_{1,t}^{\I}$, first note by the equilibrium condition for the ineligible users that a user in group $g \in \G_I$ takes the express lane at period $t$ only when
\begin{align*}
    v_{t,g} l_1(x_{1,t}) + \tau_t \leq v_{t,g} l_2(x_{2,t}).
\end{align*}
Then, noting that $x_{1,t} > \Tilde{x}_{1,t}$ and correspondingly that $x_{2,t} < \Tilde{x}_{2,t}$ as the user demand is fixed, consider an ineligible user in group $g \in \G_I$ that uses the express lane at equilibrium at period $t$ corresponding to the edge flows $\x$. It then holds that
\begin{align*}
    v_{t,g} l_1(\Tilde{x}_{1,t}) + \tau_t &\stackrel{(a)}{<} v_{t,g} l_1(x_{1,t}) + \tau_t, \\
    &\stackrel{(b)}{\leq} v_{t,g} l_2(x_{2,t}), \\
    &\stackrel{(c)}{<} v_{t,g} l_2(\Tilde{x}_{2,t}),
\end{align*}
where (a) follows by the assumption that $x_{1,t} > \Tilde{x}_{1,t}$, (b) holds as the ineligible user group uses the express lane at period $t$ under the edge flows $\x$, and (c) follows as $x_{2,t} < \Tilde{x}_{2,t}$. The above (strict) inequalities imply that if an ineligible user takes the express lane at period $t$ under the edge flow $\x$, then this user also takes the express lane under the edge flow $\Tilde{\x}$. That is, we have shown that $x_{1,t}^{\I} \leq \Tilde{x}_{1,t}^{\I}$.

Next, let $x_{1,t}^{\E}$ and $\Tilde{x}_{1,t}^{\E}$ denote the equilibrium aggregate flow on edge $e$ of eligible users. Then, the above derived relation for ineligible users, together with the fact that $x_{1,t} > \Tilde{x}_{1,t}$, implies that $x_{1,t}^{\E} > \Tilde{x}_{1,t}^{\E}$. We now show that this relation for the eligible users cannot hold true to derive the desired contradiction.

To see this, recall that under a set of tolls $\ttau$ eligible users use the express lane at the periods when the travel bang-per-buck ratio $\frac{v_{t,g}(l_2(x_{2,t}(\ttau)) - l_1(x_{1,t}(\ttau)))}{\tau_t}$ is the highest. Then, since the flow $x_{1,t} > \Tilde{x}_{1,t}$, the ratio $\frac{v_{t,g}(l_2(x_{2,t}) - l_1(x_{1,t}))}{\tau_t} < \frac{v_{t,g}(l_2(\Tilde{x}_{2,t}) - l_1(\Tilde{x}_{1,t}))}{\tau_t}$. Further, the fact that $x_{1,t} > \Tilde{x}_{1,t}$ implies that there are users from some group $g$ such that $y_{1,t}^g > \Tilde{y}_{1,t}^g$, where $\y$, $\Tilde{\y}$ represent the equilibrium flows corresponding to the edge flows $\x, \Tilde{\x}$, respectively. Then, by the budget constraint of users in group $g$ it must hold that there is some period $t'$ at which the flow of users in group $g$ on the express lane is lower under the edge flow $\x$, i.e., $y_{1,t'}^g < \Tilde{y}_{1,t'}^g$. However, for $\y$ to be an equilibrium, it must hold that for all periods $t'$ such that the flow of users in group $g$ on the express lane is lower than that under $\Tilde{\y}$ that the aggregate express lane flow must be higher under $\y$. Note if the aggregate express lane flow at one of the periods $t'$ is lower under $\y$, then users in group $g$ have a profitable deviation at flow $\y$, as they would achieve a lower travel cost when using the express lane at that period as compared to increasing their flow on the express lane at period $t$ (given the change in their travel bang-per-buck ratios).

Next, since the flow of users in group $g$ is lower under the flow $\y$ (as compared to $\Tilde{\y}$) at a subset of periods $t'$ when the aggregate express lane flow is higher under $\y$, there must be other groups $g'$ such that their flow is higher at the periods $t'$. By the above argument for group $g$, we have that the flow $\y$ can only be an equilibrium if for all groups $g''$ their express lane flow is higher than under $\Tilde{\y}$ at one of the periods with a higher aggregate express lane flow under $\y$. However, there are a subset of periods for which the aggregate express lane flow is higher under $\y$ (which holds by our assumption that $x_{1,t} > \Tilde{x}_{1,t}$ at period $t$) and another set of periods for which the total aggregate express lane flow is lower under $\y$ (which holds by eligible users' budget constraints). Thus, it must be that there is some user group such that their flow is lower under $\y$ at a period when the aggregate express lane flow is lower under $\y$ and is higher at another period when the aggregate express lane flow is higher under $\y$ as compared to $\Tilde{\y}$. However, then the flow $\y$ cannot be an equilibrium, which gives us our desired contradiction, proving our claim.

\subsection{Proof of Theorem~\ref{thm:existence-uniqueness-homogeneous}} \label{apdx:convexProgramPf}

To establish the existence of CBCP $(\ttau, B)$-equilibria, we first note that there exists a feasible solution to Problem~\eqref{eq:obj}-\eqref{eq:edgeConstraint}. Next, we can verify that the KKT conditions of Problem~\eqref{eq:obj}-\eqref{eq:edgeConstraint} precisely correspond to the equilibrium conditions for all users, which establishes the existence of a CBCP $(\ttau, B)$-equilibrium.

To prove this claim, we derive the necessary optimality conditions of Problem~\eqref{eq:obj}-\eqref{eq:edgeConstraint} and show that these conditions imply the equilibrium conditions for both the eligible and ineligible users.

To derive the optimality conditions, we first formulate the following Lagrangian of Problem~\eqref{eq:obj}-\eqref{eq:edgeConstraint}:
\begin{align*}
    \Ll = &\sum_{t = 1}^T \left[ \sum_{e = 1}^2 \int_{0}^{x_{e,t}} l_e(\omega) d \omega + \sum_{g \in \G_I} \frac{y_{1,t}^g \tau_t}{v_{t,g}} \right] + \sum_{g \in \G_E} \mu_g \left( \sum_{t = 1}^T \tau_t y_{1,t}^g - B \right) \\
    &- \sum_{g \in \G} \sum_{t \in [T]} \lambda_t^g (y_{1,t}^g + y_{2,t}^g - 1) - \sum_{g \in \G} \sum_{t \in [T]} \sum_{e = 1}^2 s_{e,t}^g y_{e,t}^g,
\end{align*}
where $\sum_{g \in \G} y_{e,t}^g = x_{e,t}, \text{ for all } e \in E , t \in [T]$ as in Constraint~\eqref{eq:edgeConstraint}. We now derive the first order optimality conditions for both the eligible and ineligible user groups.

\paragraph{Ineligible Users:}

For an ineligible user group $g \in \G_I$, we consider the following first order optimality condition by evaluating the derivative of the above defined Lagrangian. In particular, for the first edge we obtain that
\begin{align*}
    \frac{\partial \Ll}{\partial y_{1,t}^g} = l_1(x_{1,t}) + \frac{\tau_t}{v_{t,g}} - \lambda_t^g - s_{1,t}^g = 0.
\end{align*}
From the sign constraint that $s_{1,t}^g \geq 0$ and the complimentary slackness conditions, the above equation implies that
\begin{align*}
    l_1(x_{1,t}) + \frac{\tau_t}{v_{t,g}} \geq \lambda_t^g, \text{ if } y_{1,t}^g \geq 0, \\
    l_1(x_{1,t}) + \frac{\tau_t}{v_{t,g}} = \lambda_t^g, \text{ if } y_{1,t}^g > 0.
\end{align*}
Similarly, we have for the ineligible users for the second edge that
\begin{align*}
    l_2(x_{2,t}) \geq \lambda_t^g, \text{ if } y_{2,t}^g \geq 0, \\
    l_2(x_{2,t}) = \lambda_t^g, \text{ if } y_{2,t}^g > 0.
\end{align*}
In other words, the travel cost for an ineligible user at each period $t$ is given by $v_{t,g} \lambda_t^g = \min \{ v_{t,g}l_1(x_{1,t}) + \tau_t, v_{t,g}l_2(x_{2,t}) \}$, i.e., ineligible users choose the edge corresponding to the minimum travel cost at each period $t$. As a result, the resulting outcome corresponds to the equilibrium conditions for the ineligible users.

\paragraph{Eligible Users:}

We now derive the first order necessary optimality conditions of the eligible users, $g \in \G_E$ by evaluating the derivative of the above defined Lagrangian. For the first edge we obtain that
\begin{align*}
    \frac{\partial \Ll}{\partial y_{1,t}^g} = l_1(x_{1,t}) + \mu_g \tau_t - \lambda_t^g - s_{1,t}^g = 0.
\end{align*}
From the sign constraint that $s_{1,t}^g \geq 0$ and the complimentary slackness conditions, the above equation implies that
\begin{align*}
    l_1(x_{1,t}) + \mu_g \tau_t \geq \lambda_t^g, \text{ if } y_{1,t}^g \geq 0, \\
    l_1(x_{1,t}) + \mu_g \tau_t = \lambda_t^g, \text{ if } y_{1,t}^g > 0.
\end{align*}
Similarly, we have for the eligible users for the second edge that
\begin{align*}
    l_2(x_{2,t}) \geq \lambda_t^g, \text{ if } y_{2,t}^g \geq 0, \\
    l_2(x_{2,t}) = \lambda_t^g, \text{ if } y_{2,t}^g > 0.
\end{align*}
Using the above defined relations, the sum of the travel costs for the eligible users over the $T$ periods under the optimal solution to Problem~\eqref{eq:obj}-\eqref{eq:edgeConstraint} is given by
\begin{align*}
    \sum_{t = 1}^T \lambda_t^g - \mu_g B = \sum_{t = 1}^T \min \{ l_1(x_{1,t}) + \mu_g \tau_t , l_2(x_{2,t}) \} - \mu_g B.
\end{align*}
In the case when $\mu_g = 0$, we have that the travel cost for eligible users in group $g$ satisfies
\begin{align*}
    \sum_{t = 1}^T \min \{ l_1(x_{1,t}) , l_2(x_{2,t}) \} \leq \sum_{t = 1}^T \sum_{e \in E} l_e(x_{e,t}) \Tilde{y}_{e,t}^g
\end{align*}
for any feasible $\Tilde{\y}^g$ for $g \in \G_E$ as $\Tilde{y}_{1,t}^g + \Tilde{y}_{2,t}^g = 1$ for all periods $t \in [T]$. As a result, in the case when $\mu_g = 0$, the optimal solution to Problem~\eqref{eq:obj}-\eqref{eq:edgeConstraint} corresponds to the equilibrium conditions for the eligible users.


Next, consider the case when $\mu_g>0$. In this case, by the complimentary slackness conditions it holds that $\sum_{t = 1}^T \tau_t y_{1,t}^g = B$. Next, suppose without loss of generality that the edge flows are ordered such that $x_{1,1} \leq x_{1,2} \leq \ldots \leq x_{1,T}$. Then, since $\mu_g$ is independent of $t$ it must hold that there is some $T_1$ such that $\mu_g \tau_t + l_1(x_{1,t}) \geq l_2(x_{2,t})$ for $t > T_1$ and that $\mu_g \tau_t + l_1(x_{1,t}) \leq l_2(x_{2,t})$ for $t \leq T_1$ by the monotonicity of the $x_{1,t}$'s and the fact that the total demand $x_{1,t}+x_{2,t}$ is a fixed quantity across the $T$ periods. Here $T_1$ is such that $\sum_{t = 1}^{T_1} \tau_t y_{1,t}^g = B$. That is, a user in group $g$ is routed on edge one at those periods when the ratio of the difference in the travel time to toll ratio on that edge is the lowest, which aligns with the characterization of the optimal solution of eligible users' individual optimization problem in Lemma~\ref{lem:optEligible} for the setting when eligible users' values of time are time-invariant. As a result, it follows that the total cost of the eligible users is given by
\begin{align*}
    \sum_{t = 1}^{T_1} l_1(x_{1,t}) + \sum_{t = T_1}^T l_2(x_{2,t}) \leq \sum_{t = 1}^T \sum_{e = 1}^2 l_e(x_{e,t}) \Tilde{y}_{e,t}^g,
\end{align*}
which holds for any feasible $\Tilde{\y}^g$ for $g \in \G_E$. Thus, the optimal solution of Problem~\eqref{eq:obj}-\eqref{eq:edgeConstraint} corresponds to the equilibrium conditions for the eligible users for which $\mu_g > 0$, which establishes our claim.



\subsection{Proof of Lemma~\ref{lem:TollMonotonicity}} \label{apdx:PfmonotonicityTolls}

We prove this claim by contradiction. Suppose for contradiction that $x_{1,t}(\Tilde{\ttau}) > x_{1,t}(\ttau)$, where $x_{e,t}(\ttau)$ denotes the equilibrium aggregate flow on edge $e$ at period $t$ under the CBCP scheme with toll $\ttau$. We use this relation to first show that the number of ineligible users on the express lane at period $t$, denoted as $x_{1,t}(\Tilde{\ttau})_{\I}$ and $x_{1,t}(\ttau)_{\I}$ for the two tolls must satisfy $x_{1,t}(\Tilde{\ttau})_{\I} \leq x_{1,t}(\ttau)_{\I}$, and then use this fact to derive our desired contradiction.

To establish that $x_{1,t}(\Tilde{\ttau})_{\I} \leq x_{1,t}(\ttau)_{\I}$, first note by the equilibrium condition for the ineligible users that a user in group $g \in \G_I$ takes the express lane at period $t$ only when
\begin{align*}
    v_{t,g} l_1(x_{1,t}(\Tilde{\ttau})) + \Tilde{\tau}_t \leq v_{t,g} l_2(x_{2,t}(\Tilde{\ttau})).
\end{align*}
Then, noting that $x_{1,t}(\Tilde{\ttau}) > x_{1,t}(\ttau)$ and correspondingly that $x_{2,t}(\Tilde{\ttau}) < x_{2,t}(\ttau)$ as the user demand is fixed, consider an ineligible user group $g$ that uses the express lane at equilibrium at period $t$ under the toll $\Tilde{\ttau}$. It then holds that
\begin{align*}
    v_{t,g} l_1(x_{1,t}(\ttau)) + \tau_t &\stackrel{(a)}{<} v_{t,g} l_1(x_{1,t}(\Tilde{\ttau})) + \Tilde{\tau}_t, \\
    &\stackrel{(b)}{\leq} v_{t,g} l_2(x_{2,t}(\Tilde{\ttau})), \\
    &\stackrel{(c)}{<} v_{t,g} l_2(x_{2,t}(\ttau)),
\end{align*}
where (a) follows as $\Tilde{\tau}_t>\tau_t$ and by the assumption that $x_{1,t}(\Tilde{\ttau}) > x_{1,t}(\ttau)$, (b) holds as the ineligible user group uses the express lane at period $t$ under the toll $\Tilde{\ttau}$, and (c) follows as $x_{2,t}(\Tilde{\ttau}) < x_{2,t}(\ttau)$. The above (strict) inequalities imply that if an ineligible user takes the express lane at period $t$ under toll $\Tilde{\ttau}$, then this user also takes the express lane under the toll $\ttau$. That is, $x_{1,t}(\Tilde{\ttau})_{\I} \leq x_{1,t}(\ttau)_{\I}$.

Next, let $x_{1,t}(\ttau)_{\E}$ denote the equilibrium aggregate flow on edge $e$ of eligible users under the CBCP scheme with toll $\ttau$. Then, the above derived relation for ineligible users, together with the fact that $x_{1,t}(\Tilde{\ttau}) > x_{1,t}(\ttau)$, implies that $x_{1,t}(\Tilde{\ttau})_{\E} > x_{1,t}(\ttau)_{\E}$. We now show that this relation for the eligible users cannot hold true to derive the desired contradiction.


To see this, recall that under a set of tolls $\ttau$ eligible users use the express lane at the periods when the travel bang-per-buck ratio $\frac{v_{t,g}(l_2(x_{2,t}(\ttau)) - l_1(x_{1,t}(\ttau)))}{\tau_t}$ is the highest. Then, since the express lane flows satisfy $x_{1,t}(\Tilde{\ttau}) > x_{1,t}(\ttau)$, the new ratio $\frac{v_{t,g}(l_2(x_{2,t}(\ttau)) - l_1(x_{1,t}(\ttau)))}{\tau_t} > \frac{v_{t,g}(l_2(x_{2,t}(\Tilde{\ttau})) - l_1(x_{1,t}(\Tilde{\ttau})))}{\Tilde{\tau}_t}$. Further, the fact that $x_{1,t}(\Tilde{\ttau}) > x_{1,t}(\ttau)$ implies that there are users from some group $g$ such that $y_{1,t}^g(\Tilde{\ttau}) > y_{1,t}^g(\ttau)$, where $\y(\ttau)$, $\y(\Tilde{\ttau})$ represent the equilibrium flows corresponding to the tolls $\ttau, \Tilde{\ttau}$, respectively. Then, by the budget constraint of users in group $g$ it must hold that there is some period $t'$ at which the flow of users in group $g$ on the express lane reduces under the toll $\Tilde{\ttau}$, i.e., $y_{1,t'}^g(\Tilde{\ttau}) < y_{1,t'}^g(\ttau)$. However, for $\y(\Tilde{\ttau})$ to be an equilibrium, it must hold that for all periods $t'$ such that the flow of users in group $g$ on the express lane reduces that the aggregate express lane flow must increase. Note if the aggregate express lane flow at one of the periods $t'$ decreases under the toll $\Tilde{\ttau}$, then users in group $g$ have a profitable deviation at the new flow $\y(\Tilde{\ttau})$, as they would achieve a lower travel cost when using the express lane at that period as compared to increasing their flow on the express lane at period $t$ (given the change in their travel bang-per-buck ratios).

Next, since the flow of users in group $g$ decrease at a subset of periods $t'$ when the aggregate express lane flow is higher under the toll $\Tilde{\ttau}$ (relative to that under the toll $\ttau$), there are other groups $g'$ such that their flow is higher at the periods $t'$. By the above argument for group $g$, the flow $\y(\Tilde{\ttau})$ can only be an equilibrium if for all groups $g''$ their express lane flow increases at one of the periods with a higher aggregate express lane flow. However, there are a subset of periods for which the total aggregate express lane flow increases under the toll $\Tilde{\ttau}$ (which holds by our assumption that $x_{1,t}(\Tilde{\ttau}) > x_{1,t}(\ttau)$ at period $t$) and another set of periods for which the total aggregate express lane flow decreases (which holds by eligible users' budget constraints). Thus, it must be that under the toll $\Tilde{\ttau}$ there is some user group such that their flow decreases at a period when the aggregate express lane flow decreases and increases at a period when the aggregate express lane flow increases. However, then the flow $\y(\Tilde{\ttau})$ cannot be an equilibrium, which gives us our desired contradiction, proving our claim.

\subsection{Proof of Proposition~\ref{prop:SubstitutesViolation}} \label{apdx:pfSubstitutesViolation}

Let the number of periods $T = 2$, where the toll on the express lane is equal in both periods, i.e., $\tau = \tau_1 = \tau_2$, and consider a setting with one eligible group (with total mass of one unit) with budget $B = 3 \frac{\tau}{2}$ and one ineligible user group (with total mass of one unit). Further, let the eligible users have a time-invariant value of time, the travel time on the edges be such that $l_2(0.5)>l_1(1.5)$, and suppose that the ineligible users are such that they have a high value of time in the first period and a low value of time in the second period such that the ineligible users always use the express lane in period one and do not use the express lane in period two. Then, the equilibrium is such that only the eligible user group uses the express lane at period two and 0.5 units of the eligible user group (and the ineligible user group) use the express lane at period one. In particular, the following relation holds (with strict inequality):
\begin{align*}
    \frac{l_2(1) - l_1(1)}{\tau} > \frac{l_2(0.5) - l_1(1.5)}{\tau}
\end{align*}
From the above relation, it follows by continuity that the following relation must hold for some $\epsilon > 0$
\begin{align*}
    \frac{l_2(1) - l_1(1)}{\tau(1+\epsilon)} > \frac{l_2(0.5+\epsilon) - l_1(1.5-\epsilon)}{\tau}.
\end{align*}
That is, the above relation implies that if the toll at period two is increased to $\tau(1+\epsilon)$, the flow of the eligible users on the express lane reduces in the first period, which proves our claim.

\subsection{Proof of Lemma~\ref{lem:BudgetMonotonicity}} \label{apdx:pfBudgetMonotonicity}

To prove this claim, we proceed by contradiction as in the proof of Lemma~\ref{lem:TollMonotonicity}. In particular, we suppose that there is some period $t$ such that $x_{1,t}(B) > x_{1,t}(\Tilde{B})$, where $x_{e,t}(B)$ denotes the equilibrium aggregate flow on edge $e$ at period $t$ under the CBCP scheme where eligible users receive a budget $B$. Since the tolls for the two CBCP schemes are the same, we can show using a similar line of reasoning for ineligible users to that in the proof of Lemma~\ref{lem:TollMonotonicity} that when $x_{1,t}(B) > x_{1,t}(\Tilde{B})$, it must hold that $x_{1,t}(B)_{\I} \leq x_{1,t}(\Tilde{B})_{\I}$, where $x_{e,t}(B)_{\I}$ denotes the equilibrium aggregate flow of ineligible users on edge $e$ at period $t$ under the CBCP scheme where eligible users receive a budget $B$. As a result for $x_{1,t}(B) > x_{1,t}(\Tilde{B})$ to hold, it must follow that $x_{1,t}(B)_{\E} > x_{1,t}(\Tilde{B})_{\E}$, where $x_{e,t}(B)_{\E}$ denotes the equilibrium aggregate flow of eligible users on edge $e$ at period $t$ under the CBCP scheme where eligible users receive a budget $B$. We now show that the relation that $x_{1,t}(B)_{\E} > x_{1,t}(\Tilde{B})_{\E}$ cannot hold true to derive our desired contradiction.

To see this, recall that, given a budget $B$ and express lane tolls $\ttau$, eligible users use the express lane at the periods when the travel bang-per-buck ratio $\frac{v_{t,g}(l_2(x_{2,t}(B)) - l_1(x_{1,t}(B)))}{\tau_t}$ is the highest. Then, since the express lane flows satisfy $x_{1,t}(B) > x_{1,t}(\Tilde{B})$, the new ratio $\frac{v_{t,g}(l_2(x_{2,t}(B)) - l_1(x_{1,t}(B)))}{\tau_t} > \frac{v_{t,g}(l_2(x_{2,t}(\Tilde{B})) - l_1(x_{1,t}(\Tilde{B})))}{\tau_t}$. Further, the fact that $x_{1,t}(B) > x_{1,t}(\Tilde{B})$ implies that there are users from some group $g$ such that $y_{1,t}^g(B) > y_{1,t}^g(\Tilde{B})$, where $\y(B)$, $\y(\Tilde{B})$ represent the equilibrium flows corresponding to the budgets $B, \Tilde{B}$, respectively. Then, by the budget constraint of users in group $g$ it must hold that there is some period $t'$ at which the flow of users in group $g$ on the express lane is lower under the budget $B$, i.e., $y_{1,t'}^g(\Tilde{B}) > y_{1,t'}^g(B)$. However, for $\y(B)$ to be an equilibrium, it must hold that for all periods $t'$ such that the flow of users in group $g$ on the express lane reduces that the aggregate express lane flow must increase. Note if the aggregate express lane flow at one of the periods $t'$ decreases under the budget $B$ (relative to that under the budget $\Tilde{B}$), then users in group $g$ have a profitable deviation at the flow $\y(B)$, as they would achieve a lower travel cost when using the express lane at that period as compared to increasing their flow on the express lane at period $t$ (given the change in their travel bang-per-buck ratios).

Next, since the flow of users in group $g$ decrease at a subset of periods $t'$ when the aggregate express lane flow is higher under the budget $B$ (relative to that under the budget $\Tilde{B}$), there are other groups $g'$ such that their flow is higher at the periods $t'$. By the above argument for group $g$, the flow $\y(B)$ can only be an equilibrium if for all groups $g''$ their express lane flow increases at one of the periods with a higher aggregate express lane flow. However, there are a subset of periods for which the total aggregate express lane flow increases under the toll $B$ (which holds by our assumption that $x_{1,t}(B) > x_{1,t}(\Tilde{B})$ at period $t$) and another set of periods for which the aggregate express lane flow decreases (which holds by eligible users' budget constraints). Thus, it must be that under the budget $B$ there is some user group such that their flow decreases at a period when the aggregate express lane flow decreases and increases at a period when the aggregate express lane flow increases. However, then the flow $\y(B)$ cannot be an equilibrium, which gives us our desired contradiction, proving our claim.

\subsection{Proof of Proposition~\ref{prop:nonMonotonicityBudget}} \label{apdx:pfBudgetMonViolation}

Let the number of periods $T = 1$ and suppose that all users are eligible with a total demand of one unit. Further, suppose that the toll on the express lane is $\tau$ and the budget given to eligible users is $B = \frac{\tau}{2}$. In addition, consider a travel time function where $l_1(0.5) = 0.5$, $l_1(0.75) = 1.5$, $l_2(0.5) = 2$, and $l_2(0.25) = 1.99$. Then, since fractional flows are allowed, the resulting equilibrium flow pattern will be such that $0.5$ units of users use the express lane and $0.5$ units of users use the general-purpose lane, resulting in a total travel cost of $1.25$ (as $0.5 \times 0.5 + 2 \times 0.5 = 1.25$) for each user. On the other hand, if the budget were increased to $\Tilde{B} = \frac{3 \tau}{4}$, then the resulting equilibrium flow pattern will be such that $0.75$ units of users use the express lane and $0.25$ units of users use the general-purpose lane, resulting in a total travel cost of $1.6225$ (as $1.5 \times 0.75 + 1.99 \times 0.25 = 1.6225$) for each user. These relations establish that an increase in the budget of eligible users can in fact increase their travel costs even if their values of time are time-invariant, which proves our claim.

\subsection{Proof of Lemma~\ref{lem:contRelations}} \label{apdx:contRelations}

To prove this claim, we leverage Berge's theorem, which establishes continuity relations of the optimal solutions of parametric optimization problems in their parameters.

\begin{theorem} [Berge's Theorem~\cite{kreps-book}] \label{thm:berge's}
Consider the parametric constrained maximization problem
\[ \max F(z, \theta), \text{ s.t. } z \in A(\theta).  \]
Let $Z(\theta)$ be the set of solutions of this problem with parameter $\theta$. If
\begin{enumerate} [label=(\alph*)]
    \item $F$ is a continuous function in $(z, \theta)$, 
    \item $\theta \implies A(\theta)$ is lower semi-continuous, and
    \item there exists for each $\theta$ a set $B(\theta) \subseteq A(\theta)$ such that $Z(\theta) \subseteq B(\theta)$, $\sup \{F(z, \theta): z \in B(\theta) \} = \sup \{F(z, \theta): z \in A(\theta) \}$ , and $\theta \implies B(\theta)$ is an upper semi-continuous and locally bounded correspondence.
\end{enumerate}
Then:
\begin{enumerate}
    \item $Z(\theta)$ is non-empty for all $\theta$, and $\theta \implies Z(\theta)$ is an upper semi-continuous and locally bounded correspondence.
    \item If $Z(\theta)$ is a singleton set, i.e., $Z(\theta) = \{z(\theta) \}$ for all $\theta$ in an open set of parameter values, then $z(\theta)$ is a continuous function over that set of parameter values.
\end{enumerate}
\end{theorem}

We now show that the convex Program~\eqref{eq:obj}-\eqref{eq:budget} satisfies conditions (a), (b), and (c) of Berge's theorem to derive the desired continuity relations. To this end, first note that condition (a) of Theorem~\ref{thm:berge's} clearly holds as the Objective~\eqref{eq:obj} is continuous in $(\y, (\ttau, B))$. In particular, the Objective~\eqref{eq:obj} is independent of the budget $B$, depends linearly on the toll $\ttau$, and is continuous in the edge flows by the continuity of the travel time function. Next, let $\F(\ttau, B)$ denote the set of all feasible flows satisfying Constraints~\eqref{eq:allocation}-\eqref{eq:budget}. We now establish that conditions (b) and (c) of Theorem~\ref{thm:berge's} hold by establishing that the correspondence $(\ttau, B) \implies \F(\ttau, B)$ is (i) locally bounded, (ii) upper semi-continuous, and (iii) lower semi-continuous.

\paragraph{Local Boundedness of $\F(\ttau, B)$:}

By the definition of a locally bounded correspondence (Definition~\ref{def:locallyBdd}), we show for every $(\ttau, B) \in \F_U$ that there is an $\epsilon(\ttau, B)>0$ and a bounded set $\Y(\Bar{\ttau}, \Bar{B})$ is a subset of this bounded set for all $(\Bar{\ttau}, \Bar{B})$ that lie within an $\epsilon$ ball of $(\ttau, B)$. To this end, fix $(\ttau, B)$ and let $\epsilon = \frac{1}{3} \min_{t = 1, \ldots, T} \tau_t$. Next, observe that for all $(\Bar{\ttau}, \Bar{B})$ that lie within an $\epsilon$ ball of $(\ttau, B)$ that any feasible flow $\y \in \F(\Bar{\ttau}, \Bar{B})$ must satisfy $\Bar{\tau}_1^t y_{1,t}^g \leq \Bar{B}$ by the budget Constraint~\eqref{eq:budget}. This relation implies that $y_{1,t}^g \leq \frac{B+\epsilon}{\Bar{\tau}_1^t }$ from which it follows by the definition of $\epsilon$ that $y_{1,t}^g \leq \frac{B+\epsilon}{2 \epsilon}$, which provides a uniform bound on $\y \in \F(\Bar{\ttau}, \Bar{B})$. This establishes that the correspondence $\F(\ttau, B)$ is locally bounded.

\paragraph{Upper Semi-Continuity of $\F(\ttau, B)$:}

To show upper semi-continuity of $\F(\ttau, B)$, we consider a sequence $(\{ \y^n, \ttau^n, B^n \})$ where $\y^n \in \F(\ttau^n, B^n)$ for all $n$ with a limit $(\y, \ttau, B)$. As $\y^n \geq \0$ for all $n$ and since the positive orthant is closed, it follows that $\y \geq 0$. Next, since $\sum_{t = 1}^T \tau_t^n y_{1,t}^{gn} \leq B^n$ for all $n$ for all user groups $g$ it follows by the continuity of the dot product that $\sum_{t = 1}^T \tau_{t} y_{1,t}^{g} \leq B$. Finally, since $y_{1,t}^{gn} + y_{2,t}^{gn} = 1$ for all $n$ it also follows that $y_{1,t}^g+y_{2,t}^g = 1$ for all eligible user groups $g \in \G_{E}$ and periods $t$. Hence, we have established that $\y \in \F(\ttau, B)$, which proves upper semi-continuity.

\paragraph{Lower Semi-Continuity of $\F(\ttau, B)$:}

To establish lower semi-continuity of $\F(\ttau, B)$, we show that for each sequence $(\ttau^n, B^n)$ with limit $(\ttau, B)$ and a point $\y \in \F(\ttau, B)$ that there is a sequence $\y^n$ with limit $\y$ such that $\y^n \in \F(\ttau^n, B^n)$ for all $n$. To show this, first consider the case when $B = 0$. In this setting, by the positivity of the tolls it follows that $y_{1,t}^g = 0$ and $y_{2,t}^g = 1$ for all eligible user groups $g$. Thus, the sequence where $y_{1,t}^{gn} = 0$ and $y_{2,t}^{gn} = 1$ holds for all eligible user groups is feasible for all $n$. Analogously, if the flow $\y_1^g = \0$ (where $\y_1 = ( y_{1,t}^g: t \in [T])$), i.e., there is no flow of eligible users in group $g$ on the express lane, then setting $\y_1^{gn} = \0$ results in a feasible sequence. Thus, consider the setting when for some user group $g$ $\y_1^g \neq \0$, the budget $B>0$, and consider $n$ large enough such that $\sum_{t = 1}^T \tau_t^n y_{1,t}^g$ is uniformly bounded away from zero. Next, let $y_{1,t}^{gn} = \min \left\{ \frac{B^n}{B} \frac{\sum_{t = 1}^T \tau_{t} y_{1,t}^g}{\sum_{t = 1}^T \tau_t^n y_{1,t}^g} y_{1,t}^g, 1 \right\}$ and $y_{2,t}^{gn} = 1 - y_{1,t}^{gn}$. Then, it is clear that $y_{1,t}^{gn} \geq 0$ and $y_{1,t}^{gn} + y_{2,t}^{gn} = 1$. Furthermore, the budget constraint is also satisfied as
\begin{align*}
    \sum_{t = 1}^T \tau_t^n y_{1,t}^{gn} &= \sum_{t = 1}^T \tau_t^n \min \left\{ \frac{B^n}{B} \frac{\sum_{t = 1}^T \tau_{t} y_{1,t}^g}{\sum_{t = 1}^T \tau_t^n y_{1,t}^g} y_{1,t}^g, 1 \right\}, \\
    &\leq \min \left\{ \frac{B^n}{B} \frac{\sum_{t = 1}^T \tau_{t} y_{1,t}^g}{\sum_{t = 1}^T \tau_t^n y_{1,t}^g} \sum_{t = 1}^T \tau_t^n y_{1,t}^g, \sum_{t = 1}^T \tau_t^n \right\}, \\
    &\leq \min \left\{ B^n \frac{\sum_{t = 1}^T \tau_{t} y_{1,t}^g}{B}, \sum_{t = 1}^T \tau_t^n \right\}, \\
    &\leq \min \left\{ B^n, \sum_{t = 1}^T \tau_t^n \right\}, \\
    &\leq B^n.
\end{align*}
Thus, for all eligible user groups (depending on whether $\y_1^g = \0$ or $\y_1^g \neq \0$) we can construct a sequence of $y_{e,t}^{gn}$ such that $\y^n \in \F(\ttau^n, B^n)$ for all $n$, which establishes lower semi-continuity.

Having established that $(\ttau, B) \implies \F(\ttau, B)$ is (i) locally bounded, (ii) upper semi-continuous, and (iii) lower semi-continuous, the upper semi-continuity and local boundedness of the correspondence $\Y(\ttau, B)$ follows from a direct application of Berge's theorem (Theorem~\ref{thm:berge's}). Furthermore, noting from Lemma~\ref{lem:uniqueness-edge} that the aggregate edge flows are unique, it follows from Berge's theorem that the aggregate edge flows $\x(\ttau, B)$ are continuous in $(\ttau, B)$, which proves our claim.


\section{Additional Experimental Results}
In this section, we present additional detailed metrics and graphics corresponding to the experimental results. 

\subsection{PeMS Data} \label{apdx:pemsDataTable}

\begin{table}[!ht]
\centering
\caption{{\small \sf US 101 Express Lanes PeMS Utilization Data (September, 2022)}} \label{tab:US101UtilizationData}
\begin{tabular}{|ll|llll|}
\hline
                        & &\multicolumn{3}{|c}{Station ID }   &         \\
                        & &\multicolumn{3}{|c}{(Postmile location)}   & \\
                        & & 400388 & 405673 & 404532 &  \\
                        & Hour & (405.4)& (406.9) &(408.4) & Mean \\
                        \hline
\% Flow on express      & 7 - 8 AM        & 14\%                                                & 14\%                                                & 15\%                                                & 14\%    \\
                        & 8 - 9 AM        & 16\%                                                & 15\%                                                & 15\%                                                & 15\%    \\
                        & 9 - 10 AM        & 13\%                                                & 12\%                                                & 13\%                                                & 13\%    \\
                        & 7 - 10 AM & 14\%                                                & 14\%                                                & 14\%                                                & 14\%    \\
                        \hline
Mean mph express lanes   & 7 - 8 AM        & 74                                                  & 70                                                  & 67                                                  & 70      \\
                        & 8 - 9 AM        & 66                                                  & 63                                                  & 65                                                  & 65      \\
                        & 9 - 10 AM        & 67                                                  & 69                                                  & 65                                                  & 67      \\
                        & 7 - 10 AM & 69                                                  & 67                                                  & 66                                                  & 67      \\
                        \hline
Mean mph GP lanes        & 7 - 8 AM        & 51                                                  & 44                                                  & 51                                                  & 49      \\
                        & 8 - 9 AM        & 40                                                  & 36                                                  & 48                                                  & 41      \\
                        & 9 - 10 AM        & 43                                                  & 44                                                  & 50                                                  & 46      \\
                        & 7 - 10 AM & 45                                                  & 42                                                  & 50                                                  & 45      \\
                        \hline
Mean travel time savings & 7 - 8 AM        & 32\%                                                & 37\%                                                & 24\%                                                & 31\%    \\
                        & 8 - 9 AM        & 40\%                                                & 43\%                                                & 26\%                                                & 36\%    \\
                        & 9 - 10 AM        & 36\%                                                & 35\%                                                & 23\%                                                & 32\%    \\
                        & 7 - 10 AM & 36\%                                                & 38\%                                                & 24\%                                                & 33\%   \\ \hline
\end{tabular}
\end{table}

\subsection{Experimental results} \label{apdx:additionalExpResults}


Table~\ref{tab:ParetoResultsApdx} includes the optimal CBCP parameter (toll and budget levels) as well as the equilibrium percentage of users on the express lane by group, the average travel times, average travel costs by group, and the total toll revenue for 18 different Pareto weighted schemes for the objective function in Section~\ref{subsec:ObjectiveNumerical}. Furthermore, Figure~\ref{fig:OptGraphs} displays the distributions of the optimal a) CBCP parameters, b) express lane usage, c) average travel times, and d) average social costs and total toll revenue with respect to the value of the Pareto weight for eligible users, holding the remaining two weights ($\lambda_I, \lambda_R$) equal to 1. 

\begin{table}[!ht]
\centering
\scriptsize
    \centering
    \caption{{\small \sf Optimal CBCP and Equilibrium Edge Flows and Travel Times for Various Pareto Weighted Schemes}} \label{tab:ParetoResultsApdx}
    \begin{tabular}{|lll|ll|lll|ll|ll|l|}
    \hline
        \multicolumn{3}{|c|}{Weights} & \multicolumn{2}{|c|}{Optimal} & \multicolumn{3}{|c|}{Express lane usage} & \multicolumn{2}{|c|}{Average TTs}&\multicolumn{2}{|c|}{Total travel costs (\$)} & Total toll\\ 
        \multicolumn{3}{|c|}{} & \multicolumn{2}{|c|}{CBCP (\$)} & \multicolumn{3}{|c|}{(\%)} & \multicolumn{2}{|c|}{(minutes)}&\multicolumn{2}{|c|}{(\% diff from (1,1,1))} & revenue (\$) \\ 

        $\lambda_e$ & $\lambda_i$ & $\lambda_t$ &$\tau$ & B &Overall & 
        Eligible&Ineligible&Express&GP&Eligible&Ineligible&(\% diff from (1,1,1))\\\hline
1         & 0         & 0         & 19   & 90     & 19            & 95           & 3            & 22.1           & 30.3          & 17738 (-27\%)    & 864766 (1\%)    & 21337 (-76\%) \\
0         & 1         & 0         & 0    & 0      & 25            & 60           & 18           & 28.2           & 28.2          & 22210 (-9\%)     & 808169 (-6\%)   & 0 (-100\%)    \\
0         & 0         & 1         & 15   & 0      & 16            & 0            & 19           & 19.4           & 31.4          & 24664 (2\%)      & 868444 (1\%)    & 94098 (6\%)   \\
1         & 1         & 1         & 13   & 0      & 17            & 0            & 21           & 20.3           & 30.9          & 24282 (0\%)      & 856500 (0\%)    & 89097 (0\%)   \\
2         & 1         & 1         & 12   & 0      & 18            & 0            & 21           & 20.9  & 30.7 & 24133 (-1\%)     & 852576 (0\%)    & 84963 (-5\%)  \\
3         & 1         & 1         & 12   & 0      & 18            & 0            & 21           & 20.9           & 30.7          & 24133 (-1\%)     & 852576 (0\%)    & 84963 (-5\%)  \\
4         & 1         & 1         & 12   & 0      & 18            & 0            & 21           & 20.9           & 30.7          & 24133 (-1\%)     & 852576 (0\%)    & 84963 (-5\%)  \\
5         & 1         & 1         & 11   & 0      & 18            & 0            & 22           & 21.5  & 30.5 & 23983 (-1\%)     & 848654 (-1\%)   & 80390 (-10\%) \\
6         & 1         & 1         & 11   & 0      & 18            & 0            & 22           & 21.5           & 30.5          & 23983 (-1\%)     & 848654 (-1\%)   & 80390 (-10\%) \\
7         & 1         & 1         & 11   & 0      & 18            & 0            & 22           & 21.5           & 30.5          & 23983 (-1\%)     & 848654 (-1\%)   & 80390 (-10\%) \\
8         & 1         & 1         & 10   & 0      & 19            & 0            & 23           & 22.1  & 30.3 & 23831 (-2\%)     & 844763 (-1\%)   & 75391 (-15\%) \\
9         & 1         & 1         & 10   & 0      & 19            & 0            & 23           & 22.1           & 30.3          & 23831 (-2\%)     & 844763 (-1\%)   & 75391 (-15\%) \\
10        & 1         & 1         & 10   & 0      & 19            & 0            & 23           & 22.1           & 30.3          & 23831 (-2\%)     & 844763 (-1\%)   & 75391 (-15\%) \\
11        & 1         & 1         & 10   & 15     & 19            & 30           & 17           & 22.3           & 30.2          & 21900 (-10\%)    & 846033 (-1\%)   & 56060 (-37\%) \\
12        & 1         & 1         & 11   & 45     & 19            & 82           & 7            & 22.6           & 30.1          & 18850 (-22\%)    & 848827 (-1\%)   & 24101 (-73\%) \\
13        & 1         & 1         & 12   & 50     & 19            & 83           & 6            & 22.4           & 30.2          & 18611 (-23\%)    & 851683 (-1\%)   & 23890 (-73\%) \\
14        & 1         & 1         & 13   & 55     & 19            & 85           & 6            & 22.2           & 30.3          & 18424 (-24\%)    & 854178 (0\%)    & 23855 (-73\%) \\
15        & 1         & 1         & 13   & 55     & 19            & 85           & 6            & 22.2           & 30.3          & 18424 (-24\%)    & 854178 (0\%)    & 23855 (-73\%)\\ \hline
\end{tabular}
\end{table}

\begin{figure}
\centering
\begin{subfigure}{.45\textwidth}
\centering
    \includegraphics[width=\textwidth]{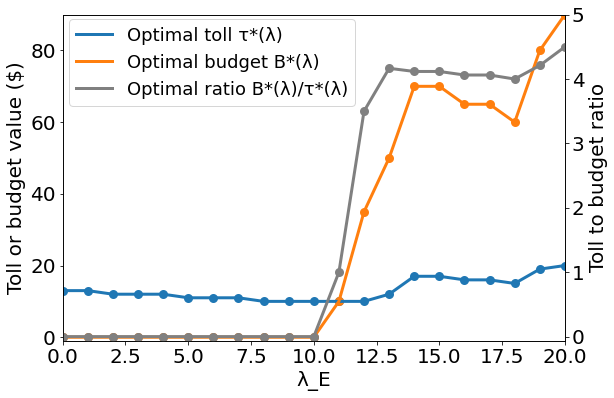}
    \caption{Toll and budget}
    \label{fig:optTollBudget}
\end{subfigure} 
    \begin{subfigure}{.45\textwidth}
    \centering
    \includegraphics[width=\textwidth]{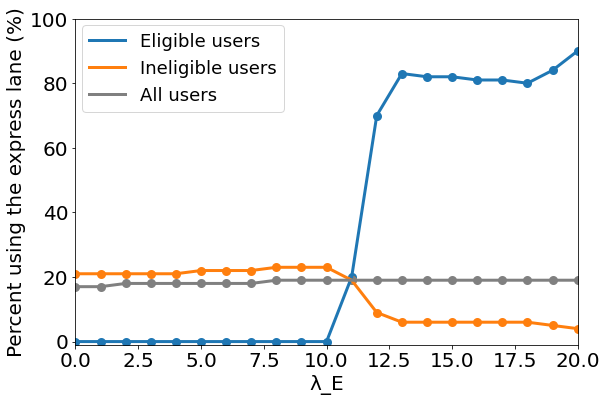}
    \caption{Share of users using express lane}
    \label{fig:OptExressShares}
\end{subfigure}
\begin{subfigure}{.45\textwidth}
    \centering
    \includegraphics[width=\textwidth]{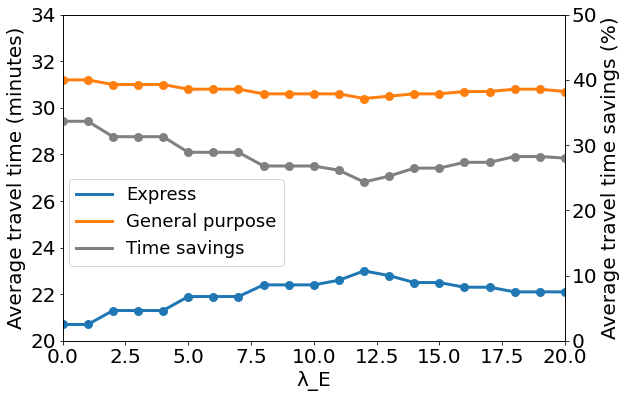}
    \caption{Average travel times}
    \label{fig:tOptTravelTimes}
\end{subfigure} 
\begin{subfigure}{.45\textwidth}
    \centering
    \includegraphics[width=\textwidth]{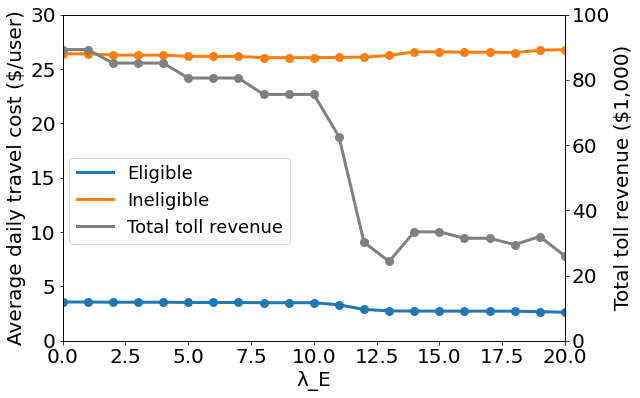}
    \caption{Average social costs and total toll revenue}
    \label{fig:OptCosts}
\end{subfigure}
    \caption{{\small \sf Distributions of the optimal a) CBCP parameters, b) express lane usage, c) average travel times, and d) average social costs and total toll revenue with respect to the value of the Pareto weight for eligible users, holding the remaining two weights ($\lambda_I, \lambda_R$) equal to 1.}}
    \label{fig:OptGraphs}
\end{figure}

\end{document}